\documentclass[12pt]{amsart}

\usepackage{amsmath,amsthm,amsfonts,amssymb,eucal, bbm}
\usepackage{scalerel}
\usepackage{mathabx}
\usepackage{hyperref}

\usepackage{color}

\usepackage{xcolor}

\renewcommand {\a}{ \alpha }
\renewcommand{\b}{\beta}
\newcommand{\e}{\epsilon}

\newcommand{\g}{\gamma}
\newcommand{\G}{\Gamma}

\newcommand{\vark}{\varkappa}
\renewcommand{\d}{\delta}

\renewcommand{\l}{\lambda}
\renewcommand{\L}{\Lambda}
\newcommand{\z}{\zeta}
\renewcommand{\t}{\theta}

\newcommand{\p}{\partial}
\newcommand{\om}{\omega}
\newcommand{\Om}{\Omega}

\newcommand{\oq}{\ {\raise 7pt\hbox{${\scriptstyle\circ}$}}
	\kern -7pt{
		\hbox{$Q$}}}

\newcommand{\R}{ \mathbb R}

\newcommand {\ba}{\mathbf a}

\newcommand {\BA}{\mathbf A}

\newcommand {\BS}{\mathbf S}

\newcommand {\bx}{\mathbf x}

\newcommand {\bk}{\mathbf k}

\newcommand {\bq}{\mathbf q}
\newcommand {\bm}{\mathbf m}
\newcommand {\bl}{\mathbf l}

\newcommand {\by}{\mathbf y}

\newcommand {\bn}{\mathbf n}

\newcommand{\SG}{{\sf{\Gamma}}}

\newcommand{\SfG}{{\sf{G}}}
\newcommand{\sg}{{\sf{g}}}

\newcommand{\SPsi}{{\sf{\Psi}}}

\newcommand{\SK}{{\sf{K}}}
\newcommand{\SN}{{\sf{N}}}

\newcommand{\SU}{{\sf{U}}}
\newcommand{\SV}{{\sf{V}}}
\newcommand{\SW}{{\sf{W}}}

\newcommand{\sff}{{\sf{f}}}

\newcommand{\sv}{{\sf{v}}}

\newcommand{\SB}{{\sf B}}
\newcommand{\SQ}{{\sf Q}}

\newcommand{\SfS}{{\sf S}}

\newcommand {\bal}{\boldsymbol\alpha}
\newcommand{\bbeta}{\boldsymbol\beta}

\newcommand{\CR}{\mathcal R}
\newcommand{\CK}{\mathcal K}
\newcommand{\CB}{\mathcal B}

\newcommand{\CT}{\mathcal T}

\newcommand{\CH}{\mathcal H}

\newcommand{\CP}{\mathcal P}

\newcommand{\CA}{\mathcal A}

\newcommand{\CM}{\mathcal M}

\newcommand{\CC}{\mathcal C}

\newcommand{\plainW}[1]{\textup{{\textsf{W}}}^{#1}}
\newcommand{\plainC}[1]{\textup{{\textsf{C}}}^{#1}}
\newcommand{\plainB}{\textup{{\textsf{B}}}}

\newcommand{\plainN}[1]{\textup{{\textsf{N}}}^{#1}}

\newcommand{\plainL}[1]{\textup{{\textsf{L}}}^{#1}}

\DeclareMathOperator{\tr}{{tr}}

\newcommand{\scalel}[1]
{{\scaleto{#1}{3pt}}}
\newcommand{\scalet}[1]
{{\scaleto{#1}{4pt}}}

\newcommand{\4}{\ \!{\vrule depth3pt height9pt}
	{\vrule depth3pt height9pt}
	{\vrule depth3pt height9pt}{\vrule depth3pt height9pt}}

\DeclareMathOperator {\dist} {{dist}}

\DeclareMathOperator{\op}{{Op}}

\DeclareMathOperator{\iop}{{\sf Int}}

\DeclareMathOperator{\supp}{{supp}}
\DeclareMathOperator{\dc}{d}
\DeclareMathOperator{\DC}{D}

\newtheorem{thm}{Theorem}[section]
\newtheorem{cor}[thm]{Corollary}
\newtheorem{lem}[thm]{Lemma}
\newtheorem{prop}[thm]{Proposition}

\theoremstyle{definition}
\newtheorem{example}[thm]{Example}

\newtheorem*{remark}{Remark}
\newtheorem{rem}[thm]{Remark}

\numberwithin{equation}{section}

\newcommand{\Z}{\mathbb Z}

\newcommand{\CF}{\mathcal F}

\newcommand{\1}{\mathbbm 1}

\newenvironment{lcases}
  {\left\lbrace\begin{aligned}}
  {\end{aligned}\right.}

\begin{document}
	\hoffset -4pc

\title
[Eigenvalue asymptotics]
{{Eigenvalue asymptotics for the one-particle density matrix and one-particle kinetic energy density operator}} 
\author{S\o ren Fournais}
\address{ Department of Mathematics\\ 
University of Copenhagen \\
Universitetsparken 5 \\DK-2100 Copenhagen Ø \\
Denmark}
\email{fournais@math.ku.dk}
\author{Alexander V. Sobolev}
\address{Department of Mathematics\\ University College London\\
	Gower Street\\ London\\ WC1E 6BT UK}
 \email{a.sobolev@ucl.ac.uk}

\begin{abstract} 
Let $\psi(\bx)$, $\bx \in\R^{3N}$, be 
an eigenfunction of the $N$-particle atomic Schr\"odinger operator. We consider 
the one-particle density matrix  
$\g(x, y)$ and 
one-particle kinetic energy density 
$\varkappa(x, y)$, $x, y\in\R^3$, 
associated with the eigenfunction $\psi$. 
Both functions play a central 
role in quantum chemistry computations of atomic and molecular bound states: 
the knowledge of the eigenvalue behaviour 
of the integral operators $\SG$ and $\SK$ with kernels 
$\g(x, y)$ and $\varkappa(x, y)$  
serves to estimate the errors due to finite-dimensional 
approximations.  We find the following asymptotic formulas 
for their eigenvalues $\l_k(\SG)>0$ and $\l_k(\SK)>0$:
\begin{align*}
\lim_{k\to \infty} k^{\frac{8}{3}} \,\l_k(\SG) = A^{\frac{8}{3}},\quad
\lim_{k\to \infty} k^2\,\l_k(\SK) = B^2,
\end{align*}
where $A$ and $B$ are 
non-negative constants given explicitly in terms of the eigenfunction~$\psi$. These 
asymptotics are determined by the singularities of the function 
$\psi$ at pair coalescence points of the particles. To identify and isolate 
these singularities we use some recent regularity results for $\psi$. 
At the last step we apply Birman-Solomyak spectral asymptotics results for 
pseudodifferential operators with homogeneous symbols.

In the special case where the eigenfunction $\psi$ is totally antisymmetric, it exhibits 
enhanced regularity, which leads to a faster decay of 
the eigenvalues $\l_k(\SG)$ and $\l_k(\SK)$. The asymptotic formulas take the form 
\begin{align*} 
\lim_{k\to \infty} k^{\frac{10}{3}} \,\l_k(\SG) = 
\big(A_{\textup{asym}}\big)^{\frac{10}{3}},\quad 
\lim_{k\to \infty} k^{\frac{8}{3}} \,\l_k(\SK) = \big(B_{\textup{asym}}\big)^{\frac{8}{3}},
\end{align*}
where $A_{\textup{asym}}$ and $B_{\textup{asym}}$ are 
non-negative constants given explicitly in terms of the gradient of $\psi$. 
 
\end{abstract}

\maketitle 
\tableofcontents

\section{Introduction}  

This paper builds on the techniques developed in papers 
\cite{FS2021}, \cite{HearnSob2023}, \cite{Sobolev2022}, \cite{Sobolev2022a}. Our aim is 
to present the main results of these works in a coherent and 
streamlined manner while extending them with several new findings. 
In particular, compared with papers \cite{Sobolev2022} and 
\cite{Sobolev2022a}, Theorem \ref{thm:maingk} 
provides an improved formula for the 
asymptotic coefficients, making their dependence on the 
eigenfunction explicit. Theorem \ref{thm:maingkasym} is entirely new and 
establishes spectral asymptotics in the totally antisymmetric case.

\subsection{Background} 
Consider on $\plainL2(\R^{3N})$ the Schr\"odinger operator 
 \begin{align}
\CH = &\ \CH_0 + V,\quad \CH_0 = - \Delta = - \sum_{k=1}^N \Delta_k,\notag\\
V(\bx) = &\ - Z
\sum_{k=1}^N \frac{1}{|x_k|}  
 + \sum_{1\le j< k\le N} \frac{1}{|x_j-x_k|},\label{eq:potential}
\end{align}
describing an atom with $N$ electrons 
with coordinates $\bx = (x_1, x_2, \dots, x_N),\,x_k\in\R^3$, $k= 1, 2, \dots, N$, 
and a nucleus with charge $Z>0$ fixed at the origin. The notation $\Delta_k$ is used for 
the Laplacian w.r.t. the variable $x_k$. 
The operator $H$ acts on the Hilbert space $\plainL2(\R^{3N})$ 
and it is self-adjoint on the domain 
$D(\CH) =\plainW{2,2}(\R^{3N})$, since the potential in~\eqref{eq:potential} 
is an infinitesimal perturbation 
relative to the unperturbed operator $\CH$, see e.g.~\cite[Theorem X.16]{ReedSimon2}.   
Let $\psi = \psi(\bx)$ 
be an eigenfunction of the operator $\CH$ with eigenvalue $E\in\R$, i.e. $\psi\in D(\CH)$ and 
\begin{align}\label{eq:eigen}
(\CH-E)\psi = 0.
\end{align}  
We always assume that $\psi$ is normalized:
\begin{align}\label{eq:normalized}
\|\psi\|_{\plainL2(\R^{3N})}=1.
\end{align}
For each $j=1, \dots, N$, we represent
\begin{align*}
\bx = (\hat\bx_j, x_j), \quad \textup{where}\ 
\hat\bx_j = (x_1, \dots, x_{j-1}, x_{j+1},\dots, x_N),
\end{align*}
with obvious modifications if $j=1$ or $j=N$. 
The \textit{one-particle density matrix} is defined as the function 
\begin{align}\label{eq:den}
\g(x, y) = \sum_{j=1}^N\int\limits_{\R^{3N-3}} \overline{\psi(\hat\bx_j, x)} \psi(\hat\bx_j, y)\  
d\hat\bx_j,\quad (x,y)\in\R^3\times\R^3. 
\end{align} 
Note that traditionally, the one-particle density matrix is defined as  
$\g(y, x)$, see e.g. \cite{LLS2019}, but the permutation $x\leftrightarrow y$ 
will have no bearing on our results.

Introduce also the function 
\begin{align}\label{eq:tau}
\vark(x, y) = 
\sum_{j=1}^N\int\limits_{\R^{3N-3}} \overline{\nabla_x\psi(\hat\bx_j, x)} 
\cdot\nabla_y \psi(\hat\bx_j, y)\  
d\hat\bx_j, 
\end{align}
that we call the \textit {one-particle kinetic energy density matrix}. 
The choice of this term does not seem to be standard, but it is partly 
motivated by the fact that the integral
\begin{align*}
\int_{\R^3} \vark(x, x)\, dx 
\end{align*}
gives the kinetic energy of the $N$ particles, see e.g. 
\cite[Section 2A]{Davidson1976}, \cite[Chapter 3]{LiebSei2010} or 
\cite[Section 4]{LLS2019}. 
If we model the electrons as spinless fermions (resp. bosons), i.e. assume that the eigenfunction 
is fully antisymmetric (resp. symmetric) with respect to the permutations $x_k\leftrightarrow x_j$, 
$j, k = 1, 2, \dots, N, j\not = k$, 
then the formulas for $\g(x, y)$ and $\vark(x, y)$ take a more compact form:
\begin{align}\label{eq:fb}
\begin{lcases}
\g(x, y) = & N \int\limits_{\R^{3N-3}} 
\overline{\psi(\hat\bx, x)} \psi(\hat\bx, y)\ d\hat\bx,\\
\vark(x, y) = & N\int\limits_{\R^{3N-3}} \overline{\nabla_x\psi(\hat\bx, x)} 
\cdot\nabla_y \psi(\hat\bx, y)\,d\hat\bx,
\end{lcases}
\end{align}
where $\hat\bx = \hat\bx_N$.

Denote by $\SG = \iop(\g)$ and $\SK = \iop(\vark)$ the integral operators 
with kernels $\g(x, y)$ and $\vark(x, y)$ respectively. 
We call $\SG$ \textit {the one-particle density operator}, 
and $\SK$ - \textit {the one-particle kinetic energy density operator}. 
Since $\psi, \nabla\psi\in\plainL2(\R^{3N})$, 
both $\SG$ and $\SK$ are trace class.
We are interested in the exact decay rate of the eigenvalues $\l_k, k= 1, 2, \dots,$ 
for $\SG$ and $\SK$. 
Recall that in the physics literature 
the eigenvalues $\l_k(\SG)$ are 
called  \textit{occupation numbers}, see e.g. \cite{SzaboOstlund1996}.
Both functions $\g(x, y)$ and $\varkappa(x, y)$ are 
key objects in multi-particle quantum mechanics, see 
\cite{RDM2000, Davidson1976, LLS2019, LiebSei2010, SzaboOstlund1996} 
for details and futher references, 
and both operators $\SG$ and $\SK$ play a central 
role in quantum chemistry computations of atomic and molecular bound states, see e.g.  the papers 
\cite{Fries2003_1, Fries2003, Lewin2004, Lewin2011}  and the book \cite{SzaboOstlund1996}. 
The knowledge of the eigenvalue behaviour should 
serve to estimate the errors due to finite-dimensional 
approximations,  see e.g. 
\cite{Cioslowski2020}, 
\cite{CioStras2021},  \cite{Fries2003} and \cite{HaKlKoTe2012}.

The occupation number asymptotics in the case $N=2$ were discussed 
in \cite{Cioslowski2020, CioPrat2019}, and it was found that $\l_k(\SG)$  
behave like $k^{-8/3}$ for large $k$. The case of arbitrary $N$ 
was studied in \cite{Sobolev2022} (for $\SG$) 
and \cite{Sobolev2022a} (for $\SK$) 
under the assumption that the function $\psi$ satisfies an exponential bound 
\begin{align}\label{eq:exp}
\sup_{\bx\in\R^{3N}} e^{c |\bx|}\, 
 |\psi(\bx)|<\infty
\end{align}  
with some constant $c>0$.  
For discrete eigenvalues $E$ under 
the essential spectrum of $\CH$, the bound \eqref{eq:exp} follows from 
\cite{DHSV1978_79}. It also holds for embedded eigenvalues as long as they are away from the so-called \textit{thresholds}, see \cite{CombesThomas1973}, \cite{FH1982}. 
For more references and detailed discussion we quote \cite{SimonSelecta}. 
It was shown in \cite{Sobolev2022} and \cite{Sobolev2022a} that 
\begin{align}\label{eq:bolg}
\lim_{k\to \infty} k^{\frac{8}{3}} \,\l_k(\SG) = A^{\frac{8}{3}},\quad
\lim_{k\to \infty} k^2\,\l_k(\SK) = B^2,
\end{align}
with some non-negative constants $A$ and $B$. 

In the current paper we obtain two asymptotic results which are 
contained in Theorems~\ref{thm:maingk} 
and \ref{thm:maingkasym} respectively. The asymptotic formulas in Theorem \ref{thm:maingk} 
look exactly like \eqref{eq:bolg}. The main difference is that in 
Theorem \ref{thm:maingk} 
the coefficients $A$ and $B$ are given explicitly in terms of the eigenfunction $\psi$, 
whereas the expressions for $A$ and $B$ in \cite{Sobolev2022} and \cite{Sobolev2022a} are 
much less transparent. 
We elaborate on this point in Sect.~\ref{subsect:ra}. Another difference is that 
in both Theorems \ref{thm:maingk} and 
\ref{thm:maingkasym} we do not need 
the exponential bound~\eqref{eq:exp}, but make more natural assumptions 
on the decay of $\psi$ at infinity. 
Theorem \ref{thm:maingkasym} is new, and it concerns the case where the eigenfunction $\psi$ 
is totally antisymmetric, so that both coefficients $A$ and $B$ vanish, 
which means that the eigenvalues decay faster than in \eqref{eq:bolg}. 
Now we proceed to the precise statements and discussion of our main results. 

\subsection{Main results}
We obtain asymptotic formulas for $\l_k(\SG)$ and $\l_k(\SK)$ in two cases: 
without any symmetry assumptions on $\psi$, and under the assumption that $\psi$ is totally antisymmetric. Introduce the notation $\tilde\bx_{j,k}$ for the 
variable $\bx$ with $x_j$ and $x_k$ taken out:
\begin{align*}
\R^{3N-6} \ni \tilde\bx_{j, k} = \big( x_1, \dots, x_{j-1}, x_{j+1},\dots, x_{k-1}, x_{k+1},\dots, x_N\big), 
\ j < k. 
\end{align*}
With this notation we write $\bx = (\tilde\bx_{j, k}, x_j, x_k)$. 

Let us consider first the general case, and define the function $H(x), x\in\R^3$, by  
\begin{align}\label{eq:H}
\begin{lcases}
H(x) 
= &\ \frac{1}{8} 
\sum_{1\le j<k \le N}\, \int_{\R^{3N-6}}\,  
|\psi (\tilde\bx_{j, k}, x, x)|^2 \,d\tilde\bx_{j, k}, \ N\ge 3,\\[0.3cm]
H(x) =  &\ 
\frac{1}{8}\,   
|\psi (x, x)|^2, \ N = 2.
\end{lcases}
\end{align}
Introduce also the asymptotic coefficients: 
\begin{align}\label{eq:coeffAB}
A = \frac{1}{3}\bigg(\frac{2}{\pi}\bigg)^{\frac{5}{4}}
 \int_{\R^{3}} H(x)^{\frac{3}{8}}\, dx,\quad 
  B = \frac{4}{3\pi}
 \int_{\R^{3}} H(x)^{\frac{1}{2}}\, dx. 
\end{align}  
As shown by T. Kato in \cite{Kato1957}, the eigenfunction $\psi$ is Lipschitz. Together 
with some natural decay conditions in Theorem 
\ref{thm:maingk}, this ensures that the function $H(x)$ is well-defined and the coefficients 
\eqref{eq:coeffAB} are finite. 

The decay conditions on the function $\psi$ will 
be stated in terms of \textit{the one-particle density}
\begin{align}\label{eq:density}
\rho(x) := \g(x, x) = 
\sum_{j=1}^N\int\limits_{\R^{3N-3}} \, |\psi(\hat\bx_j, x)|^2\, d\hat\bx_j,
\end{align}
and its \textit{lattice quasi-norms} that are defined as follows. 
Let $\CC = [-1/2, 1/2)^3\in \R^3$ be a unit cube, and let $\CC_n = \CC+n, n\in \mathbb Z^3,$ 
be its translations by integer vectors.  For $p\ge 1$, $q >0$ and a function 
$f\in \plainL{p}_{\rm loc}(\R^3)$ define its \textit{lattice (quasi-)norm}
\begin{align}\label{eq:lattice}
\4\, f\4_{\, p, q} = \bigg[\sum_{n\in \Z^3} \, \|f\|_{\plainL{p}(\CC_n)}^q\bigg]^{\frac{1}{q}}.
\end{align}
This functional defines a norm (if $q\ge 1$) or a quasi-norm (if $q <1$). 
It is easily checked that for any function $f$ satisfying 
$|f(x)|\le C (1+|x|)^{-\b}, \b >0,$ the quasi-norm $\4\, f\4_{\, p, q}$ is finite for all $p \ge 1$ 
and $q>3\b^{-1}$.

\begin{thm}\label{thm:maingk}
Let $\psi\in D(\CH)$ be a normalized eigenfunction of $\CH$. 
\begin{enumerate}
\item \label{item:m}
Suppose that the one-particle density $\rho$ satisfies
$\4\,\rho\4_{\, 1, 3/8}<\infty$. 
Then the eigenvalues $\l_k(\SG), k = 1, 2, \dots,$ 
of the operator $\SG$  satisfy the asymptotic formula 
\begin{align}\label{eq:sgamma}
\lim_{k\to \infty} k^{\frac{8}{3}} \,\l_k(\SG) = A^{\frac{8}{3}}.
\end{align}
\item \label{item:gm}
Suppose that $\4\,\rho\4_{\, 1, 1/2}<\infty$. Then 
the eigenvalues $\l_k(\SK), k = 1, 2, \dots,$ 
of the operator $\SK$  satisfy the asymptotic formula 
\begin{align}\label{eq:skappa}
\lim_{k\to \infty} k^{2} \,\l_k(\SK) = B^2.
\end{align}
\end{enumerate}
\end{thm}

We show in  Sect. \ref{subsect:ascoeff} that the coefficients $A$ and $B$ 
are finite under the 
conditions imposed on the one-particle density in parts  
\eqref{item:m} and \eqref{item:gm} respectively.

Suppose now that $\psi$ is totally antisymmetric, i.e.~for all 
$x, y\in\R^3$, the function $\psi$ satisfies the relation 
\begin{align}\label{eq:antisym}
\psi(\tilde\bx_{j, k}, y, x) = 
- \psi(\tilde\bx_{j, k}, x, y),\quad \textup{for all}\quad 1\le j<k\le N.
\end{align}
Clearly, in this case both coefficients $A$ and 
$B$ in \eqref{eq:coeffAB} equal zero, which means that the eigenvalues of $\SG$ and $\SK$ have a faster decay rate than in Theorem~\ref{thm:maingk}. The relation~\eqref{eq:antisym} holds, e.g.~if one assumes that all $N$  particles are spinless fermions. On the other hand, 
the fermionic nature of particles may 
manifest itself differently if we introduce 
the spin variable. 
In this case the antisymmetry of the \textit{full} eigenfunction comes either from the spatial 
component $\psi(\bx)$ or from the spin component. 
If particles $j$ and $k$ 
are in a \textit{singlet} configuration, then the spin component is antisymmetric, whereas    
the function $\psi(\bx)$ itself is symmetric under the permutation $j\leftrightarrow k$, 
see \cite[Subsect.~3.3.1]{HaKlKoTe2012}. Then  
$\psi(\tilde\bx_{j, k}, x, x)$ is not identically zero,  
so both $A$ and $B$ are positive. 
If particles $j$ and $k$ 
are in a \textit{triplet} configuration, 
the antisymmetry under permutation $j\leftrightarrow k$ is carried by the spatial component 
$\psi(\bx)$, see \cite[Subsect. 3.3.2]{HaKlKoTe2012}, 
and then 
we have \eqref{eq:antisym} for the pair $j, k$. 
Thus we can say that the total antisymmetry 
\eqref{eq:antisym} is tantamount to all pairs of particles being in a triplet configuration. 
In this case 
the new asymptotic coefficients depend on the derivative 
\begin{align}\label{eq:sv}
\sv(\hat\bx, x):=\nabla_x\psi(\hat\bx, x), \, \hat\bx = \hat\bx_N.
\end{align}    
Since $\psi$ is Lipschitz, we have $\sv\in\plainL\infty(\R^{3N})$. 
As we shall see later in Remark \ref{rem:grad}, 
$\sv(\hat\bx, x)$ is in fact continuous 
on $\R^{3N}\setminus\{\bx = (\hat\bx, 0)\}$. 
Let $\tilde\bx = \tilde\bx_{N-1, N} = (x_1, x_2, \dots, x_{N-2})$, and 
introduce the function $E(x, \om)$, $x\in\R^3$, $\om\in\mathbb S^2$:
\begin{align}\label{eq:E}
\begin{lcases}
E(x, \om) =  &\ 
\frac{N(N-1)}{24^2} 
\int\limits_{\R^{3N-6}}\, 
| \sv(\tilde\bx, x, x)\cdot\om |^2 \,d\tilde\bx, 
\ N\ge 3,\\
E(x, \om) =  &\ \frac{2}{24^2} 
\,   | \sv(x, x)\cdot\om |^2, \ N = 2.
\end{lcases}
\end{align}
The new asymptotic coefficients are given by
\begin{align}\label{eq:coeffABasym}
A_{\textup{asym}} = \bigg(\frac{1}{3\pi^6}\bigg)^{\frac{2}{5}} \,
 \int_{\R^{3}}\, \int_{\mathbb S^2}\, E(x, \om)^{\frac{3}{10}}\, d\om dx,\quad 
  B_{\textup{asym}} = \bigg(\frac{8}{3\pi^9}\bigg)^{\frac{1}{4}}\,
 \int_{\R^{3}}\, \int_{\mathbb S^2}\, E(x, \om)^{\frac{3}{8}}\, d\om dx. 
\end{align}  
 As in Theorem \ref{thm:maingk}, their finiteness is guaranteed 
 under appropriate conditions on the function $\rho(x)$.

\begin{thm}\label{thm:maingkasym} Suppose that $\psi\in D(\CH)$ is a normalized eigenfunction of $\CH$ which 
is totally antisymmetric, i.e. \eqref{eq:antisym} holds.
\begin{enumerate}
\item 
Suppose that  
$\4\,\rho\4_{\, 1, 3/10}<\infty$. 
Then the eigenvalues $\l_k(\SG), k = 1, 2, \dots,$ 
of the operator $\SG$  satisfy the asymptotic formula 
\begin{align}\label{eq:sgamma_asym}
\lim_{k\to \infty} k^{\frac{10}{3}} \,\l_k(\SG) = 
\big(A_{\textup{asym}}\big)^{\frac{10}{3}}.
\end{align}
\item 
Suppose that $\4\,\rho\4_{\, 1, 3/8}<\infty$. Then 
the eigenvalues $\l_k(\SK), k = 1, 2, \dots,$ 
of the operator $\SK$  satisfy the asymptotic formula 
\begin{align}\label{eq:skappa_asym}
\lim_{k\to \infty} k^{\frac{8}{3}} \,\l_k(\SK) = \big(B_{\textup{asym}}\big)^{\frac{8}{3}} .
\end{align}
\end{enumerate}
\end{thm}
  
As in the case of Theorem \ref{thm:maingk}, the finiteness of the coefficients    
 $A_{\textup{asym}}$ and $B_{\textup{asym}}$ is established in Sect. \ref{subsect:ascoeff}.
  
A few remarks are in order.   
\begin{remark} 
\begin{enumerate}
\item 
Theorems \ref{thm:maingk} and \ref{thm:maingkasym} extend to the case of a 
molecule with several nuclei whose positions are fixed. 
The modifications are straightforward.
\item 
To satisfy the conditions of Theorems \ref{thm:maingk} and \ref{thm:maingkasym} it suffices 
to assume a polynomial decay of the one-particle density: $\rho(x)\le C(1+|x|)^{-\b}$ 
with a suitable $\b >0$. Indeed, this assumption ensures that 
$\4\,\rho\4_{\, 1, q}<\infty$ for $q\b > 3$. Thus, the conditions of Theorem 
\ref{thm:maingk} are satisfied for $\b > 8$ (Part \eqref{item:m}) and 
$\b >6$  (Part \eqref{item:gm}). Similarly for Theorem \ref{thm:maingkasym}.
On the other hand, as we have mentioned previously, the eigenfunction $\psi$ satisfies \eqref{eq:exp} in all 
``physically sensible" cases, e.g. for 
the eigenfunctions associated with eigenvalues below the essential spectrum. This leads to the 
exponential bound for the one-particle density: $\rho(x)\le C e^{-c_1|x|},\, x\in\R^3$, with some 
positive constants $C$ and $c_1$.  
Therefore the lattice quasi-norms in Theorems \ref{thm:maingk} 
and \ref{thm:maingkasym} are automatically finite. The conditions on the lattice norms make it clear ``how much" of the exponential decay is required for the main results.

\item 
Theorems \ref{thm:maingk} and \ref{thm:maingkasym} show that the spectral asymptotics 
depend only on the behaviour of the eigenfunction $\psi$ near the pair coalescence points 
$x_j=x_k$, $j, k = 1, 2, \dots, N$, $j\not = k$.  
Neither the points  
$x_j = 0, j = 1, 2, \dots, N$, nor 
the coalescence points of higher orders (e.g. $x_j = x_k = x_l$ with pair-wise distinct $j, k, l$)   
contribute to the asymptotics. 
\item 
Sharp upper bounds for the eigenvalues $\l_k(\SG)$ and $\l_k(\SK)$ were obtained in \cite{Sobolev2025}. For the general case considered in Theorem \ref{thm:maingk} it was shown that 
\begin{align*}
k^{\frac{8}{3}}\,\l_k(\SG)\le C\, \4\,\rho\4_{\, 1, 3/8},\quad 
k^{2}\,\l_k(\SG)\le C \, \4\,\rho\4_{\, 1, 1/2}.
\end{align*}
Under the assumption \eqref{eq:antisym}, it was shown that 
\begin{align*}
k^{\frac{10}{3}}\,\l_k(\SG)\le C\, \4\,\rho\4_{\, 1, 3/10},\quad 
k^{\frac{8}{3}}\,\l_k(\SG)\le C \, \4\,\rho\4_{\, 1, 3/8}.
\end{align*}
The constants $C>0$ in these bounds do not depend on the eigenfunction $\psi$. 
\end{enumerate}
\end{remark}

\subsection{Outline of the proof}
The proof follows the plan of \cite{Sobolev2022} and \cite{Sobolev2022a}. 
It splits into three steps that are briefly summarized below.  

\textit{Step 1: factorization.} 
First we represent the operators $\SG$ and $\SK$ as 
the products 
\begin{align}\label{eq:fctr}
\SG = \iop(\SPsi)^* \iop(\SPsi),\quad \SK = \iop(\SV)^*\, \iop(\SV),
\end{align}
where $\iop(\SPsi):\plainL2(\R^3)\mapsto \plainL2(\R^{3N-3}; \mathbb C^N)$ 
and 
$\iop(\SV):\plainL2(\R^3)\mapsto \plainL2(\R^{3N-3}; \mathbb C^{3N})$
are integral operators with vector-valued kernels $\SPsi(\hat\bx, x)$ and $\SV(\hat\bx, x)$
that are defined in Sect.~\ref{subsect:fact}. 
These representations imply that 
$\l_k(\SG) = s_k(\iop(\SPsi))^2$ and $\l_k(\SK) = s_k(\iop(\SV))^2$, $k=1, 2 ,\dots,$
where $s_k(T)$ are the singular values of the operator $T$. 
Therefore the asymptotic formulas \eqref{eq:sgamma} and \eqref{eq:skappa} take the form 
\begin{align}\label{eq:main1}
\lim_{k\to\infty}k\, s_k(\iop(\SPsi))^{\frac{3}{4}} = A,\quad 
\lim_{k\to\infty}k\, s_k(\iop(\SV)) = B,
\end{align}
and in a similar way one can interpret \eqref{eq:sgamma_asym} and 
\eqref{eq:skappa_asym}. 
If $\g(x, y)$ and $\varkappa(x, y)$ are given by 
\eqref{eq:fb}, then 
\begin{align}\label{eq:ferbos}
\SPsi(\hat\bx, x) = \sqrt N\, \psi(\hat\bx, x),\quad 
\SV(\hat\bx, x) = \sqrt N\, \nabla_x\, \psi(\hat\bx, x).
\end{align}
The precise definitions of $\SPsi$  and $\SV$ for the general case are given in Sect.~\ref{subsect:fact}.

For simplicity, in Steps 2 and 3 we limit our attention to 
the operator $\iop(\SPsi)$ as defined in \eqref{eq:ferbos}. 
Afterwards we briefly comment on the case of the 
operator $\iop(\SV)$.

\textit{Step 2: estimates for singular values.} 
The asymptotic analysis of the operator $\iop(\SPsi)$ 
begins with effective bounds for its 
singular values. By effective bounds we mean bounds for the 
operator $b\, \iop(\SPsi)\, a$ with weights 
$a = a(x), x\in\R^3$, $b = b(\hat\bx), \hat\bx\in \R^{3N-3}$,  
of the form 
\begin{align}\label{eq:efb}
s_k(b\, \iop(\SPsi)\, a)\le C(a, b) k^{-\frac{4}{3}},
\end{align}
where the factor $C(a, b)$ depends explicitly on some 
integral norms of $a$ and $b$.  
A general theory of spectral estimates for integral operators was developed 
by M.~Birman and M.~Solomyak in \cite{BS1977}. This monograph provides explicit sharp bounds 
for norms (or quasinorms) of integral operators in classes of compact operators 
in terms of the norms (or quasinorms) of their kernels in various function spaces. In particlular, 
for the bounds of the type \eqref{eq:efb} (i.e.~individual bounds for every singular value $s_k$), 
the appropriate function class is a certain \textit{Nikol'skii-Besov} space. 
Thus a substantial effort in our study 
is dedicated to estimating the Nikol'ski-Besov (quasi)norms 
for $\psi(\hat\bx, x)$ as a function of the 
variable $x$.  In this analysis we rely on the pointwise bounds for the 
function $\psi(\hat\bx, x)$ and its derivatives $\p_x^m\psi(\hat\bx, x)$ that were obtained 
by S.~Fournais and T.\O.~S\o rensen 
in the recent paper \cite{FS2021} and developed by P.~Hearnshaw and A.~Sobolev in \cite{HearnSob2023}. In the totally antisymmetric case 
the function $\psi$ possesses an enhanced smoothness, and we derive appropriate pointwise 
bounds for $\psi$ and its derivatives adapting the results of \cite{FHOS2005} and \cite{FS2021}.  
 
\textit{Step 3: asymptotics.}  
The asymptotic behavior of $s_n(\iop(\SPsi))$ as $n\to\infty$, 
is determined by 
the singularities of the function $\psi$ at the 
pair coalescence points $x_k = x_l$, $k\not = l$. 
To isolate the relevant singularities 
we use the canonical representation of the form 
\begin{align}\label{eq:isolate}
\psi(\bx) = e^{F(\bx)} \, \phi(\bx), 
\end{align}
where the function $F$ is chosen in such a way that the exponential factor contains the 
main singularities of $\psi$, whereas the factor $\phi$ is more regular than $\psi$. 
Due to the explicit form of the potential 
\eqref{eq:potential} such a function $F$ can be written explicitly, 
see Sect.~\ref{sect:regpsi} for details. 
In the literature the function $e^F$ is called a \textit{Jastrow factor}. 
To outline the main idea of the proof, 
we assume that the system consists of two particles, 
i.e. that $N=2$. 
Under this assumption the problem retains all its crucial features, 
but permits to avoid some tedious technical details. 
For $N=2$ we have $\bx = (t, x)\in \R^3\times\R^3$, 
and the operator $\iop(\SPsi)$, as defined in \eqref{eq:ferbos},  
acts from $\plainL2(\R^3)$ into $\plainL2(\R^3)$. 

Assume that both particles are separated from the nucleus, 
that is $|x|>\varepsilon$, $|t| > \varepsilon$ with some $\varepsilon>0$. Under this condition the main singularity of $e^F$ is $\frac{1}{4}|x-t|$, so that in the neighbourhood of the diagonal $x = t$ 
the function $\psi$ can be rewritten as 
\begin{align}\label{eq:split}
\psi(t, x) = \bigg(e^{F(t, x)} - \frac{1}{4}|x-t|\bigg) \, \phi(t, x) 
+ \frac{1}{4}|x-t|\, \phi(t, x),
\end{align}
where the first kernel on the right-hand side is smoother than the second one. 
Thus, in line with the results of \cite{BS1977}, the first kernel  
gives an asymptotic contribution of order 
$o\big(k^{-4/3}\big)$. The asymptotics \eqref{eq:main1} is produced by the second kernel in  \eqref{eq:split} since it contains the factor $|x-t|$ which 
is a homogeneous function of order one. 
Integral operators with homogeneous kernels 
represent a special case of pseudo-differential operators with asymptotically homogeneous symbols, 
which have been extensively studied by
M.~Birman and M.~Solomyak in \cite{BS1970, BS1977_1} and \cite{BS1979}, 
see also \cite{BS1977}. In fact, \cite{BS1977_1} and \cite{BS1979} also allow 
anisotropically homogeneous symbols, but we do not need these generalizations.  
Their results imply that a 
kernel of the form $Q(x-t)$ where $Q(x)$ is homogeneous of order $\a > -3$, 
leads to the eigenvalue 
asymptotics of order $k^{-1-\a/3}$. In our case $Q(x) = |x|$, and 
\eqref{eq:main1} follows from the Birman-Solomyak theory. 

For the operator $\iop(\SV)$, instead of the kernel $\psi(t, x)$ 
we study the vector-valued kernel $\nabla_x\psi(t, x)$. Differentiating 
\eqref{eq:split} we conclude that the leading singularity of $\nabla_x\psi(t, x)$ is the 
vector-valued kernel 
\begin{align*}
\frac{1}{4}\, \frac{x-t}{|x-t|}\, \phi(t, x). 
\end{align*}
Here the homogeneous factor has order zero, 
and hence the Birman-Solomyak theory leads to the asymptotics \eqref{eq:main1} 
for the operator $\iop(\SV)$.

Under the conditions of Theorem~\ref{thm:maingkasym}, i.e.~in the antisymmetric case, 
we also use a representation of the form \eqref{eq:isolate}, but with a different Jastrow factor $F$. 
Omitting the details we just mention that now 
the leading singularity of the function $\psi(t, x)$ 
(resp. $\nabla_x\psi(t, x)$) is homogeneous 
of order two (resp.~one), which leads to the asymptotics 
\eqref{eq:sgamma_asym} and \eqref{eq:skappa_asym}. 
This case requires the Birman-Solomyak asymptotics for matrix-valued symbols.   

For $N\ge 3$ the proof needs to be modified to incorporate representations of the form 
\eqref{eq:split} in the neighbourhood of each pair coalescence points 
$x_j = x_k$, \ $j, k = 1, 2, \dots, N$,\ $j \not = k$, see \eqref{eq:repr}.  
As a result, the proof requires an 
extra step involving the study of an 
 integral operator $\CM$ whose structure mimics that of the operators $\iop(\SPsi)$ 
 or $\iop(\SV)$ in the neighbourhood of the coalescence points. 
In order to make it adaptable for the use in different  situations, 
the model operator is allowed to have 
arbitrary order of homogeneity at the pair coalescence points.   
The study of the operator $\CM$ reduces to certain pseudo-differential operator 
with an asymptotically homogeneous symbol, whereupon the Birman-Solomyak results 
\cite{BS1977_1, BS1979} yield 
the spectral asymptotic formula for $\CM$. Application of this formula to 
the operators  $\iop(\SPsi)$ and $\iop(\SV)$ completes the proof of Theorems 
\ref{thm:maingk} and \ref{thm:maingkasym}.

\subsection{Plan of the paper}
The paper is organized as follows. 
Section \ref{sect:reg} contains some general results on local 
regularity for the scalar elliptic equation of the form \eqref{eq:mod} on a 
ball of variable radius $\ell$. The main focus is on the regularity bounds with explicit  
dependence on $\ell$. The main tool is the scale analysis suggested in \cite{FS2021} 
and supplemented by \cite{HearnSob2023}. 
The conclusions of Section \ref{sect:reg} are applied to the study of the eigenfunction $\psi$ 
in Sect.~\ref{sect:regpsi}. The main outcome of this section are bounds for the 
function $\psi$ and its derivatives of all orders with explicit dependence on the 
distance to the coalescence points. The whole of Sect. \ref{sect:regasym} is devoted to the analysis of regularity for  totally antisymmetric functions $\psi$. 
In Sect.~\ref{sect:compact} 
we put together some general information on  
compact operators with power-like spectral behaviour. Most of this information 
can be found in \cite[Chapter 11]{BS}. 
The focus of Sect.~\ref{sect:besov} is on the Nikol'skii-Besov spaces. Here 
we establish an elementary but very useful test of membership in these spaces.  
This test becomes useful in Sect.~\ref{subsect:BS} due to the general results of \cite{BS1977} 
on sharp spectral bounds for integral operators  
in terms of the norms (or quasinorms) of their kernels in 
the Nikol'skii-Besov spaces. They allow us to derive spectral 
estimates for some auxiliary integral operators thereby preparing the grounds for the study of 
the operators $\iop(\SPsi)$ and $\iop(\SV)$. In Sect.~\ref{subsect:symbols} we provide 
the asymptotic formula for pseudodifferential operators with homogeneous symbols 
obtained in \cite{BS1977_1}.
In Sect.~\ref{sect:model} the results gathered in Sect.~\ref{sect:intop} are applied to the 
study of the model operator operator $\CM$ mentioned earlier. 
The proof of the main results begins in Sect.~\ref{sect:factor}. First, in Sect.~\ref{subsect:ascoeff} we check that the coefficients \eqref{eq:coeffAB} and 
\eqref{eq:coeffABasym} are finite under the conditions of Theorems \ref{thm:maingk} and 
\ref{thm:maingkasym} respectively. 
In Sect.  \ref{subsect:fact} we detail the factorization \eqref{eq:fctr} 
which allows us to recast the main theorems 
in terms of the operators $\iop(\SPsi)$ and $\iop(\SV)$, see Theorems 
\ref{thm:recast} and \ref{thm:recastasym}. The rest of Sect. \ref{sect:factor} is 
taken up by technical lemmata designed to reduce the operators of interest to 
the model operator discussed in Sect. \ref{sect:model}. 
Sect. \ref{sect:prfrecast} and \ref{sect:prfrecastasym} contain the proofs 
of Theorems \ref{thm:recast} and \ref{thm:recastasym} respectively. These complete 
the proof of Theorems \ref{thm:maingk} and \ref{thm:maingkasym}. 
In the Appendix we give 
an elementary extension result for Sobolev spaces established 
in \cite{Sobolev2022a}, 
 which is used in 
Sect. \ref{subsect:mbr}.

\subsection{Notational conventions} 
We conclude the introduction with some general notational conventions.  

\textit{Coordinates.} 
As mentioned earlier, we use the following standard notation for the coordinates: 
$\bx = (x_1, x_2, \dots, x_N)$,\ where $x_j\in \R^3$, $j = 1, 2, \dots, N$. 
The vector $\bx$ is often represented in the form  
\begin{align*}
\bx = (\hat\bx_j, x_j) \quad \textup{with}\quad   
\hat\bx_j = (x_1, x_2, \dots, x_{j-1}, x_{j+1},\dots, x_N)\in\R^{3N-3}, 
\end{align*}
for arbitrary $j = 1, 2, \dots, N$. Most frequently we use this notation with $j=N$, and  
write $\hat\bx = \hat\bx_N$, so that $\bx = (\hat\bx, x_N)$. 
In order to write 
formulas in a more compact and unified way, we sometimes use the notation 
$x_0 = 0$. 

In the space $\R^d, d\ge 1,$ 
the notation $|x|$ stands for the Euclidean norm. 

For $N\ge 3$ it is also useful to introduce the notation 
for $\bx$ with $x_j$ and $x_k$ taken out: 
\begin{align}\label{eq:xtilde}
\tilde\bx_{k, j} = \tilde\bx_{j, k} 
= (x_1, \dots, x_{j-1}, x_{j+1}, \dots, x_{k-1}, x_{k+1},\dots, x_N), \quad \textup{for}\  j <k.
\end{align}
In this case we write $\bx = (\tilde\bx_{j, k}, x_j, x_k)$, $j < k$. For any $j\le N-1$ 
the vector $\hat\bx$ can be represented as $\hat\bx = (\tilde\bx_{j, N}, x_j)$.

For $R>0$ denote by $B(x, R) = \{z\in \R^d: |x-z|<R\}\subset\R^d$ the ball of radius $R>0$ centered at $x\in\R^d$. 
In particular, 
$B(\hat\bx, R) = \{\hat\by\in\R^{3N-3}: |\hat\bx-\hat\by|<R\}$.

\textit{Derivatives.} 
Let $\mathbb N_0 = \mathbb N\cup\{0\}$.
If 
\[ 
x = (x^{(1)}, x^{(2)},\dots,  x^{(d)})\in \R^d\quad \textup{ and}\quad  
m = (m^{(1)}, m^{(2)},\dots,  m^{(d)})\in\mathbb N_0^d,
\]
 then 
the derivative $\p_x^m$ is defined in the standard way:
\begin{align*}
\p_x^m = \p_{x^{(1)}}^{m^{(1)}}
\p_{x^{(2)}}^{m^{(2)}} \cdots 
\p_{x^{(d)}}^{m^{(d)}}.
\end{align*} 
The symbols $\plainW{l, p}(\Om)$, $\plainW{l, p}_{\textup{loc}}(\Om)$, $l = 0, 1, \dots$, 
$1\le p\le \infty$, 
 stand for the standard Sobolev spaces on the open set $\Om$, 
 and  $\plainC{l, \mu}(\overline{\Om})$, 
$l = 0, 1, \dots$, $\mu\in (0, 1]$,   
 stands for the standard H\"older space on $\Om$.

\textit{Bounds.} 
For two non-negative numbers (or functions) 
$X$ and $Y$ depending on some parameters, 
we write $X\lesssim Y$ (or $Y\gtrsim X$) if $X\le C Y$ with 
some positive constant $C$ independent of those parameters. 
If $X\lesssim Y\lesssim X$, then we write $X\asymp Y$. 
 To avoid confusion we may comment on the nature of 
(implicit) constants in the bounds. 
We also use the notation $X\wedge Y = \min\{X, Y\}$. 

\textit{Cut-off functions.} 
We systematically use the following smooth cut-off functions. Let 
\begin{align}\label{eq:sco}
\t, \z\in \plainC\infty([0, \infty)), 
\quad \z = 1-\t, 
\end{align}
be functions such that $0\le \t\le 1$ and 
\begin{align}\label{eq:sco1} 
\t(t) = 0,\quad \textup{if}\quad t>1;\ \quad
\t(t) = 1,\quad \textup{if}\quad 0\le t<\frac{1}{2}. \ 
\end{align}

\textit{Integral operators.} 
The notation $\iop(\CK)$ is used for the integral operator with kernel $\CK$, 
e.g. $\SG = \iop(\g)$.  
The functional spaces, where $\iop(\CK)$ acts are obvious from the context.

\section{Regularity estimates}\label{sect:reg}

\subsection{$\plainC{1, \e}$-regularity for elliptic equations}
In what follows we rely on the well-known $\plainC{1, \e}$-regularity bounds 
for solutions of second order elliptic equations on bounded domains.  
This type of regularity is discussed e.g. in \cite[Ch. 3]{LadUra1968} and  
\cite[Ch. 8]{GilTru2001}. 
In this paper we do not need the most general form of the equation.  
For our purposes it suffices to consider the equation 
\begin{align}\label{eq:mod}
\big(-\Delta + \ba(x)\cdot\nabla + b(x)\big) u = g, 
\end{align}
on an open ball $B\subset \R^d$, where 
all the coefficients are $\plainL\infty(B)$-functions. 
The proposition below provides some convenient sup-norm bounds for the 
weak solution and its first derivatives. Since the proof is quite short, we 
provide it for the sake of completeness. It can be found in \cite{HearnSob2023}, and 
it is similar to the argument in \cite[Proposition A.2]{FS2021}.

\begin{prop}\label{prop:reg} 
Let $B_R = B(x_0, R)\subset \R^d$ for some $x_0\in\R^d$ and $R>0$. 
Suppose that $u\in \plainW{1, 2}(B_R)$ is a 
weak solution of the equation \eqref{eq:mod}, 
where $\ba, b, g\in\plainL\infty(B_R)$, and 
\begin{align*}
\|\ba\|_{\plainL\infty(B_R)} + \|b\|_{\plainL\infty(B_R)}
\le M, 
\end{align*}
with some constant $M>0$. 
Then for any $r \in (0, R)$ the function $u$ belongs 
to $\plainW{2, 2}(B_r)\cap\plainC{1,\mu}(\overline{B_r})$ with an arbitrary 
$\mu\in [0, 1)$,  and 
\begin{align}\label{eq:regc1}
\|u\|_{\plainC{1, \mu}(\overline{B_r})}\lesssim \|u\|_{\plainL2(B_R)}
+ \|g\|_{\plainL\infty(B_R)},
\end{align}
with an implicit constant that depends only on the constant 
$M$, index $\mu$, dimension $d$ and the radii $r$ and $R$. 
\end{prop} 

\begin{proof} 
The inclusion $u\in \plainW{2, 2}(B_r)$ is a direct consequence of the standard interior regularity result given in, for example, \cite[Theorem 8.8]{GilTru2001}.

In order to prove that the weak solution $u\in \plainW{1, 2}(B_R)$ has 
the $\plainC1$-regularity in $B_r$ we repeatedly apply the following 
elementary fact.

Assume that $u\in \plainW{1, p}(B_\rho)$ with some $p \in( 1, \infty)$ and $\rho \le R$. Then for any 
$\nu < \rho$  the following is true:
\begin{enumerate}
\item 
If $p\le d$, then $u\in \plainW{1, q}(B_\nu)$ with $q = p(1+d^{-1})$ and 
\begin{align}\label{eq:step}
\|u\|_{\plainW{1,q}(B_\nu)}\lesssim &\ \|u\|_{\plainL{p}(B_\rho)} + \|g\|_{\plainL{p}(B_\rho)}\notag\\
\lesssim &\ \|u\|_{\plainW{1,p}(B_\rho)} + \|g\|_{\plainL{\infty}(B_\rho)}. 
\end{align}
\item 
If $p > d$, then $u\in \plainC{1, \e}(\overline{B_\nu})$, $\e = 1 - dp^{-1}$, and 
\begin{align}\label{eq:regc11}
\|u\|_{\plainC{1, \epsilon}(\overline{B_\nu})}\lesssim &\ \|u\|_{\plainL{p}(B_\rho)} 
+ \|g\|_{\plainL{p}(B_\rho)}\notag\\
\lesssim &\ \|u\|_{\plainW{1,p}(B_\rho)} + \|g\|_{\plainL{\infty}(B_\rho)}.
\end{align}
\end{enumerate}

Indeed, under the assumption $u\in \plainW{1, p}(B_\rho)$, by interior $\plainL{p}$-estimates (see, e.g. \cite[Theorem 9.11]{GilTru2001}) we have $u\in\plainW{2, p}(B_\nu)$ and the standard bound 
holds:
\begin{align}\label{eq:gt}
\|u\|_{\plainW{2,p}(B_\nu)}\lesssim \|u\|_{\plainL{p}(B_\rho)} + \|g\|_{\plainL{p}(B_\rho)}.
\end{align}
If $p < d$, then we use the bounded embedding $\plainW{2,p}\subset \plainW{1, q}$,\ for all 
$q\in[p, p^*]$, $p^* = dp(d-p)^{-1}$, see e.g. 
\cite[Theorem 6, Sect. 5.7]{Evans1998}. In particular, the value 
$q = p(1+d^{-1})$ belongs to the interval $[p, p^*]$, which proves  
\eqref{eq:step}. 
If $p = d$, then $\plainW{2,p}\subset \plainW{1, q}$ for all $q\in [p, \infty)$, and hence 
for $q = p(1+d^{-1})$ in particular. Hence \eqref{eq:step} holds again.

If $p >d$, then we use the embedding 
$\plainW{2,p}\subset \plainC{1, \e}$, see \cite{Evans1998}, so that 
\eqref{eq:gt} leads to \eqref{eq:regc11}.

Let us proceed with the proof of \eqref{eq:regc1}. 
If $d = 1$, then the solution $u\in\plainW{1, 2}(R)$ satisfies \eqref{eq:regc11} with 
$p=2$, which immediately implies \eqref{eq:regc1}.
Suppose that $d\ge 2$, and define the sequence 
\begin{align*}
q_n = 2(1+d^{-1})^n, \ n = 0, 1, \dots. 
\end{align*}
Let $k\ge 1$ be the index such that $q_{k-1}\le d$ and $q_{k}\ge d(1-\mu)^{-1}>d$.  
Pick finitely many numbers $r_n>0$ such that 
\begin{align*}
r < r_k<r_{k-1}<\dots < r_1 < r_0  = R.
\end{align*} 
Since $u\in\plainW{1,2}(B_R)$ we can apply the bound \eqref{eq:step} with $\nu = r_1, \rho = R$ 
and $p = q_0 = 2$, $q = q_1$. Repeating this step successively for $\rho =r_n, \nu =r_{n+1}$ and 
$p = q_{n}$, $q = q_{n+1}$  
for all $n = 1, \dots, k$, we arrive at the bound 
\begin{align*}
\|u\|_{\plainW{1,q_{k}}(B_{r_k})} 
\lesssim &\ 
\|u\|_{\plainW{1, q_{k-1}}(B_{r_{k-1}})} + \|g\|_{\plainL{\infty}(B_{r_{k-1}})}
\lesssim \dots\\
\lesssim  &\ \|u\|_{\plainW{1, q_{\scalel{1}}}(B_{r_{\scalel{1}}})} 
+ \|g\|_{\plainL{\infty}(B_{r_{\scalel{1}}})}
\lesssim 
\|u\|_{\plainL{2}(B_R)} + \|g\|_{\plainL{\infty}(B_R)}.
\end{align*}
As $q_k \ge d(1-\mu)^{-1}> d$, we can now use \eqref{eq:regc11} which gives
\begin{align*}
\|u\|_{\plainC{1, \e}(\overline{B_r})} 
\lesssim &\ 
\|u\|_{\plainW{1,q_{k}}(B_{r_k})} +  \|g\|_{\plainL{\infty}(B_{r_k})}\\
\lesssim &\ 
\|u\|_{\plainL{2}(B_R)} + \|g\|_{\plainL{\infty}(B_R)},
 \end{align*}
 where $\e = 1 - dq_k^{-1}\ge \mu$. 
This completes the proof of \eqref{eq:regc1}. 
\end{proof}

Higher order smoothness. 
Let $x = (z, t)$ where $z\in \R^n$ and $t\in \R^l$, and $n + l = d$. 
Assume that for some $\bk\in \mathbb N_0^d$, we have that 
all the derivatives $\p_z^\bm g, 0\le \bm\le \bk,$ are bounded on $B_R:=B(x_0, R)$ and  
\begin{align}\label{eq:highreg}
|\p_z^\bm\, \ba(z, t)| + |\p_z^\bm b(z, t)| \le M,\quad \textup{for all}\quad 0\le \bm\le \bk, 
\end{align}
and a.e. $x\in B_R$. Instead of the equation \eqref{eq:mod} we consider 
a similar equation depending on the parameter 
 $\ell\in (0, 1]$:
\begin{align}\label{eq:modmod}
-\Delta u + \ell\, \ba\cdot\nabla u + \ell^2 b\, u = g.
\end{align}

\begin{thm}\label{thm:modmod}
Suppose that the condition \eqref{eq:highreg} is satisfied. 
Let $u\in \plainW{1, 2}(B_R)$ be a weak solution of the equation \eqref{eq:modmod}.  
Then for any $r \in (0, R)$,  
and all $\bold 0\le \bm\le\bk$ 
the function $\p_z^\bm u$ belongs to 
$\plainW{2,2}(B_r)\cap \plainC{1, \mu}(B_r)$ for an arbitrary $\mu\in [0, 1)$ 
and 
\begin{align}\label{eq:dd0}
\|\p_z^\bm \, u\|_{\plainC{1, \mu}(\overline{B_r})}\lesssim 
\|u\|_{\plainL2(B_R)} + \sum_{\bold0 \le \bq\le \bm}\ \|\p_z^\bq\, g\|_{\plainL\infty(B_R)}.
\end{align}
If $|\bm|\ge1$, then for any $\nu\in (r, R)$,
\begin{align}\label{eq:dd1}
\|\p_z^\bm \, u\|_{\plainC{1, \mu}(\overline{B_r})}\lesssim 
\ell^2\,\|u\|_{\plainL\infty(B_\nu) }  + \|\nabla u\|_{\plainL\infty(B_\nu) }
+ \sum_{0 < \bq\le \bm}\ \|\p_z^\bq\, g\|_{\plainL\infty(B_R)}. 
\end{align} 
If $|\bm|\ge2$, then 
\begin{align}\label{eq:dd2}
\|\p_z^\bm \, u\|_{\plainC{1, \mu}(\overline{B_r})}\lesssim 
\ell^2\,\|u\|_{\plainL\infty(B_\nu) }&  + \ell\,\|\nabla u\|_{\plainL\infty(B_\nu) }\notag \\
&\ + \max_{\bl\le\bk: |\bl|=2}\|\p_z^\bl\, u\|_{\plainL\infty(B_\nu)}
+ \sum_{0 < \bq\le \bm}\ \|\p_z^\bq\, g\|_{\plainL\infty(B_R)}.
\end{align}
\end{thm}

\begin{proof} 
The bound \eqref{eq:dd0} is a direct consequence of \eqref{eq:dd1} and \eqref{eq:regc1}. 

Proof of \eqref{eq:dd1}. 
We begin with a formal manipulation assuming 
that all the cluster derivatives $u_\bm$ exist and are as smooth as necessary. 
Applying the operator $\p_z^{\bm}$ to the equation \eqref{eq:modmod} 
we obtain the following equation for the function 
$u_\bm := \p_z^\bm u$:
\begin{align}\label{eq:cluster}
-\Delta u_\bm +  \ell\ \ba\cdot\nabla u_\bm  + \ell^2 b\, u_\bm = G_\bm,
\end{align}
with 
\begin{align*}
G_\bm := -  \ell \sum\limits_{\bold0\le \bq <\bm} 
{\bm\choose \bq}\big(\p_z^{\bm-\bq} \ba\big)\cdot \nabla u_{\bq}
-\ell^2 
\sum\limits_{\bold0\le \bq <\bm} 
{\bm\choose \bq}\big(\p_z^{\bm-\bq} b\big)\, u_{\bq} 
+  g_\bm, 
\end{align*}
where $g_\bm = \p_z^\bm g$. 
It follows from \eqref{eq:highreg} that 
\begin{align}\label{eq:prq1}
\|G_\bm\|_{\plainL\infty(B_\rho)}
\lesssim  &\ \sum\limits_
{\bold0\le \bq < \bm} T_{\bq}(\l, \rho) 
+ \|g_\bm\|_{\plainL\infty(B_\rho)},\notag \\[0.2cm] 
&\qquad T_{\bq}(\l, \rho) = \ell\,\|\nabla u_\bq\|_{\plainL\infty(B_\rho)} 
+ \ell^2 \|u_\bq\|_{\plainL\infty(B_\rho)},
\end{align}
for all $\rho < R$. 

Now we need to justify the above formal calculations. First note that 
since $u$ is a weak solution of \eqref{eq:modmod}, by Proposition \ref{prop:reg} it belongs to 
$\plainW{2, 2}(B_\rho)\cap\plainC1(\overline{B_\rho})$ for all $\rho <R$, so that 
$T_\bold0(\ell, \rho) = 
\ell\|\nabla u\|_{\plainL\infty(B_\rho)} + \ell^2 \|u\|_{\plainL\infty(B_\rho)}$ is finite, 
and as a result, for $|\bm| = 1$ we have 
\begin{align*}
\|G_{\bm}\|_{\plainL\infty(B_\rho)}
\lesssim &\ T_{\bold0}(\ell, \rho) + \| g_\bm\|_{\plainL\infty(B_{\rho})}\notag\\
\lesssim &\ 
\ell\|\nabla u\|_{\plainL\infty(B_{\rho})} + \ell^2\, \|u\|_{\plainL\infty(B_{\rho})} 
+ \| g_\bm\|_{\plainL\infty(B_\rho)},\quad 
0<\rho <R.
\end{align*} 
By Proposition \ref{prop:reg}, 
for all $\bm$ such that $|\bm|=1$ and all $r < \rho$, we conclude that 
$u_{\bm} \in \plainW{2, 2}(B_r)\cap\plainC1(\overline{B_r})$ and 
\begin{align*}
\|u_\bm\|_{\plainC{1, \mu}(\overline{B_{r}})}
\lesssim &\ \|u_\bm\|_{\plainL\infty(B_\rho)} + \|G_\bm\|_{\plainL\infty(B_\rho)}\\
\lesssim &\ \|\nabla u\|_{\plainL\infty(B_\rho)} + \|G_\bm\|_{\plainL\infty(B_\rho)}
\lesssim  \|\nabla u\|_{\plainL\infty(B_{\rho})} 
+ \ell^2\, \|u\|_{\plainL\infty(B_{\rho})} + \| g_\bm\|_{\plainL\infty(B_\rho)}.
\end{align*}
This proves \eqref{eq:dd1} for $|\bm|=1$.  

Using the above observation as the induction base, 
we now prove  
that $u_\bm$ is indeed well-defined for all $\bm\le\bk$
 and that it 
is a weak solution of 
the equation \eqref{eq:modmod}. 
To provide the induction step 
assume that for some $s\le |\bk|-1$ and all $\bm\le\bk$ such that $|\bm|\le s$, 
the function $u_\bm$ is a weak solution 
of \eqref{eq:modmod} in $B_\rho$  
for all $\rho < R$ 
(and hence belongs to 
$\plainW{2, 2}(B_r)\cap\plainC{1, \mu}(\overline{B_r}), r < \rho,$ by Proposition \ref{prop:reg}), 
and that it satisfies 
\eqref{eq:dd1}. 
Let us prove that the same is true for
the function $u_\bn$ for all $\bn\le\bk: |\bn|= s+1$.

First note that for all $\rho <\nu < R$ 
the function $G_{\bn}$ satisfies the bound 
\begin{align}\label{eq:prq2}
\|G_{\bn}\|_{\plainL\infty(B_\rho)}\lesssim 
\ell\|\nabla u\|_{\plainL\infty(B_{\nu})} + \ell^2\, \|u\|_{\plainL\infty(B_{\nu})}
+ \sum_{\bold0< \bq\le \bn} \|g_\bq\|_{\plainL\infty(B_{\nu})}.
\end{align}  
Indeed, in the bound \eqref{eq:prq1},  
for the term $T_{\bold0}(\ell, \rho)$ we use the exact formula (second line in \eqref{eq:prq1}).  
For $T_\bq(\ell, \rho)$ with $ \bold0 <\bq<\bn$, we use 
\eqref{eq:dd1} to obtain
that 
\begin{align*}
T_\bq(\ell, \rho) = & \ell\,\|\nabla u_\bq\|_{\plainL\infty(B_\rho)} 
+ \ell^2 \,\|u_\bq\|_{\plainL\infty(B_\rho)}\\
\lesssim  &\  \ell \,\|u_\bq\|_{\plainC1(\overline{B_{\rho}})} 
 \lesssim  \ell\,\|\nabla u\|_{\plainL\infty(B_{\nu})} 
+ \ell^3\, \|u\|_{\plainL\infty(B_{\nu})} 
+ 
\ell\, \sum_{\bold0 < \bq'\le \bq} \|g_{\bq'}\|_{\plainL\infty(B_{\nu})},
\end{align*}
which leads to \eqref{eq:prq2}. 
Furthermore, as $u_\bm\in \plainW{2, 2}(B_\rho)$ for all 
$|\bm|\le s$, we have  
$u_{\bn}\in \plainW{1, 2}(B_\rho)$. Now,  
 integrating the equation \eqref{eq:modmod} 
 against the function $(-1)^{s+1}\p_z^{\bn} \eta$, 
 with an arbitrary $\eta\in\plainC\infty_0(B_\rho)$ we obtain 
 \begin{align*}
0 = (-1)^{s+1}\int \nabla u\cdot \nabla\ \p_z^{\bn} \eta\, d x
 &\ + (-1)^{s+1}  \ell \int \ba \cdot \nabla u \,\p_z^{\bn} \eta\, dx\\ 
 + &\ (-1)^{s+1} \ell^2\int b\, u\, \p_z^{\bn} \eta\, dx - 
 (-1)^{s+1} \ell^2\int g\, \p_z^{\bn} \eta\, dx
  \end{align*}
 \begin{align*}
= \int \nabla u_\bn\cdot \nabla \eta\, dx
 +    \ell &\ \int \ba\cdot \nabla u_\bn \, \eta\, dx\\ 
 + &\ \ell^2\,\int b\, u_\bn  \eta\, dx 
 -  \int G_{\bn}\, \eta\, dx, 
 \end{align*}
and hence $u_{\bn}$ is a weak solution of \eqref{eq:cluster} in $B_\rho$. 
Since the coefficients on the left-hand side 
of \eqref{eq:cluster} 
are bounded uniformly in $\ell\in (0, 1]$, by  Proposition \ref{prop:reg}, 
$u_\bn\in\plainW{2, 2}(B_r)\cap\plainC{1, \mu}(\overline{B_r})$ for all $r<\rho$ and 
\begin{align}\label{eq:n2}
\|u_{\bn}\|_{\plainC{1, \mu}(\overline{B_r})}
\lesssim &\ \|u_{\bn}\|_{\plainL\infty(B_\rho)}
+ \|G_{\bn}\|_{\plainL\infty(B_\rho)}\notag\\
\lesssim &\ \|u_{\bn}\|_{\plainL\infty(B_\rho)} + 
\ell\|\nabla u \|_{\plainL\infty(B_{\nu})} + \ell^2\, \|u\|_{\plainL\infty(B_{\nu})}
+  \sum\limits_{\bold0 < \bq\le \bn}
\|g_\bq\|_{\plainL\infty(B_{\nu})}, 
\end{align}
where we have also used \eqref{eq:prq2}. 
Let $\bm\le \bk, |\bm|=s$, be such that $|\bn-\bm|=1$. Therefore, by \eqref{eq:dd1},  
\begin{align*}
\|u_\bn\|_{\plainL\infty(B_\rho)}\lesssim \|\nabla u_\bm\|_{\plainL\infty(B_\rho)}
\lesssim \|\nabla u\|_{\plainL\infty(B_\nu)} + \ell^2\, \|u\|_{\plainL\infty(B_\nu)} 
+ \sum\limits_{\bold0 < \bq\le \bm}
 \|g_\bq\|_{\plainL\infty(B_{\nu})}.
\end{align*}
Together with \eqref{eq:n2} this leads to 
\begin{align*}
\|u_{\bn}\|_{\plainC{1, \mu}(\overline{B_r})}
\lesssim \|\nabla u\|_{\plainL\infty(B_\nu)}
+ \ell^2\|u\|_{\plainL\infty(B_\nu)} 
+ \sum\limits_{\bold0< \bq\le \bn} 
\|g_\bq\|_{\plainL\infty(B_{\nu})},
\end{align*}
as required. 
It remains to conclude that by induction, the bound \eqref{eq:dd1}  
holds for all $\bm\le \bk$.  
  
The proof of \eqref{eq:dd2} is also conducted by induction.  
 Assume that $|\bm|=2$. From 
 \eqref{eq:n2} we get that 
\begin{align}\label{eq:m2}
\|u_{\bm}\|_{\plainC{1,\mu}(\overline{B_r})}
\lesssim &\ 
\|u_{\bm}\|_{\plainL\infty(B_\rho)} + \ell\,\|\nabla u\|_{\plainL\infty(B_\nu)}
+ \ell^2\,\|u\|_{\plainL\infty(B_\nu)} 
+ \sum\limits_{\bold0< \bq\le \bm} 
\|g_\bq\|_{\plainL\infty(B_{\nu})}
\notag\\
\lesssim &\ 
\max_{\bl\le \bk: |\bl|=2} \|\p_z^{\bl} u\|_{\plainL\infty(B_\nu)} + \ell \|\nabla u\|_{\plainL\infty(B_\nu)}
+ \ell^2 \|u\|_{\plainL2(B_\nu)} 
+ \sum\limits_{\bold0< \bq\le \bm} 
\|g_\bq\|_{\plainL\infty(B_{\nu})},
 \end{align}
for all $r < \rho < \nu<R$, which gives \eqref{eq:dd2}. 
The bound \eqref{eq:m2} serves as the induction base. Let us now provide the induction step. 
Suppose that \eqref{eq:dd2} holds for all $\bm\le\bk$ such that $2\le |\bm|\le s$ 
with some $s = 2, 3, \dots, |\bk|-1$. Let us prove that it holds for 
all  $\bn\le \bk, |\bn|=s+1$. 
Let $\bm, |\bm|=s$, be such that $|\bn-\bm|=1$. 
Thus \eqref{eq:dd2} for $u_\bm$ implies that  
\begin{align*}
\|u_\bn\|_{\plainL\infty(B_\rho)}\lesssim &\ \|\nabla u_\bm\|_{\plainL\infty(B_\rho)} \\
\lesssim &\ 
\max_{\bl\le\bk: |\bl|=2} \|\p_z^{\bl} u\|_{\plainL\infty(B_\nu)} 
+ \ell \|\nabla u\|_{\plainL\infty(B_\nu)}
+ \ell^2 \|u\|_{\plainL2(B_\nu)}
+ \sum\limits_{\bold0< \bq\le \bm}  
 \|g_\bq\|_{\plainL\infty(B_{\nu})}.
\end{align*}
Substituting this inequality in \eqref{eq:n2}, we obtain \eqref{eq:dd2} for $u_\bn$. 

Consequently, \eqref{eq:dd2} holds for all $\bm\le \bk$, $|\bm|\ge 2$. 
\end{proof}

\subsection{Dependence on the radius}
Our next step is to establish the dependence of the 
bounds for the solution of \eqref{eq:mod} on the radius of the ball. 
Precisely, consider a weak solution $u$ of the equation \eqref{eq:mod} in the ball 
$B(x_0, R\ell)$ with some $x_0\in\R^d$, $R >0$, $\ell\in (0, 1]$. 
Suppose that for some $\bk\in \mathbb N_0^{d}$,  
the coefficients $\ba$ and $b$ in \eqref{eq:mod} 
satisfy the bounds 
\begin{align}\label{eq:cd}
|\p_z^{\bm} \ba(x)| + 
|\p_z^{\bm} b(x)|\lesssim \ell^{-|\bm|}, \quad x\in B(x_0, R\ell),
\end{align}
for all $\bm\le \bk$, with 
constants potentially depending on $\bk$, $R$ and $x_0$, but not on $\ell$. 
We also assume that derivatives $\p_z^\bm g$ are bounded 
on $B(x_0, R\ell)$ for all $\bm\le \bk$ 
(we do not assume that they are bounded uniformly in $\ell$).
In the next theorem we obtain bounds for the derivatives $\p_z^{\bm} u$ 
with explicit dependence on the parameter $\ell\in (0, 1]$.

\begin{thm}\label{thm:reg2}  
Assume the conditions \eqref{eq:cd}, and let $u$ be a weak solution of the equation 
\eqref{eq:mod} in $B(x_0, R\ell)$.
Then for all $\bm\le \bk$ and all $r<R$ 
the derivatives $\p_z^{\bm} u$ belong to 
$\plainC1\big(\overline{B(x_0, r\ell)}\big)$. 
Furthermore, if $|\bm|+k\ge 1$, where $k = 0, 1$, then for all $\nu\in (r, R)$ we have 
\begin{align}\label{eq:reg2}
\|\nabla^k \p_z^{\bm} u\|_{\plainL\infty(B(x_0, r\ell))}
\lesssim \ell^{1-|\bm|-k}&\, \bigg[\ell\|u\|_{\plainL\infty(B(x_0, \nu\ell))}
+ \|\nabla u\|_{\plainL\infty(B(x_0, \nu\ell))} \notag\\
&\ \qquad+ \ell\sum_{\bold0< \bq\le \bm}
\ell^{|\bq|} \, \|\p_z^\bq g\|_{\plainL\infty(B(x_0, R\ell))}\bigg].
\end{align}
If $|\bm|\ge 2$, then also  
\begin{align}\label{eq:d2ell}
\|\p_z^{\bm} u\|_{\plainL\infty(B(x_0, r\ell))}
\lesssim 
\ell^{2-|\bm|}&\,\bigg[
\max_{\bl\le\bk: |\bl|=2} \|\p_z^{\bl} u\|_{\plainL\infty(B(x_0, \nu \ell))}
+ \|\nabla u\|_{\plainL\infty(B(x_0, \nu\ell))}\notag\\
&\ \qquad+  \|u\|_{\plainL\infty(B(x_0, \nu\ell))} 
+ \sum_{\bold0 < \bq\le \bm}
\ell^{|\bq|} \, \|\p_z^\bq g\|_{\plainL\infty(B(x_0, R\ell))}\bigg].
\end{align}
The implicit constants in \eqref{eq:reg2} and \eqref{eq:d2ell} may depend 
on the constants $r, \nu, R$, order $\bk$ of the derivative, and the 
constants in \eqref{eq:cd}. In particular, 
if the constants in \eqref{eq:cd} are independent of $x_0$, then 
so are the constants in \eqref{eq:reg2} and \eqref{eq:d2ell}.
\end{thm}
 
\begin{proof}
The idea of the proof follows an argument in \cite{FS2021}. Namely, we rescale the 
problem by introducing the function 
\begin{align*}
w(y) = u(x_0 + \ell y),\quad y\in B(0, R),
\end{align*}
and the coefficients 
\begin{align*}
\bal(y) =  \ba(x_0+\ell y),\quad 
\b(y) = b(x_0+\ell y), \quad \om(y) = g(x_0+ \ell y).
\end{align*}
After the scaling the equation \eqref{eq:mod} takes the form
\begin{align}\label{eq:mod_scale}
-\Delta w + \ell\ \bal\cdot\nabla w
+  \ell^2 \beta w = \ell^2 \om.
\end{align}
Note that by \eqref{eq:cd},
\begin{align*}
|\p_z^{\bm}\bal(y)| + 
|\p_z^{\bm}\b(y)|\lesssim 1, \quad y\in B(0, R)
\end{align*}
for all $\bm\le\bk$.

The equation \eqref{eq:mod_scale} has the form \eqref{eq:modmod} with  
$\ell^2 \om$ instead of $g$. Therefore we can apply Theorem 
\ref{thm:modmod} which leads to the following bounds for 
$\p_z^\bm w$:
\begin{align}\label{eq:d1}
\|\p_z^\bm w\|_{\plainC1(\overline{B_{r}})}
\lesssim \ell^2\, \|w\|_{\plainL\infty(B_{\nu})} + \|\nabla w\|_{\plainL\infty(B_{\nu})} 
+ \ell^2 \sum_{\bold0<\bq\le \bm}\, \|\p_z^\bq \om\|_{\plainL\infty(B_{\nu})}.
\end{align}
If $|\bm|\ge 2$, then 
\begin{align}\label{eq:d2}
\|\p_z^\bm w\|_{\plainC1(\overline{B_r})}
\lesssim 
\max_{\bl\le\bk: |\bl|=2} \|\p_z^{\bl} w\|_{\plainL\infty(B_\nu)}
&\ + \ell \|\nabla w\|_{\plainL\infty(B_\nu)}\notag\\
&\ + \ell^2 \|w\|_{\plainL\infty(B_\nu)} + 
\ell^2 \sum_{\bold0<\bq\le \bm}\, \|\p_z^\bq \om\|_{\plainL\infty(B_{\nu})}.
\end{align}
Since
\begin{align*}
 \nabla^k w_{\bm}(y) = \ell^{|\bm|+k}\big(\nabla^k 
 \p_z^{\bm} u\big)( x_0+\ell y), 
\,  
\om_\bm(y) = \ell^{|\bm|} \, \p_z^\bm g(x_0 + \ell y),\quad 
k = 0, 1,\, \bm\in\mathbb N_0^{3M},
\end{align*}
the bound \eqref{eq:d1} for $|\bm|\ge 1$ rewrites as \eqref{eq:reg2}. If $\bm = \bold0$ 
and $|k| = 1$, then \eqref{eq:reg2} is trivial. 

The bound \eqref{eq:d2ell} is obtained from \eqref{eq:d2} in the same way. 
\end{proof} 
 
Now we consider the equation \eqref{eq:mod} on an open set $\Om\subset \R^d$. Let  
\begin{align}\label{eq:dc}
\dist(x) = \dist_\Om(x):= \inf_{z\notin\Om} \, |x-z|
\end{align}
 be the distance from $x\in \R^d$ 
 to the complement of $\Om$. The function $\dist(x)$ is Lipschitz with Lipschitz constant one:
 \begin{align*}
 |\dist(x) - \dist(y)|\le |x-y|,\quad \textup{for all}\quad x, y\in\Om.
 \end{align*}
Instead of the condition \eqref{eq:cd} we assume that 
\begin{align}\label{eq:cdscale}
|\p^\bm_z \ba(x)| + |\p^\bm_z b(x)|\lesssim \ell(x)^{-|\bm|},\ 
\ell(x) = \min\{1, \dist(x)\}, \quad x\in\Om.
\end{align} 
for all $\bm\le\bk$, with some $\bk\in \mathbb N_0^d$. We also assume that for all $\bm\le\bk$ the derivatives 
$\p^\bm_z g$ are bounded uniformly on every compact subset of $\Om$.

\begin{thm}\label{thm:dist}
Assume the condition \eqref{eq:cdscale}, and let $u\in\plainW{1,2}(\Om)$ be a weak solution of \eqref{eq:mod} in $\Om$. 
Then for all $\bm\le \bk$  
the derivatives $\p_z^{\bm} u$ belong to 
$\plainC1(\Om)$. 
Furthermore, if $|\bm|+k\ge 1$, where $k = 0, 1$, then for all $x\in\Om$ and any $R<1$ we have 
\begin{align}\label{eq:reg2dist}
|\nabla^k \p_z^{\bm} u(x)|
\lesssim \ell(x)^{1-|\bm|-k}&\, \bigg[\ell(x)\|u\|_{\plainL\infty(B(x, R\ell(x)))}
+ \|\nabla u\|_{\plainL\infty(B(x, R\ell(x)))} \notag\\
&\ \qquad+ \ell(x)\sum_{\bold0< \bq\le \bm}
\ell(x)^{|\bq|} \, \|\p_z^\bq g\|_{\plainL\infty(B(x, R\ell(x)))}\bigg].
\end{align}
If $|\bm|\ge 2$, then also  
\begin{align}\label{eq:d2elldist}
|\p_z^{\bm} u(x)|
\lesssim 
\ell(x)^{2-|\bm|}&\,\bigg[
\max_{\bl\le\bk: |\bl|=2} \|\p_z^{\bl} u\|_{\plainL\infty(B(x, R\ell(x)))}
+ \|\nabla u\|_{\plainL\infty(B(x, R\ell(x)))}\notag\\
&\ \qquad+  \|u\|_{\plainL\infty(B(x, R\ell(x)))} 
+ \sum_{\bold0 < \bq\le \bm}
\ell(x)^{|\bq|} \, \|\p_z^\bq g\|_{\plainL\infty(B(x, R\ell(x)))}\bigg].
\end{align}
The implicit constants in \eqref{eq:reg2dist} and \eqref{eq:d2elldist} may depend 
on the constants $R$, order $\bm$ of the derivative, and the 
constants in \eqref{eq:cdscale}.   
\end{thm}
 
\begin{proof} 
As the function $\dist(x)$, the capped distance $\ell(x)$ is also Lipschitz:  
$|\ell(x)-\ell(y)|\le |x-y|,\, x, y\in\Om$. Fix a point $x_0\in\Om$ and denote 
$\ell_0 := \ell(x_0)$. Therefore, for every $x\in B(x_0, R\ell_0)$ we have 
\begin{align*}
(1-R)\ell_0\le \ell(x)\le (1+R)\ell_0,
\end{align*}
i.e. $\ell(x)\asymp \ell_0$, $x\in B(x_0, R\ell_0)$, 
so that $\ba$ and $b$ satisfy the bound \eqref{eq:cd} with $\ell$ replaced with 
$\ell_0$. 
Use Theorem \ref{thm:reg2} for the ball 
$B(x_0, R\ell_0)$ and with $\nu = R/2, r = R/4$. Then 
the bound \eqref{eq:reg2} yields
\begin{align*}
|\nabla^k \p_z^{\bm} u(x_0)|\le &\ 
\|\nabla^k \p_z^{\bm} u\|_{\plainL\infty(B(x_0, R\ell_0/4))}\\
\lesssim &\ \ell_0^{1-|\bm|-k}\, \bigg[\ell_0\|u\|_{\plainL\infty(B(x_0, R\ell_0/2))}
+ \|\nabla u\|_{\plainL\infty(B(x_0, R\ell_0/2))} \notag\\
&\ \qquad+ \ell_0\sum_{\bold0< \bq\le \bm}
\ell_0^{|\bq|} \, \|\p_z^\bq g\|_{\plainL\infty(B(x_0, R\ell_0))}\bigg].
\end{align*}
This immediately entails the estimate \eqref{eq:reg2dist}. 

The bound \eqref{eq:d2elldist} is derived from 
\eqref{eq:d2ell} in the same way.    
\end{proof} 

\section{Regularity of $\psi$}\label{sect:regpsi}
 
\subsection{Bounds for first and second derivatives}  
Now we can apply the results of the previous section to the study of the 
function $\psi$. 
To get more precise results on the regularity of $\psi$ we use 
the representation  $\psi = e^F \phi$ with a 
function $F$ that satisfies 
$F, \nabla F\in \plainL\infty_{ {\rm loc}}(\R^{3N})$ and 
that captures the main singularity of the solution, so that 
$\phi$ is more regular than $\psi$. 
In the mathematical physics literature the function $e^F$ is often called 
a \textit{Jastrow factor}, see e.g. \cite{HaKlKoTe2012}.  
After the substitution the equation \eqref{eq:eigen} rewrites as 
\begin{align}\label{eq:jastrow}
-\Delta \phi - 2\nabla F\cdot \nabla\phi + (V- \Delta F - |\nabla F|^2-E)\phi = 0. 
\end{align}
More precisely, if $\psi\in \plainW{1, 2}_{\tiny {\rm loc}}(\R^{3N})$ is a weak 
solution of \eqref{eq:eigen}, 
then we also have $\phi\in \plainW{1, 2}_{\tiny {\rm loc}}(\R^{3N})$ 
and $\phi$ is a weak solution of \eqref{eq:jastrow}.  
 If one chooses $F$ to be 
\begin{align}\label{eq:oldF}
\sum_{j=1}^N \bigg(-\frac{Z}{2}|x_j| 
+ \frac{1}{4}\sum_{j<k\le N}|x_j-x_k|\bigg),
\end{align} 
as in \cite{Leray1984}, 
then $\nabla F$ is bounded,  
$\Delta F = V$, and hence all the coefficients in \eqref{eq:jastrow} are uniformly bounded. 
The  function $F$ itself however is not bounded. To remedy this, 
we take the following regularized variant of \eqref{eq:oldF}. 
Let the functions $\t, \z$ be as specified in \eqref{eq:sco} and \eqref{eq:sco1}.
From now on we assume that $F$ is given by 
\begin{align}\label{eq:jascut}
F(\bx) = -\frac{Z}{2}\sum_{j=1}^N |x_j|\t(|x_j|) 
+ \frac{1}{4}\sum_{1\le j<k\le N}|x_j-x_k| \t(|x_j-x_k|), 
\end{align} 
so that 
\begin{align}\label{eq:Finf}
F,\, \nabla F,\, \nabla^k (V-\Delta F)\in \plainL\infty(\R^{3N}),\quad \textup{for all}
\quad k = 0, 1, \dots.
\end{align} 
A different regularization (still satisfying \eqref{eq:Finf}) of the Jastrow factor was used in 
\cite{HOS2001}, \cite{FHOS2002}, \cite{FHOS2004} and \cite{FS2021}. 

Using the factorization 
\begin{align}\label{eq:phidef}
\psi(\bx) = e^{F(\bx)} \phi(\bx),
\end{align}
 it was shown in \cite[Theorem 1.2]{HOS2001} 
that for any positive $r, R$ such that $r < R$ the bound holds:
\begin{align*}
\|\nabla\psi\|_{\plainL\infty(B(\bx_0, r))} \lesssim \|\psi\|_{\plainL\infty(B(\bx_0, R))}. 
\end{align*}
The next fact  was proved in 
\cite[Proposition A.2]{FS2021}, see \cite[Lemma 4.1]{HearnSob2023} for a shorter proof: 

\begin{prop} \label{prop:grpsi}
Let $\psi\in \plainW{1, 2}_{\tiny {\rm loc}}(\R^{3N})$ be a weak solution of \eqref{eq:eigen}. 
Then $\psi \in \plainW{1,\infty}_{\tiny {\rm loc}}(\R^{3N})$ and for all numbers 
$r>0$ and $R>0$ such that $r < R$ we have  
\begin{align}\label{eq:grphi}
\|\phi\|_{\plainL\infty(B(\bx_0, r))} + \|\nabla\phi\|_{\plainL\infty(B(\bx_0, r))} 
\lesssim \|\psi\|_{\plainL2(B(\bx_0, R))},
\end{align}
\begin{align}\label{eq:grpsi}
\|\psi\|_{\plainL\infty(B(\bx_0, r))} + \|\nabla\psi\|_{\plainL\infty(B(\bx_0, r))} 
\lesssim \|\psi\|_{\plainL2(B(\bx_0, R))},
\end{align}
where the implicit constant does not depend on $\psi$ and $\bx_0\in \R^{3N}$, 
but depends on $r$ and $R$. 
\end{prop} 
 
As was shown in \cite{HOS2001} (see also Proposition \ref{prop:reg}), 
$\phi \in \plainC{1,\mu}_{\tiny{\rm loc}}(\R^{3N})$ for 
all $\mu\in (0,1)$.  
This conclusion was improved in \cite{FHOS2005} by extracting yet another Jastrow factor.  Precisely, let 
\begin{align*}
F_3(\bx) = \frac{2-\pi}{12\pi} Z \sum_{1\le j < k\le N} 
x_j\cdot x_k \ln\big(|x_j|^2 + |x_k|^2\big) \t(|x_j|)\t(|x_k|). 
\end{align*}
The function $F_3$ belongs to $\plainW{2,p}(\R^{3N})$ for all 
$p < \infty$, but not for $p=\infty$. However, $F_3$ is $\plainC\infty$ on the 
domain
\begin{align}\label{eq:away}
P_r =\{\bx\in\R^{3N}:  |x_j|>2r, \quad \textup{for all}\quad j = 1, 2, \dots, N\},
\end{align}
for any $r >0$.
We use the following regularity result established in \cite[Remark 1.6]{FHOS2005}:

\begin{prop}\label{prop:f3}
The function 
\begin{align*}
\tilde\phi(\bx) := e^{-F_3(\bx)} \, \phi(\bx) =  e^{-F(\bx) - F_3(\bx)} \psi(\bx)
\end{align*}
belongs to $\plainW{2, \infty}(\R^{3N})$ and for each pair $r, R$, $0 < r < R$, and each $\bx_0\in\R^{3N}$ we have the bound  
\begin{align*}
\|\tilde\phi\|_{\plainW{2, \infty}(B(\bx_0, r))}\lesssim \| \psi\|_{\plainL2(B(\bx_0, R))}, 
\end{align*}
with a constant independent of $\psi$ and $\bx_0$. 
\end{prop}

One cannot claim that 
$\phi\in \plainW{2, \infty}$.  
Nevertheless, the following, slightly weaker fact is true.  

\begin{lem} \label{lem:f3}
For each $j = 1, 2, \dots, N$, the derivative $\p_{x_j}^m \phi(\bx)$, $|m|=2$, 
$m\in \mathbb N_0^3$,   
is globally bounded, and 
for each pair $r, R$, $0 < r < R$, and every $\bx \in\R^{3N}$ we have the bound  
\begin{align}\label{eq:nabla2phi}
\|\p_{x_j}^m\phi\|_{\plainL\infty(B(\bx, r))}\lesssim \|\psi\|_{\plainL2(B(\bx, R))}, 
\end{align}
with a constant independent of $\psi$ and $\bx$. 

Furthermore, suppose that $\bx\in P_r$.  
Then for all $\bm\in \mathbb N_0^{3N}$ such that $|\bm|=2$,  
we have 
\begin{align}\label{eq:deraway}
\|\p_{\bx}^{\bm}\, \phi\|_{\plainL\infty(B(\bx, r))}
\lesssim \|\psi\|_{\plainL2(B(\bx, R))}, 
\end{align}
where the constant is independent of $\psi$ and $\bx\in P_r$.
\end{lem}

\begin{proof} 
By inspection, derivatives of the function $F_3(\bx)$ with respect to $x_j$ of orders 
one and two are uniformly bounded. Therefore, by Proposition \ref{prop:f3},  
$\p_{x_j}^m \phi = 
\p_{x_j}^m\big(e^{F_3}\tilde\phi\big)\in \plainL{\infty}(\R^{3N})$, as claimed, and   
the bound \eqref{eq:nabla2phi} holds.  

Under the assumption $\bx\in P_r$  
the derivatives $\p_{\bx}^{\bm}\,F_3(\bx)$, 
with $|\bm|\le 2$ are bounded by 
$|\ln r|\wedge 1$. Now the bound \eqref{eq:deraway} 
for the function $\phi = e^{F_3}\tilde\phi$  
follows from Proposition \ref{prop:f3} again.
\end{proof}

\subsection{Higher order derivatives} \label{subsect:hod}
Relying on the above estimates and on the results of Sect. \ref{sect:reg}, now we shall 
establish some bounds for higher order derivatives. 
Originally, such bounds were obtained 
by S. Fournais and T.\O. S\o rensen in \cite{FS2021}. 

Introduce the coalescence sets
\begin{align}\label{eq:sigmaj}
\Sigma_j = \bigg\{\bx\in\R^{3N}: |x_j| 
\prod_{k\not = j}|x_k-x_j| = 0\bigg\},\quad 
\Pi = \bigcup_{j=1}^N\, \Sigma_j.
\end{align}
Then 
\begin{align}\label{eq:dq}
\dc_j(\bx) = \min\{|x_j|, \frac{1}{\sqrt 2}
|x_j-x_k|: \ &\ k\not = j\},\ j = 1, 2, \dots, N, \notag\\
&\  \textup{and} 
\quad \DC(\bx) = \min_{1\le j\le N} \, \dc_j(\bx), 
\end{align}
define the distances from $\bx$ to $\Sigma_j$ and to $\Pi$ respectively. As in Sect. \ref{sect:reg} 
introduce also the capped distances 
\begin{align}\label{eq:lj}
\l_j(\bx) := \min\{1, \dc_j(\bx)\},\quad \L(\bx) = \min_{1\le j\le N}\, \l_j(\bx).  
\end{align}
Along with the  distances $\l_j(\bx)$ and $\L(\bx)$, introduce also 
\begin{align}\label{eq:tildedist}
\tilde\l_j(\bx) = \min\{1, \frac{1}{\sqrt 2} \, |x_j-x_k|: \ k\not = j\},
\quad \tilde\L(\bx) = \min_{1\le j\le N}  \, \tilde\l_j(\bx),
\end{align}
so that $\l_j(\bx) = \tilde\l_j(\bx)\wedge |x_j|$.
From now on we use the following notation:
\begin{align}\label{eq:locl2}
M(\bx, R) = M(\bx, R; \psi) = \|\psi\|_{\plainL2(B(\bx, R))},\quad \bx\in \R^{3N},\ R>0.
\end{align}
The following proposition is a consequence of \cite[Theorem 1.1]{FS2021} and Proposition~\ref{prop:grpsi}.  

\begin{prop}\label{prop:FS} 
For all $\bx\notin\Sigma_j$, $j = 1, 2, \dots, N$,  and any $R>0$ 
we have the bound 
\begin{align}\label{eq:FS}
|\p_{x_j}^{m}\psi(\bx)|\lesssim \big(1+\l_j(\bx)^{1-|m|}\big)\,  M(\bx, R; \psi), 
\, \quad \textup{for all}\quad m\in\mathbb N_0^{3},
\end{align}
with constants independent of $\psi$.  

Furthermore, for all $\bx\notin \Pi$, and any 
$R>0$ we have the bound
\begin{align*}
|\p_{\bx}^{\bm}\psi(\bx)|
\lesssim \big(1+ \L(\bx)^{1-|\bm|}\big)\,  M(\bx, R; \psi), 
\, \quad \textup{for all}\quad \bm\in\mathbb N_0^{3N},
\end{align*}
with constants independent of $\psi$.  
\end{prop}
 
We need however, a more precise, new bound for the function $\phi$ defined in \eqref{eq:phidef}.   

\begin{thm}\label{thm:regphi} 
For all 
$\bx\notin\Sigma_j$, $j = 1, 2, \dots, N$, and any $R>0$ 
we have the bound 
\begin{align}\label{eq:regphi}
|\p_{x_j}^m \,\phi(\bx)|
\lesssim \big(1+ \l_j(\bx)^{2-|m|}\big) M(\bx, R; \psi),\quad m\in\mathbb N_0^{3},
\end{align}
with a constant independent of $\phi$, but dependent on $R$. 

Furthermore, let $\bx\in P_r\setminus\Pi$ with an arbitrary $r>0$.  
Then  
\begin{align}\label{eq:secondmix}
|\p_{\bx}^{\bm}\phi(\bx)|
\lesssim \big(1+ \tilde\L(\bx)^{2-|\bm|}\big)\,  M(\bx, R; \psi), 
\, \quad \textup{for all}\quad \bm\in\mathbb N_0^{3N}.
\end{align}
The constants in \eqref{eq:secondmix} are 
independent of $\phi$, but dependent on $r$ and $R$.  
\end{thm}

The main tool in the proof is Theorem 
\ref{thm:dist} applied to equation \eqref{eq:jastrow} (with the function 
\eqref{eq:jascut}), which is satisfied by the function $\phi$. In order to check that 
the coefficients in the equation \eqref{eq:jastrow} satisfy 
the conditions of Theorem \ref{thm:dist}, we start with the following property of 
the Jastrow factor.

\begin{lem}\label{lem:nablajas}
Let $F$ be the function \eqref{eq:jascut}. Then for all 
$\bx\notin\Sigma_j$ and all 
$m\in\mathbb N_0^3$ we have
\begin{align}\label{eq:nablajas}
\big|\p_{x_j}^m \nabla F(\bx)\big| 
+ \big|\p_{x_j}^m |\nabla F(\bx)|^2\big|\lesssim \dc_j(\bx)^{-|m|}.
\end{align}
If $|m|\ge 1$, then  
\begin{align}\label{eq:ef}
\big|\p_{x_j}^m e^{F(\bx)}\big| \lesssim  \dc_j(\bx)^{1-|m|}.
\end{align} 
Furthermore, for all $\bm\in \mathbb N_0^{3N}$ and all $\bx\notin\Pi$, we have 
\begin{align}\label{eq:nablajas1}
\big|\p_{\bx}^\bm \nabla F(\bx)\big| 
+ \big|\p_{\bx}^\bm |\nabla F(\bx)|^2\big|\lesssim \DC(\bx)^{-|\bm|}.
\end{align}
If $|\bm|\ge 1$, then  
\begin{align}\label{eq:ef1}
\big|\p_{\bx}^\bm e^{F(\bx)}\big| \lesssim  \DC(\bx)^{1-|\bm|}.
\end{align} 

\end{lem}

\begin{proof} We start with the following elementary observation. 
Let $g \in\plainC\infty(\R^3\setminus \{0\})$ be such that 
\begin{align*}
|\p^m g(x)|\lesssim |x|^{-|m|}, \quad \textup{for all}\quad m\in \mathbb N_0^3. 
\end{align*}  
Then for $j, k = 1, 2, \dots, N$, $j\not = k$, we have 
\begin{align}\label{eq:ediff}
|\p_{x_j}^m g(x_j)| + |\p_{x_j}^m g(x_j-x_k)|\lesssim \dc_j(\bx)^{-|m|},
\end{align}
for all $\bx\notin \Sigma_j$.
Taking the function $g(x) = \nabla \big(|x| \t(x)\big)$, we obtain from \eqref{eq:ediff} 
the required bound \eqref{eq:nablajas} 
for $\nabla F$. 
For $|\nabla F|^2$ we use the Leibniz rule, which leads again to \eqref{eq:nablajas}.

To prove the bound \eqref{eq:ef} observe that the derivative $\p_{x_j}^m e^F, |m|\ge1,$ 
can be written as a sum of finitely many terms of the form 
\begin{align}\label{eq:dj}
 \p_{x_j}^{k_1}F \, 
 \p_{x_j}^{k_2}F \cdots 
 \p_{x_j}^{k_s}F \, e^{F},
\end{align}
where $1 \le |k_j|$, $\sum_{j=1}^s |k_j| = |m|$, $1\le s\le |m|$.
By \eqref{eq:nablajas} each such term can be estimated by 
\begin{align*}
\dc_j(\bx)^{(1-|k_1|) + (1-|k_2|) + \cdots (1-|k_s|) } 
= \dc_j(\bx)^{s - |m|}.
\end{align*}
Due to the presence of the cut-off $\t$ in the definition \eqref{eq:jascut}, the derivative 
$\p_{x_j}^m e^F$ is supported on the set where $\dc_j(\bx)\lesssim 1$. 
Therefore \eqref{eq:dj} 
does not exceed $\dc_j(\bx)^{1-|m|}$, as claimed.

The proof of \eqref{eq:nablajas1} and \eqref{eq:ef1} requires only obvious modifications. 
\end{proof}

\begin{proof}[Proof of Theorem \ref{thm:regphi}]
Since $\phi$ satisfies the equation \eqref{eq:jastrow}, we can apply Theorem \ref{thm:dist} 
with 
\begin{align*}
\ba = -2\nabla F,\quad b = V-\Delta F - |\nabla F|^2 - E,\quad g = 0,
\end{align*}
$z = x_j$ and the domain $\Om = \R^{3N}\setminus \Sigma_j$. 
Then by \eqref{eq:Finf} and \eqref{eq:nablajas}, the coefficients in \eqref{eq:jastrow} satisfy  
\eqref{eq:cdscale} 
for all $\bm = m\in\mathbb N_0^3$,  
with $\ell(\bx) = \l_j(\bx) = \min\{1, \dc_j(\bx)\}$ for all $\bx\in \Om$.
For $0\le |m|\le 1$ the bound \eqref{eq:regphi} holds due to Proposition \ref{prop:grpsi}. For 
$|m|\ge 2$ we use  \eqref{eq:d2elldist} with $z=x_j$ to obtain for all $\bx\notin\Sigma_j$ and arbitrary $R < 1$ that 
\begin{align*}
|\p_{x_j}^m \phi(\bx)|\lesssim \l_j(\bx)^{2-|m|} 
\big(&  \max_{l:|l|=2}\, \|\p_{x_j}^l \phi\|_{\plainL\infty(B(\bx, \l_j(\bx)R/2))} \\
&\ +  \|\nabla \phi\|_{\plainL\infty(B(\bx, \l_j(\bx) R/2))} 
+ \|\phi\|_{\plainL\infty(B(\bx, \l_j(\bx) R/2))} 
  \big), 
\end{align*}
with an implicit constant depending on $R$. 
By \eqref{eq:grphi} and \eqref{eq:nabla2phi} the bracket on the 
right-hand side is bounded by $M(\bx, R; \psi)$, as claimed.

Suppose now that $\bx\in P_r$. 
Now we apply Theorem \ref{thm:dist} again, but choose the domain $\Om$ differently: 
$\Om = \R^{3N}\setminus\Pi$. 
Then by \eqref{eq:Finf} and \eqref{eq:nablajas1} again, 
the coefficients in \eqref{eq:jastrow} satisfy \eqref{eq:cdscale} with 
$\ell(\bx) = \L(\bx)$ for all $\bx\in \Om$.
Using \eqref{eq:d2elldist} with $z = \bx$ and arbitrary $R<1$ we obtain that 
\begin{align*}
|\p_\bx^\bm \phi(\bx)|\lesssim \L(\bx)^{2-|\bm|} 
\big(&  \max_{\bl:|\bl|=2}\, \|\p_{\bx}^\bl \phi\|_{\plainL\infty(B(\bx, \L(\bx)R/2))} \\
&\ +  \|\nabla \phi\|_{\plainL\infty(B(\bx, \L(\bx) R/2))} 
+ \|\phi\|_{\plainL\infty(B(\bx, \L(\bx) R/2))} 
  \big), 
\end{align*}
for all $\bx\in \Om$.  
By \eqref{eq:grphi} the last two terms on the right-hand side 
are bounded by $M(\bx, R; \psi)$. 
As far as the first term is concerned, observe that for $\bx\in P_r$ and arbitrary  
$\by = (y_1, y_2, \dots, y_N)\in B(\bx, \L(\bx)R/2)$, for any $k = 1, 2, \dots, N$,  we have 
\begin{align*}
|y_k|\ge |x_k|-|\bx-\by|> |x_k| - 
\L(\bx) R/2\ge |x_k|\big(1 - R/2\big)>r/2.
\end{align*}
Therefore, by \eqref{eq:deraway}, this first term is also bounded from above by $M(\bx, R; \psi)$. 
Recall that $\L(\bx)\asymp\tilde\L(\bx)$ for $\bx\in P_r$, 
with constants depending on $r$. Now \eqref{eq:secondmix} follows, which completes 
the proof of Theorem \ref{thm:regphi}. 
\end{proof}

\subsection{Real analyticity of $\psi(\bx)$} \label{subsect:ra}
Here we briefly mention some further regularity results that have no direct bearing 
on our line of reasoning but have been used in \cite{Sobolev2022} and \cite{Sobolev2022a} 
to derive the asymptotics \eqref{eq:bolg}. 

Along with the standard 
notation $\bx = (x_1, x_2, \dots, x_N)\in\R^{3N}$ 
below we use the notation $x_0 = 0$. 
Thus, unless otherwise stated, the indices labelling the particles run from $0$ to $N$. 
For each pair $j, k: j\not = k$,  
denote 
\begin{align*}
\SfS_{j,k} =  
\bigcap_{\substack{l \not = s\\
(l, s)\not = (j, k)}}  
\{\bx\in\R^{3N}: x_l\not = x_s\}.
\end{align*}
In words, the set $\SfS_{j, k}$ includes 
the coalescence point $x_j = x_k$, but excludes all the other coalescence points. 
Our main focus will be on the properties of the function $\psi$ 
near the ``diagonal" set 
\begin{align*}
\SfS^{(\rm d)}_{j,k} = \{\bx\in\SfS_{j,k}: x_j = x_k\}.
\end{align*}
The sets introduced above are obviously symmetric with respect to permutations of 
indices, e.g. $\SfS_{j,k} = \SfS_{k,j}$, $\SfS^{(\rm d)}_{j, k}=\SfS^{(\rm d)}_{k,j}$.
Observe also that the sets $\SfS_{j,k}$, 
$\SfS_{j,k}^{(\dc)}$ are of full measure in $\R^{3N}$ and $\R^{3N-3}$ respectively, 
and that they are connected.   

The following property is a consequence of \cite[Theorem 1.4]{FHOS2009}.

\begin{prop}\label{prop:repr}
For each pair of indices $j, k = 0, 1, \dots, N$ such that $j\not = k$,  there exists 
an open connected set $\Om_{j,k} = \Om_{k, j}\subset \R^{3N}$, such that 
\begin{align*}
\SfS_{j,k}^{(d)}\subset\Om_{j,k}\subset \SfS_{j,k},
\end{align*}
and two uniquely defined functions $\xi_{j, k}, \eta_{j, k}$, real analytic on $\Om_{j, k}$, such that 
for all $\bx\in \Om_{j, k}$  
the following representation holds:
\begin{align}\label{eq:repr}
 \psi(\bx) = \xi_{j, k}(\bx) + |x_j-x_k| \eta_{j, k}(\bx).
\end{align} 
\end{prop} 

Due to the uniqueness of functions $\xi_{j, k}, \eta_{j, k}$, we have the symmetry 
$\xi_{j, k} = \xi_{k, j}$, $\eta_{j, k} = \eta_{k, j}$ for all $j\not = k$
 
Note that the representation \eqref{eq:repr} pinpoints the main singularity (i.e. $|x_j-x_k|$) 
at the pair coalescence point, which allows one to use the asymptotic results for homogeneous kernels, 
discussed in the Introduction. The formula \eqref{eq:repr} played a central part in 
\cite{Sobolev2022} and \cite{Sobolev2022a}, and led to the asymptotic relations
\eqref{eq:bolg} with coefficients $A$ and $B$ as some combinations of integrals of the 
functions $\eta_{j, k}$. The advantage of the expressions \eqref{eq:coeffAB} is in their explicit dependence on the function $\psi$.

\section{Regularity for totally antisymmetric functions}\label{sect:regasym}
Assume now that the function $\psi(\bx)$ is totally antisymmetric, i.e.~it 
satisfies \eqref{eq:antisym}. As we shall see below, the function $\psi$ possesses better smoothness 
at the coalescence points $x_j = x_k$, compared to \eqref{eq:FS}.   
As both quantities of interest, $\gamma(x, y)$ and $\vark(x, y)$ 
are defined now by the simplified formulas \eqref{eq:fb}, we will be concerned with the 
regularity of the function $\psi$ in variable $x: = x_N$ only. In this section we write 
\begin{align*}
\Sigma: = \Sigma_N,\quad \dc(\bx) := \dc_N(\bx),\quad \l(\bx) := \l_N(\bx),
\end{align*}
see \eqref{eq:sigmaj}, \eqref{eq:dq} and \eqref{eq:lj} for the definitions of 
$\Sigma_j$, $\dc_j$ and $\l_j$.

\subsection{Enhanced smoothness: a general result} 
Represent the function \eqref{eq:jascut} in the form
\begin{align}\label{eq:f1}
\begin{lcases}
F(\bx) = &\ F_1(x) + F_2(\bx),\\
F_1(x) = &\ - \frac{Z}{2} |x|\t(x),\\
F_2(\bx) = &\  - \frac{Z}{2}\sum_{l=1}^{N-1} |x_l|\t(x_l) + 
\frac{1}{4}\sum_{1\le l<k\le N}|x_l-x_k| \t(x_l-x_k),
\end{lcases}
\end{align} 
and define 
\begin{align}\label{eq:tildepsi}
\tilde\psi(\bx) = e^{-F_1(x)}\psi(\bx) = e^{F_2(\bx)}\phi(\bx),
\end{align}
see \eqref{eq:phidef} for the definition of $\phi$. 
Thus in $\tilde\psi(\hat\bx, x)$ 
all singularities of $\psi(\hat\bx, x)$ are left untouched apart from the one at $x = 0$.
Let $M(\bx, R; \psi)$ be as defined in \eqref{eq:locl2}. 
Due to the condition \eqref{eq:antisym} the function $\tilde\psi$ 
possesses an enhanced smoothness. The required bound is proved in the next theorem. 
Recall the notation \eqref{eq:away} for the set $P_r, r>0$.

\begin{thm}\label{thm:varphider} 
Suppose that \eqref{eq:antisym} is satisfied. Then 
for all 
 $\bx\notin\Sigma$ and for any $R>0$ we have 
\begin{align}\label{eq:varphider}
|\p_{x}^m \tilde\psi(\bx)|
\lesssim (1+\l(\bx)^{2-|m|})\, M(\bx, R; \psi),
\end{align}
for all $m\in\mathbb N_0^3$, where the constants potentially depend on $R$. 

Let $r>0$ be arbitrary.
Then $\psi\in \plainC{{1},{1}}(P_r)$ and for all $\bx\in P_r\setminus \Pi$ we have 
\begin{align}\label{eq:varphider1}
|\p_{\bx}^{\bm} \psi(\bx)|
\lesssim (1+\tilde\L(\bx)^{2-|m|})\, M(\bx, R; \psi),
\end{align}
for all $\bm\in\mathbb N_0^{3N}$, where the constants potentially depend on $r$ and $R$. 
\end{thm} 
 
\begin{proof}  Recall the notation \eqref{eq:tildedist}:
\begin{align*}
\tilde\l(\bx): = \tilde\l_N(\bx)
= \min\{1, \frac{1}{\sqrt 2}
|x-x_k|: \ 1\le k\le N-1\},
\end{align*}
so that $\l(\bx) = \tilde\l (\bx)\wedge|x|$. 
Arguing as in the proof of Lemma \ref{lem:nablajas}, we obtain that 
\begin{align}\label{eq:expf2}
|\p_x^k e^{F_2(\bx)}|\lesssim \tilde\l(\bx)^{1-|k|}\le \l(\bx)^{1-|k|},
\ |k|\ge 1.
\end{align} 
 Now we can estimate the derivatives of the function $\tilde\psi$ from \eqref{eq:tildepsi}. 
If $m = 0$, then by \eqref{eq:grphi},
\begin{align*}
|\tilde\psi(\bx)|\lesssim |\psi(\bx)|\lesssim  M(\bx, R; \psi),
\end{align*} 
 i.e. \eqref{eq:varphider} holds. If $|m|=1$, then by \eqref{eq:grphi} 
 and \eqref{eq:expf2}, 
 \begin{align*}
 |\p_x^m \bigl(e^{F_2(\bx)} \phi(\bx)\bigr)|\lesssim |e^{F_2(\bx)} \p_x^m \phi(\bx)|
 + |(\p_x^m e^{F_2(\bx)}) \phi(\bx)|\lesssim M(\bx, R; \psi),
 \end{align*}
 again. 
 
Suppose that $|m|\ge 2$. By the Leibniz rule,  
\begin{align}\label{eq:bep}
|\p_x^m \big(e^{F_2(\bx)} \phi(\bx)\big)|
\lesssim &\ \sum_{0\le k\le m}|\p_x^k e^{F_2(\bx)}|\,  |\p_x^{m-k} \phi(\bx)|\notag\\
= &\ 
\sum_{0 < k< m}|\p_x^k e^{F_2(\bx)}|\,  |\p_x^{m-k} \phi(\bx)| 
+ e^{F_2(\bx)}\,  |\p_x^{m} \phi(\bx)| + |\p_x^m e^{F_2(\bx)}| |\phi(\bx)|.
\end{align}
Due to Theorem \ref{thm:regphi} and the bound \eqref{eq:expf2}, 
the first term is bounded by  
\begin{align*}
 M(\bx, R; \psi)\,\sum_{0< k< m} \l(\bx)^{1-|k|} 
(1+\l(\bx)^{2-|m|+|k|}) 
\lesssim M(\bx, R; \psi)\,  
\l(\bx)^{2-|m|}.
\end{align*}
The second term in \eqref{eq:bep} satisfies the same bound due to \eqref{eq:regphi} again.

To estimate the last term in \eqref{eq:bep} we use the property \eqref{eq:antisym}. Suppose that 
$\tilde\l(\bx) < \frac{R}{4}\wedge 1$. 
Therefore, for any such $\bx$ we can find $k \in \{1, 2, \dots, N-1\},$ such that 
 $|x-x_k| = \sqrt2\, \tilde\l(\bx)$.
Therefore, 
\begin{align*}
|\phi(\bx)|  = &\ |\phi(\tilde\bx_{k, N}, x_k, x) - \phi(\tilde\bx_{k, N}, x_k, x_k)|\\
\lesssim &\ \|\nabla_x\phi\|_{\plainL\infty(B(\bx, R/2))} \tilde\l(\bx)
\lesssim \tilde\l(\bx)\, M(\bx, R; \psi),
\end{align*}
where we have used \eqref{eq:grphi} again. 
If $\tilde\l(\bx) \ge \frac{R}{4}\wedge 1$, then 
because of \eqref{eq:grphi} we get 
\begin{align*}
|\phi(\bx)|\lesssim M(\bx, R; \psi)
\lesssim \tilde\l(\bx)\, M(\bx, R; \psi),  
\end{align*}
as before. Together with \eqref{eq:expf2} this gives the estimate
\begin{align*}
|(\p_x^m e^{F_2(\bx)}) \phi(\bx)|
\lesssim \tilde\l(\bx)^{2-|m|} M(\bx,R; \psi)
\le \l(\bx)^{2-|m|} M(\bx,R; \psi),
\end{align*}
where we have taken into account that $|m|\ge 2$. 
Thus all the terms on the right-hand of 
\eqref{eq:bep} are bounded by $(1+ \l(\bx)^{2-|m|}) M(\bx,R; \psi)$ for all $m\in\mathbb N_0^3$, which completes the proof of \eqref{eq:varphider}.

The proof of \eqref{eq:varphider1} follows the same plan. Omitting the details we note that arguing as in the proof of Lemma \ref{lem:nablajas} again, 
under the assumption $\bx\in P_r\setminus\Pi$ we obtain that 
\begin{align*}
|\p_{\bx}^{\bk} e^{F(\bx)}|\lesssim \tilde\L(\bx)^{1-|\bk|},\ |\bk|\ge 1,
\end{align*} 
with constants depending on $r$. Using this together with \eqref{eq:secondmix}, and 
following the same steps for the function $\psi(\bx) = e^{F(\bx)}\, \phi(\bx)$, we arrive at \eqref{eq:varphider1}.
\end{proof} 
 
 \begin{rem}\label{rem:grad}
 It follows from Theorem \ref{thm:varphider} that under the condition 
\eqref{eq:antisym}  
 the gradient $\nabla_x \tilde\psi$ is continuous 
 on $\R^{3N}$. Therefore, 
\[
\sv(\hat\bx, x) = \nabla_x\psi(\hat\bx, x) = \nabla_x \big( e^{-F_1(x)}\tilde\psi(\hat\bx, x)\big)
\] 
is continuous on $\R^{3N}\setminus\{\bx = (\hat\bx, 0)\}$.  
 \end{rem}

\subsection{Two-particle coalescence point} 
\label{subsect:two}
As before we assume that the condition \eqref{eq:antisym} is satisfied. 
Our purpose in this section is to study the smoothness of $\psi$ in the neighbourhood 
of a two-particle coalescence point under the assumption that these two particles 
are separated from the other particles and from the nucleus.  
Assume that the two particles in question are particles $N$ and $k$, for some $1\le k\le N-1$. 
For convenience further on we denote 
$\tilde\bx = \tilde\bx_{k, N}$ and $y = x_k$, $x = x_N$, 
so that, in line with the notational conventions adopted in the Introduction, 
$\bx = (\tilde\bx, y, x)$. 
Let $\d>0$ and $\varepsilon>0$ be fixed, and define 
\begin{align}\label{eq:wed}
\SW_k(\varepsilon, \d) = 
\{ (\tilde\bx, y, x): |x_s-x_l|> \varepsilon, &\ \, 0\le s < l\le N-1,s\not = k, l\not =k, \notag\\ 
\bigg|x_j - \frac{x+y}{2}\bigg|> \varepsilon, 0\le &\ j\le N-1, j\not = k,
\    \textup{and}\, \, 
|x-y|<\d \},\\[0.5cm]
\SW^\circ_k(\varepsilon, \d) = \SW_k(\varepsilon, \d)\cap&\ \{\bx: x \neq y\}.\notag
\end{align}
Here we have denoted $x_0 = 0$.                                                                                                                                                                                        
Since $k$ is fixed, we usually 
write $\SW(\varepsilon, \d)$ instead of $\SW_k(\varepsilon, \d)$. 
Note that 
$\SW(\varepsilon_1, \d_1)\subset \SW(\varepsilon_2, \d_2)$ for any $\varepsilon_1 > \varepsilon_2$ 
and $\d_1 < \d_2$. 
The bounds that we obtain in what follows, will depend on the constants $\d$  and $\varepsilon$ 
but we do not reflect this in the notation. 

By definition, for any $\bx\in\SW(\varepsilon, \d)$  
the particles $s, l$ with $0\le s < l\le N-1$, $s\not = k, l\not = k$,  
are separated from each other. Furthermore, 
we automatically have that 
$|x-x_j|> \varepsilon - \d/2$ and $|y-x_j| > \varepsilon - \d/2$ for all 
$j = 0, 1, \dots, N-1, j\not = k$. 
We will always assume that $\varepsilon>\d/2$, so that 
for any $\bx\in\SW(\varepsilon, \d)$ the particles 
$k$ and $N$ are separated from the other ones and from the nucleus. Because of this, 
\begin{align*}
\l_N(\bx)\asymp \l_k(\bx)\asymp\L(\bx)\asymp\tilde\L(\bx)\asymp |x-y|.
\end{align*}
Thus it follows from \eqref{eq:varphider1} that 
\begin{align}\label{eq:psitwohigh}
|\p_{\bx}^{\bm} \psi(\tilde\bx, y, x)| 
\lesssim 1+ |x-y|^{2-|\bm|}, 
\quad \textup{for all}\quad \bx\in\SW^\circ(\varepsilon, \d),\, \bm\in \mathbb N_0^{3N},
\end{align}
with constants that may depend on $\varepsilon$ and $\d$. Note that 
we have estimated the factor $M(\bx, R; \psi)$ by the norm of $\psi$ which equals $1$, see 
\eqref{eq:normalized}. 
In order to explore further the structure of the function $\psi$ 
on the domain $\SW(\varepsilon, \d)$ 
we introduce a natural orthogonal change of variables:
\begin{align}\label{eq:coord}
z = \frac{1}{\sqrt 2}(x+y),\quad  t = \frac{1}{\sqrt 2}(x-y),
\end{align}
and write $\bx = (\tilde\bx, z, t)$. 
To highlight the dependence on the new variables, we denote 
\begin{align}\label{eq:psichi}
\chi(\tilde\bx, z, t) = \psi\bigg(\tilde\bx, \frac{z-t}{\sqrt2}, \frac{z+t}{\sqrt2}\bigg).
\end{align}
In the new variables the definition \eqref{eq:wed} rewrites as 
\begin{align*}
\SW_k(\varepsilon, \d) = 
\{ (\tilde\bx, y, x): |x_s-x_l|> \varepsilon, &\ \, 0\le s < l\le N-1, s\not = k, l\not = k,\\ 
\bigg|x_j - \frac{z}{\sqrt 2}\bigg|> \varepsilon, 0\le &\ j\le N-1, j\not = k, 
\    \textup{and}\, \, 
\sqrt2\, |t|<\d \}.
\end{align*}  
The use of the same notation as in \eqref{eq:wed} will not cause confusion.  
 The antisymmetry condition 
\eqref{eq:antisym} rewrites as $\chi(\tilde \bx, z, t) = - \chi(\tilde \bx, z, -t)$.
Since $\Delta_{x, y} = \Delta_{z, t}$, the Schr\"odinger 
equation $-\Delta\psi + (V(\bx) - E)\psi = 0$ 
in the new variables rewrites as 
\begin{align*}
-\Delta \chi + \frac{1}{\sqrt2\,|t|} \chi + U \chi = 0,\quad  
\quad U(\tilde\bx, z,t)= V(\tilde\bx, z, t) - E - \frac{1}{\sqrt2\,|t|},
\end{align*}
so that $U$ is $\plainC\infty$  
with derivatives of all orders in $\plainL\infty(\SW(\varepsilon, \d))$. 
By \eqref{eq:psitwohigh}, 
\begin{align}\label{eq:psinew}
|\p_{\bx}^{\bm} \chi(\tilde\bx, z, t)| 
\lesssim 1+ |t|^{2-|\bm|}, 
\end{align}
for all $\bx\in \SW^\circ(\varepsilon, \d)$ and all $\bm\in \mathbb N_0^{3N}$.
 
Now, for an arbitrary $\a >0$  
we introduce a new family of Jastrow factors $K_\a$ and a new function $\xi_\a$:\footnote{We will see later in \eqref{eq:optimalityOfAlpha} that it is advantageous to choose $\alpha = \frac{1}{2}$, but we find it illustrative to keep the choice free up to that point.}
\begin{align}
\label{eq:psik} 
\chi(\tilde\bx, z, t) = e^{K_\a(t)} \xi_\a(\tilde\bx, z, t),
\quad K_\a(t) = \frac{\a}{2\sqrt2}\, |t|. 
\end{align} 
Arguing as in the proof of Theorem \ref{thm:varphider} 
we derive from \eqref{eq:psinew} the bound
 \begin{align}\label{eq:xinew}
|\p_{\bx}^{\bm} \xi_\a(\tilde\bx, z, t)| 
\lesssim 1+ |t|^{2-|\bm|},
\end{align}
for all $\bx\in \SW^\circ(\varepsilon, \d)$ 
and all $\bm\in \mathbb N_0^{3N}$. The advantage of introducing the 
new variables $z, t$ is that the function $\xi_\a$ exhibits very different regularity 
in $z$ and in $t$. Indeed, since 
\begin{align*}
\nabla K_\a(t) = \frac{\a}{2\sqrt2} \, \frac{t}{|t|},\quad 
\Delta K_\a (t) = \frac{\a}{\sqrt2\, |t|},
\end{align*}
the function $\xi_\a$ satisfies the equation 
\begin{align}
- \Delta \xi_\a -  \frac{\a}{\sqrt2 \, |t|} t\cdot \nabla_t \xi_\a 
+ \bigg(\frac{1-\a}{\sqrt2\, |t|} - \frac{\a^2}{8} + U\bigg) \xi_\a = 0. \label{eq:xieq}
\end{align}
The properties of $\xi_\a$ are summarized in the following lemma.

\begin{lem}\label{lem:xiz} 
For all $\a>0$ the functions $\xi_\a$ and $\nabla_\bx\xi_\a$ are 
$\plainC\infty$ in the variable $(\tilde\bx, z)$ and 
\begin{align}\label{eq:xiz}
|\p_{\tilde\bx}^\bn \p_z^m \xi_\a(\tilde\bx, z, t)| 
+ |\p_{\tilde\bx}^\bn \p_z^m \nabla_\bx \xi_\a(\tilde\bx, z, t)|\lesssim 1,
\quad m\in\mathbb N_0^3, \quad \bn\in \mathbb N_0^{3N-6},
\end{align}
 for all $\bx\in\SW(\varepsilon, \d)$.
\end{lem}
 
\begin{proof}  
Since $\xi_\a = e^{-K_\a + K_1} \xi_1$  and 
the Jastrow factor $K_\a$ does not depend on $z$ or $\tilde\bx$, 
it suffices to prove the required bounds for the value $\a = 1$. 

Observe that for $r = (\varepsilon-\d/2)/2$ the parameters 
$\varepsilon_1 = \varepsilon-r/\sqrt2$ and $\delta_1 = \d + r\sqrt2$ satisfy the 
relation $\varepsilon_1>\d_1/2$. 
One checks directly that for every $\bx\in\SW(\varepsilon, \d)$ 
the inclusion $B(\bx, r)\subset \SW(\varepsilon_1, \d_1)$ holds.
Now we use Theorem \ref{thm:modmod}. Since 
the coefficients in \eqref{eq:xieq} and 
their derivatives w.r.t.~$(\tilde\bx, z)$ of any order are uniformly bounded 
on $\SW^\circ(\varepsilon_1, \d_1)$, 
the condition \eqref{eq:highreg} is satisfied on $B(\bx, r), \bx\in\SW(\varepsilon, \d)$, 
for all $\bm = (\bn, m)$ where $\bn\in\mathbb N_0^{3N-6}$, $m\in\mathbb N_0^3$,  
with constants independent of $\bx$.  
Thus, by Theorem \ref{thm:modmod} (see \eqref{eq:dd0}), 
\begin{align*}
|\p_{\tilde\bx}^\bn \p_z^m  \xi_1(\tilde\bx, z, t)| 
+ |\p_{\tilde\bx}^\bn \p_z^m  \nabla_\bx \xi_1(\tilde\bx, z, t)| 
\lesssim  
\|\xi_1\|_{\plainL2(B(\bx, r))}
 \lesssim \|\psi\|_{\plainL2(B(\bx, r))}^2\lesssim 1, 
\end{align*}
for all $\bx\in \SW(\varepsilon, \d)$.
This proves the bound 
\eqref{eq:xiz} 
for $\xi_1$, and therefore, for $\xi_\a$ with an arbitrary 
$\a>0$. 
\end{proof} 

 \subsection{Detailed structure of $\xi_{\a}$}
Remembering that $\xi_\a(\tilde\bx, z, 0) = 0,$ we can write  
\begin{align}
\label{eq:xia}
\xi_\a(\tilde\bx, z, t)  = \xi^\circ_\a(\tilde\bx, z, t) + \eta_\a(\tilde\bx, z, t),
\end{align}
where 
$\xi^\circ_\a(\tilde\bx, z, t):= t\cdot \nabla_t\xi_\a(\tilde\bx, z, 0)$. 
Due to \eqref{eq:xinew},   
\begin{align}\label{eq:regeta}
|\p_{\bx}^\bm\,\eta_\a(\tilde\bx, z, t)|
\lesssim  1+ |t|^{2-|\bm|},
\quad \bm\in \mathbb N_0^{3N},
\end{align} 
for all $\SW^\circ(\varepsilon, \d)$.  
Here are some properties of $\eta$:

\begin{lem} \label{lem:eta}
Suppose that $\varepsilon>\d$. Then 
for all $\bx\in\SW(\varepsilon, \d)$ and any $R>0$ we have
\begin{align}\label{eq:etader}
|\eta_\a(\tilde\bx, z, t)|\lesssim |t|^2, \quad 
|\nabla_\bx \eta_\a(\tilde\bx, z, t)|\lesssim |t|,
\end{align} 
and for all $\bx\in\SW^\circ(\varepsilon, \d)$, 
\begin{align}\label{eq:regeta1}
|\p_{\bx}^\bm\, \eta_\a|\lesssim |t|^{2-|\bm|},\quad \bm\in\mathbb N_0^{3N}.
\end{align}
\end{lem}

\begin{proof}  
Let us fix a point $\bx = (\tilde\bx, z, t)\in \SW(\varepsilon, \d)$. 
An argument  
similar to the one in the proof of Lemma \ref{lem:xiz}, 
shows that 
\begin{align*}
B\big((\tilde\bx, z, 0), |t|\big)\subset B_\d:= B\big((\tilde\bx, z, 0), \d/\sqrt2\big)
\subset \SW(\varepsilon-\d/2, \d). 
\end{align*}
Since $\varepsilon_1:=\varepsilon - \d/2 >\d/2$, 
the function $\eta_\a$ satisfies \eqref{eq:regeta} on 
$B\big((\tilde\bx, z, 0), \d/\sqrt2\big)$. 
By the definition \eqref{eq:xia}, we have $\eta_\a(\tilde\bx, z, 0) = 0$ 
and $\nabla_\bx\eta_\a(\tilde\bx, z, 0) = 0$. Therefore, 
by the Fundamental Theorem of Calculus, and by \eqref{eq:regeta}, for all 
$h:|h|< \d/\sqrt2$, we have 
\begin{align*}
|\eta_\a(\tilde\bx, z, h)| = &\ \bigg|\int_0^1\,\int_0^{s_1} \, (h\cdot\nabla_t)^2 
\eta_\a(\tilde\bx, z, s_2 h)\, ds_2 \,ds_1\bigg|\\
\lesssim &\ \frac{|h|^2}{2} \, 
\max_{\bl:|\bl|=2}\|\p_\bx^\bl \eta_\a\|_{\plainL\infty(B_\d)}
\lesssim |h|^2, 
\end{align*}
which gives the first bound in \eqref{eq:etader}. By the same argument,
\begin{align*}
|\nabla_\bx \eta_\a(\tilde\bx, z, h)| 
= &\ \bigg|\int_0^1  \,  (h\cdot \nabla_t) \nabla_\bx \eta_\a(\tilde\bx, z, sh) \, ds\bigg|\\
\lesssim &\  |h| \max_{\bl:|\bl|=2}\|\p_\bx^\bl \eta_\a\|_{\plainL\infty(B((\tilde\bx, z, 0), |h|)}\lesssim |h|,
\end{align*}
which gives the second bound in \eqref{eq:etader}. 
The estimate \eqref{eq:regeta1} is a consequence of \eqref{eq:regeta} and \eqref{eq:etader}. 
\end{proof}

Substitute $\xi_\a = \xi^\circ_\a +\eta_\a$  into \eqref{eq:xieq}:
\begin{align*}
- \Delta (\xi^\circ_\a + \eta_\a) -  &\ \frac{\a}{\sqrt2 \, |t|} t\cdot \nabla_t \xi^\circ_\a 
+ \frac{1-\a}{\sqrt2\, |t|} \xi^\circ_\a +   \bigg(- \frac{\a^2}{8} + U\bigg) \xi^\circ_\a  \\
&\ -  \frac{\a}{\sqrt2 \, |t|} t\cdot \nabla_t \eta_\a 
 + \bigg(\frac{1-\a}{\sqrt2\, |t|} - \frac{\a^2}{8} +  U\bigg) \eta_\a
 = 0.
\end{align*}
The second and the third term together give
\begin{align}\label{eq:optimalityOfAlpha}
\frac{1-2\a}{\sqrt2 \, |t|}\,  t\cdot \nabla_t \xi_\a(\tilde\bx, z, 0). 
\end{align}
 This term vanishes if $\a = 1/2$. \textbf{ From now on we always assume that}
 \begin{align}
 \a = 1/2,
 \end{align}
\textbf{and omit $\a$ from the notation, i.e. we write $\eta := \eta_{1/2},\xi^\circ := \xi^\circ_{1/2}$}. 
 Then the above equation takes the form
\begin{align}\label{eq:eqetag}
-\Delta \eta = G,
\end{align}
where $G = G(\tilde\bx, z, t)$ is given by 
\begin{align}\label{eq:defG}
G = \Delta\xi^\circ + \bigg( \frac{1}{32} - U\bigg) \xi^\circ 
 +  \bigg(\frac{1}{2\sqrt2 \, |t|} t\cdot \nabla_t \eta - 
\frac{1}{2\sqrt2\, |t|}\eta\bigg)  + \bigg(\frac{1}{32} - U\bigg) \eta. 
\end{align}
We use the equation \eqref{eq:eqetag} to find 
estimates for $\eta$ that are more precise than \eqref{eq:etader}. To this end 
we shall use Theorem \ref{thm:dist}. In order to apply it, 
we first prove the following bound for $G$.

\begin{lem}
\label{lem:regG}
Let $\varepsilon>\d$, and let the functions  $\xi^\circ$, $\eta$ and $G$ be as defined above. 
Then $G$ is Lipschitz on $\SW(\varepsilon, \d)$, 
and
\begin{align}\label{eq:regG}
 |\p_{\bx}^\bm G(\tilde\bx, z, t)| \lesssim 1+|t|^{1-|\bm|},
\end{align}
for all $\bx\in\SW^\circ(\varepsilon, \d)$ and all $\bm\in\mathbb N_0^{3N}$. 
\end{lem}

\begin{proof} 
First we estimate the first two terms in \eqref{eq:defG}. By \eqref{eq:xiz}, 
\begin{align*}
\left|\p_{\bx}^{\bm} \biggl[\Delta\xi^\circ 
+ \bigg(\frac{1}{32} - U\bigg)\xi^{\circ}\biggr]\right|
\lesssim 1,\quad \bx\in \SW(\varepsilon, \d),
\end{align*}
for all $\bm\in \mathbb N_0^{3N}$. 

The terms containing $\eta$ are estimated with the help of \eqref{eq:regeta1}. 
Since $U$ is $\plainC\infty\cap\plainL\infty$ on $\SW(\varepsilon, \d)$, the derivatives of
the very last term in \eqref{eq:defG} are bounded by $|t|^{2-|\bm|}$. 
Estimate the components of the second bracket in \eqref{eq:defG}  for all $\bm$:
\begin{align*}
\bigg|\p_{\bx}^{\bm} \biggl(\frac{t}{|t|}\cdot\nabla_t \eta\biggr)\bigg|\lesssim &\ 
\sum_{\bq=\bold0}^\bm 
 \, \bigg|\p_{\bx}^{\bm-\bq} \frac{t}{|t|}\bigg| \, \big|\p_{\bx}^\bq\big(\nabla_t\eta\big)\big|\\
\lesssim  &\ 
\sum_{\bq=\bold0}^\bm \, |t|^{-|m|+|\bq|} \, |t|^{1-|\bq|} 
\lesssim |t|^{1-|\bm|}.
\end{align*} 
In the same way,
\begin{align*}
\bigg|\p_{\bx}^{\bm} \biggl(\frac{1}{|t|}\,\eta\biggr)\bigg|\lesssim &\ 
\sum_{\bq=\bold0}^\bm 
 \, \bigg|\p_{\bx}^{\bm-\bq} \frac{1}{|t|}\bigg| \, \big|\p_{\bx}^\bq\eta \big|\\
\lesssim  &\ 
\sum_{\bq=\bold0}^\bm \, |t|^{-|m|+|\bq|-1} \, |t|^{2-|\bq|} 
\lesssim |t|^{1-|\bm|}.
\end{align*} 
This completes the proof of \eqref{eq:regG}.
\end{proof}
 
Our objective is to prove the following theorem for the function $\eta$.
 
\begin{thm} 
Let $\varepsilon>\d$. 
Then for any $\mu\in [0, 1)$, and all $\bx\in \SW^\circ(\varepsilon, \d)$, we have 
\begin{align}\label{eq:ulteta}
|\p_\bx^\bm\,\eta(\tilde \bx, z, t)|\lesssim |t|^\mu \big(1+ |t|^{2-|\bm|}\big),
\quad \bm\in\mathbb N_0^{3N}, 
\end{align}  
with constants potentially depending on $\bm$ and $\mu$.
\end{thm}

\begin{proof} 
Denote 
\begin{align*}
r = \frac{\varepsilon-\d}{4}. 
\end{align*}
Arguing as in the proof of Lemma \ref{lem:xiz}, we conclude that 
the parameters 
$\varepsilon_1 = \varepsilon-r/\sqrt2$ and $\delta_1 = \d + r\sqrt2$ satisfy the 
relation $\varepsilon_1>\d_1$, and that for each $\bx\in\SW(\varepsilon, \d)$ 
the inclusion $B(\bx, r)\subset \SW(\varepsilon_1, \d_1)$ holds,

\underline{Step 1.} 
First we prove that $\eta\in \plainC{2, \mu}\big(\overline{B(\bx, r/2)}\big)$ 
for all $\mu\in [0, 1)$ 
and every $\bx\in \SW(\varepsilon, \d)$, and that 
\begin{align}\label{eq:etahol}
\|\eta\|_{\plainC{2, \mu}\big(\overline{B(\bx, r/2)}\big)}\lesssim 1.
\end{align}
To this end we use Theorem \ref{thm:modmod} with the balls 
$B(\bx, r/2)$ and $B(\bx, r)$. 
By Lemma~\ref{lem:regG} the function $G(\bx)$ and its first order derivatives 
are uniformly bounded for a.e. $\bx\in\SW(\varepsilon_1, \d_1)$.
Thus by Theorem  \ref{thm:modmod} (see \eqref{eq:dd0}) and by \eqref{eq:regeta},
\begin{align*}
\| \eta\|_{\plainC{2, \mu}(\overline{B(\bx, r/2)})}\lesssim 
\| \eta\|_{\plainL2(B(\bx, r))}  
+ \|G\|_{\plainL\infty(B(\bx, r))} 
+ \|\nabla G\|_{\plainL\infty(B(\bx, r))}
\lesssim 1,
\end{align*}
which is the bound \eqref{eq:etahol}. 

\underline{Step 2: proof of \eqref{eq:ulteta} for $|\bm|= 0, 1$.}
The bound \eqref{eq:ulteta} holds thanks to \eqref{eq:etader}. 

\underline{Step 3: proof of \eqref{eq:ulteta} for $|\bm|=2$.}
First we prove that 
\begin{align}\label{eq:zerod}
\p_\bx^{\bm} \eta(\tilde\bx, z, 0) = 0,\quad \bx\in \SW(\varepsilon, \d).
\end{align} 
If $\bm$ is such that at most one derivative falls on $t$, then \eqref{eq:zerod} 
holds because of the definition \eqref{eq:xia}. 
If both derivatives $\p_\bx^{\bm}$ fall on the variable $t$, then \eqref{eq:zerod} holds because of the antisymmetry \eqref{eq:antisym}. 
To prove \eqref{eq:ulteta} we assume without loss of generality that $|t| < r/2$. 
Now, in view of \eqref{eq:etahol} and \eqref{eq:zerod}, 
for $|\bm|= 2$ we have 
\begin{align*}
|\p_\bx^\bm \eta(\tilde\bx, z, t)|
= |\p_\bx^\bm \eta(\tilde\bx, z, t) - \p_\bx^\bm \eta(\tilde\bx, z, 0)|
\lesssim |t|^\mu,\quad \bx\in \SW(\varepsilon, \d).
\end{align*}
This proves \eqref{eq:ulteta} for $|\bm|=2$.

\underline{Step 4: proof of \eqref{eq:ulteta} for $|\bm|>2$.}
Let $\varepsilon_1, \d_1$ and $r$ be as defined in Step 1. 
Now we use Theorem \ref{thm:dist} with $\Om = \SW^\circ(\varepsilon_1, \d_1)$. 
Let $\dist_\Om(\bx)$ be as defined in \eqref{eq:dc} 
and let $\ell(\bx) = \min\{1, \dist_\Om(\bx)\}$.  
By \eqref{eq:d2elldist}, for any $\rho<1$ and all $\bx\in\Om$ we have 
\begin{align}\label{eq:w1}
|\p_\bx^{\bm} \eta(\bx)|
\lesssim 
\ell(\bx)^{2-|\bm|}&\,\bigg[
\max_{ |\bl|=2} \|\p_\bx^{\bl} \eta\|_{\plainL\infty(B(\bx, \rho\ell(\bx)))}
+ \|\nabla_\bx \eta\|_{\plainL\infty(B(\bx, \rho\ell(\bx)))}\notag\\
&\ \qquad+  \|\eta\|_{\plainL\infty(B(\bx, \rho\ell(\bx)))} 
+ \sum_{0<\bq\le\bm} 
\ell^{|\bq|} \|\p_\bx^{\bq}\, G\|_{\plainL\infty(B(\bx, \rho\ell(\bx)))}\bigg].
\end{align}
We do not use this estimate for all $\bx\in\Om$, but only for $\bx\in\SW^\circ(\varepsilon, \d)$. 
Under this condition we have $\rho_1|t| \le \ell(\bx)\le |t|$ with some constant $\rho_1<1 $ depending on 
$\varepsilon, \varepsilon_1, \d, \d_1$, and hence \eqref{eq:w1} rewrites as 
\begin{align*}
|\p_\bx^{\bm} \eta(\bx)|
\lesssim 
|t|^{2-|\bm|}&\,\bigg[
\max_{ |\bl|=2} \|\p_\bx^{\bl} \eta\|_{\plainL\infty(B(\bx, \rho|t|))}
+ \|\nabla_\bx \eta\|_{\plainL\infty(B(\bx, \rho|t|))}\notag\\
&\ \qquad+  \|\eta\|_{\plainL\infty(B(\bx, \rho|t|))} 
+ \sum_{0<\bq\le\bm} 
|t|^{|\bq|} \|\p_\bx^{\bq}\, G\|_{\plainL\infty(B(\bx, \rho|t|))}\bigg].
\end{align*}
In this estimate we pick $\rho = \rho_1/2$ 
which ensures that $\rho|t| \le \ell(\bx)/2$, and hence $B(\bx, \rho|t|)\subset \Om$. 
Now we can use the bound \eqref{eq:ulteta} for $|\bm|\le 2$ that was proved in Steps 2 and 3. 
Therefore, the terms containing $\eta$ and its derivatives 
do not exceed $|t|^\mu$. Moreover, 
because of \eqref{eq:regG}, the term containing the 
derivatives of $G$ is estimated by $|t|$. 
Putting these bounds together we arrive at 
\eqref{eq:ulteta} 
for all $|\bm| >2$.
This completes the proof. 
\end{proof}

 \subsection{Isolating the leading singularity of $\psi$}
We now write the wavefunction within $\SW(\varepsilon, \d), \varepsilon>\d$, 
in a form where we isolate the singularity which 
will give rise to the eigenvalue asymptotics for the operators $\SG$ and $\SK$. 
Using (\ref{eq:psik}) and (\ref{eq:xia}) 
with $\a = 1/2$, we can write
\begin{align}\label{eq:psising}
\chi(\tilde\bx, z, t)  &= e^{K(t)}\big( t\cdot \nabla_t\xi(\tilde\bx, z, 0) + \eta(\tilde\bx, z, t)\big)\notag\\
&= \frac{|t|t}{4\sqrt{2}} \cdot \nabla_t\xi(\tilde\bx, z, 0) + L(\tilde\bx, z, t)
\end{align}
for a remainder
\begin{align}
\label{eq:l_def}
L(\tilde\bx, z, t) = \big(e^{K(t)}-K(t)\big)  t\cdot \nabla_t\xi(\tilde\bx, z, 0) + e^{K(t)}\eta(\tilde\bx, z, t).
\end{align}

\begin{lem}
Let $\bx \in \SW^\circ(\varepsilon, \d)$. Then 
for any $\mu \in [0,1)$ and all $\bm \in \mathbb N_0^{3N}$,
\begin{align}
\label{eq:ztphi}
|\partial_\bx^{\bm}L(\bx)| 
\lesssim 1+ |t|^{2+\mu-|\bm|},
\end{align}
with constants depending on $\varepsilon$ and $\d$, but independent of $\psi$.
\end{lem}
\begin{proof}
As a preliminary, we recall the following elementary fact: 
for any $k \in \mathbb N_0^3$, $|k| \ge 1$, one has 
\begin{gather}
\label{eq:expbnd}
\partial_x^{k}e^{|x|} \lesssim |x|^{1-|k|},\\
\label{eq:expbnd1}
\partial_x^{k}\big(e^{|x|}- |x|\big) \lesssim |x|^{2-|k|}, 
\end{gather}
for $x \in \mathbb R^3$ such that $|x| \lesssim 1$.

We now show the bound (\ref{eq:ztphi}) for each term 
of (\ref{eq:l_def}). Since $K(t) = \frac{1}{4\sqrt2}|t|$, it follows from \eqref{eq:expbnd1} that 
\begin{align*}
\big|\p_t^{m} \big(\big(e^{K(t)}-K(t)\big) t\big) \big| \lesssim 1 +|t|^{3-|m|},
\quad m\in\mathbb N_0^3.
\end{align*}
Together with the bound \eqref{eq:xiz} this leads to \eqref{eq:ztphi} for the first term in 
\eqref{eq:l_def}. 
 
Consider now the second term in \eqref{eq:l_def}. If $|\bm| = 0$ or $|\bm|=1$, then \eqref{eq:ztphi} holds because of \eqref{eq:etader} and \eqref{eq:expbnd}.
So let $|\bm| \ge 2$. By the Leibniz rule
\begin{align}\label{eq:etak}
|\p_{z, t}^\bm \big(e^{K(t)}\, \eta(\tilde\bx, z, t)\big)|
\lesssim &\ \sum_{0\le \bk\le \bm} \, |\p_{z, t}^{\bk}\, e^{K(t)}|\  
|\p_{z, t}^{\bm-\bk}\, \eta(\tilde\bx, z, t)|\notag\\
= &\ 
\sum_{0< \bk < \bm} \, |\p_{z, t}^{\bk}\, e^{K(t)}|\  
|\p_{z, t}^{\bm-\bk}\, \eta(\tilde\bx, z, t)|
+ 
|e^{K(t)}|\  
|\p_{z, t}^{\bm}\, \eta(\tilde\bx, z, t)|\notag\\
&\ \qquad + |\p_{z, t}\, e^{K(t)}|\  
|\eta(\tilde\bx, z, t)|.
\end{align}
Due to \eqref{eq:ulteta}, the sum on the right-hand side is bounded by 
\begin{align*}
\sum_{0< \bk < \bm} \, |t|^{1-|\bk|}\ |t|^\mu (1+ |t|^{2-|\bm|+|\bk|}) 
\lesssim |t|^{2-|\bm|+\mu} + |t|^{3-|\bm|+\mu},
\end{align*}
and hence it satisfies \eqref{eq:ztphi}. 
The second term on the right-hand side of \eqref{eq:etak} 
satisfies \eqref{eq:ztphi} due to \eqref{eq:ulteta}. We estimate 
the last term using \eqref{eq:etader} and \eqref{eq:expbnd}:
\begin{align*}
|\p_{z, t}^{\bm}\, e^{K(t)}|\  
|\eta(\tilde\bx, z, t)|\lesssim |t|^{1-|\bm|} |t|^2 
= |t|^{3-|\bm|}. 
\end{align*}
This completes the proof of \eqref{eq:ztphi}. 
\end{proof}

\subsection{Return to the variables $(\hat{\bx}, x)$}
We are now in a position to rewrite the above results in the initial variables 
$\bx = (x_1, x_2, \dots, x_{N-1}, x)$. 
For $\bx\in \SW_k(\varepsilon, \d)$ it will be convenient to define  
\begin{align*} 
 \sff_k(\bx) = \sv\bigg(\tilde\bx_{k, N}, \frac{x+x_k}{2}, \frac{x+x_k}{2}\bigg),
\quad k = 1, 2, \dots, N-1,
\end{align*}
where we have used our usual notation $\sv(\hat\bx, x) = \nabla_x\psi(\hat\bx, x)$, 
see \eqref{eq:sv}. 
Recall that $\sv(\hat\bx, x)$ is continuous on $\R^{3N}\setminus\{\bx = (\hat\bx, 0)\}$, 
see Remark \ref{rem:grad}.

\begin{lem} Let $\varepsilon>\d$. 
Then 
for all $\bx\in\SW_k(\varepsilon, \d)$,
$k = 1, 2, \dots, N-1$, the function $\sff_k(\hat\bx, \,\cdot\,)$ 
is infinitely differentiable w.r.t. $x$ and 
\begin{align}\label{eq:alsmooth}
|\p_x^{m}\, \sff_k(\hat\bx, x)|\lesssim 1, \quad m\in\mathbb N_0^{3}.
\end{align}
Furthermore, for all $\bx\in \SW_k(\varepsilon, \d)$ we have 
\begin{align}
\psi(\bx) = &\ 
\frac{1}{24}  \big(\nabla_x \, |x-x_k|^3\big)\cdot\sff_k(\bx) 
+ L_k(\bx),\label{eq:triplet}\\[0.3cm]
\nabla_x\psi(\bx) = &\ 
\frac{1}{24}  \,\big(\nabla_x\nabla_x^T\, |x-x_k|^3\big) 
\,\sff_k(\bx)
 + L_k^{(1)}(\bx),\label{eq:triplet1}
\end{align}
where the error terms 
$L_k$ and $L_k^{(1)}$ for all 
$\bx\in \SW_k^\circ(\varepsilon, \d)$ satisfy the following bounds with 
an arbitrary $\mu\in [0, 1)$:
\begin{align}\label{eq:lsmooth}
|\p_x^{m} L_k(\hat\bx, x)|\lesssim 1+ |x-x_k|^{2+\mu-|m|},\quad m\in \mathbb N_0^{3},
\end{align}
and 
\begin{align}\label{eq:l1smooth}
|\p_x^{m}  L_k^{(1)}(\hat\bx, x)|
\lesssim 1+ |x-x_k|^{1+\mu-|m|},\quad m\in \mathbb N_0^{3}.
\end{align}
\end{lem}

\begin{proof}
Let $\bx\in \SW_k(\varepsilon, \d)$.  
Recall the variables \eqref{eq:coord} 
that we used in Sect. \ref{subsect:two}:
\begin{align*}
z = \frac{x+x_k}{\sqrt2},\quad t = \frac{x-x_k}{\sqrt2}.
\end{align*}
First we prove that 
\begin{align}\label{eq:xibal}
\sff_k(\bx) = \frac{1}{\sqrt2}\,\nabla_t\, \xi(\tilde\bx_{k, N}, z, 0).
\end{align}
Indeed, by definition \eqref{eq:psichi}, 
\begin{align*}
\nabla_t\,\chi(\tilde\bx_{k, N}, z, t)
= \frac{1}{\sqrt2}\big(\nabla_x  - \nabla_{x_k}\big)
\, \psi\bigg(\tilde\bx_{k, N}, \frac{z-t}{\sqrt2}, \frac{z+t}{\sqrt2}\bigg).
\end{align*}
Due to the antisymmetry \eqref{eq:antisym},
\begin{align*}
\nabla_{x_k}\, \psi(\tilde\bx_{k, N}, x_k, x) = - \sv(\tilde\bx_{k, N}, x, x_k),
\end{align*}
and hence, for $t = 0$  we have 
\begin{align*}
\nabla_t\, \chi(\tilde\bx_{k, N}, z, 0)
= \sqrt2\, \sv\bigg(\tilde\bx_{k, N}, \frac{z}{\sqrt2}, \frac{z}{\sqrt2}\bigg)
= \sqrt2 \,\sff_k(\bx). 
\end{align*}
Now, differentiating the expression \eqref{eq:psik} (with $\a = 1/2$) for $t\not = 0$ we get
\begin{align*}
\nabla_t \chi(\tilde\bx, z, t) = \nabla_t K(t) \,\chi(\tilde\bx, z, t) + e^{K(t)}\nabla_t\xi(\tilde\bx, z, t).
\end{align*}
In the limit $t\to 0$ this gives 
\begin{align*}
\nabla_t\chi(\tilde\bx, z, 0) = \nabla_t\xi(\tilde\bx, z, 0),
\end{align*}
which immediately leads to \eqref{eq:xibal}. 

By Lemma \ref{lem:xiz}, the equality \eqref{eq:xibal} means that 
the vector-function 
$\sff_k$ is infinitely differentiable with respect to $x$ and \eqref{eq:alsmooth} holds. 
Furthermore, \eqref{eq:xibal}, \eqref{eq:psising} and \eqref{eq:ztphi} imply that 
\begin{align*}
\psi(\bx) 
= \frac{1}{8} |x-x_k| (x-x_k)\cdot\sff_k(\bx)
+ L_k(\bx),\quad \bx\in \SW_k(\varepsilon, \d),
\end{align*}
where $L_k$ satisfies \eqref{eq:lsmooth}. Since $|x-x_k| (x-x_k) =\frac{1}{3}\nabla_x|x-x_k|^3$, 
\eqref{eq:triplet} follows.

Proof of \eqref{eq:triplet1}. Differentiate \eqref{eq:triplet}:
\begin{align*}
\nabla_x\psi(\bx) = &\ \frac{1}{24}\big(\nabla_x\nabla_x^T\, |x-x_k|^3\big) \sff_k(\bx) 
+ L_k^{(1)}(\bx),\\
 L_k^{(1)}(\bx) = &\ 
\frac{1}{24}\big(\nabla_x\sff_k^T(\bx) \big)
\, \big(\nabla_x\,|x-x_k|^3\big) + \nabla_x L_k(\bx).
\end{align*}
Due to \eqref{eq:alsmooth} and \eqref{eq:lsmooth} the remainder $L_k^{(1)}$ satisfies \eqref{eq:l1smooth}.
\end{proof}
 
\section{Compact operators} \label{sect:compact} 
Here we provide necessary information about compact operators. 
Most of the listed facts can be found in \cite[Chapter 11]{BS}. 
Let $\CH$ and $\mathcal G$ be separable Hilbert spaces.  
Let $A:\CH\to\mathcal G$ be a compact operator. 
If $\CH = \mathcal G$ and $A=A^*\ge 0$, then $\l_k(A)$, $k= 1, 2, \dots$, 
denote the positive eigenvalues of $A$ 
numbered in descending order counting multiplicity. 
For arbitrary spaces $\CH$, $\mathcal G$ and compact $A$, by $s_k(A) >0$, 
$k= 1, 2, \dots$, we denote the singular values of 
$A$ defined by $s_k(A)^2 = \l_k(A^*A) = \l_k(AA^*)$.   
Denote by 
\begin{align*}
n(s, A)= \#\{k: s_k(A) >s\}
\end{align*}
the counting function of the singular numbers. 
If $s_k(A)\lesssim k^{-1/p}, k = 1, 2, \dots$, 
with some $p >0$, then we say that $A\in \BS_{p, \infty}$ and denote
\begin{align*}
\| A\|_{p, \infty} = \sup_k s_k(A) k^{\frac{1}{p}},
\end{align*}
so that $\|A^*A\|_{p/2, \infty} = \|A A^*\|_{p/2, \infty} = \|A\|_{p, \infty}^2$.
The class $\BS_{p, \infty}$ is a complete linear space with the quasi-norm $\|A\|_{p, \infty}$. 
For all $p>0$ the functional $\|A\|_{p, \infty}$ 
satisfies the following ``triangle" inequality for  
operators $A_1, A_2\in\BS_{p, \infty}$:
\begin{align}\label{eq:triangle}
\|A_1+A_2\|_{\scalet{p}, \scalel{\infty}}^{{\frac{\scalel{p}}{\scalet{p+1}}}}
\le \|A_1\|_{\scalet{p}, \scalel{\infty}}^{{\frac{\scalel{p}}{\scalet{p+1}}}}
+ \|A_2\|_{\scalet{p, \infty}}^{{\frac{\scalel{p}}{\scalet{p+1}}}}.
\end{align}  
This inequality allows us to estimate quasi-norms of ``block-vector" operators. 
Let $A_j\in\BS_{p, \infty}$ 
be a finite collection of compact operators. Define the operator $\BA:\CH\to\oplus_j \mathcal G$ by  
$\BA = \{A_j\}_{j}$.
Since 
\begin{align*}
\BA^* \BA = \sum_j A_j^*A_j,  
\end{align*}
and $A_j^*A_j\in\BS_{q, \infty}$, $q = p/2$, 
by \eqref{eq:triangle} we have 
\begin{align}\label{eq:blockvec}
\|\BA\|_{p, \infty}^{{\frac{\scalel{2p}}{\scalet{p+2}}}} 
= \|\BA^*\BA\|_{q, \infty}^{{\frac{\scalel{q}}{\scalet{q+1}}}}
\le  \sum_j\|A_j^*A_j\|_{q, \infty}^{{\frac{\scalel{q}}{\scalet{q+1}}}} = 
\sum_j\|A_j\|_{p, \infty}^{{\frac{\scalel{2p}}{\scalet{p+2}}}}.
\end{align}
Consequently, in order to estimate the singular values of $\BA$ it suffices to 
estimate those of its components $A_j$. We use this fact throughout the paper. 
We point out a useful multiplicative inequality: if $A_1\in \BS_{p_1, \infty}$, 
$A_2\in \BS_{p_2, \infty}$ with some $p_1>0, p_2 >0$, then for $p^{-1} = p_1^{-1} +p_2^{-1}$ 
we have  
\begin{align}\label{eq:multipl}
\| A_1\, A_2\|_{p, \infty}\le \bigg(\frac{p_1}{p}\bigg)^{\frac{1}{p_1}}\, 
\bigg(\frac{p_2}{p}\bigg)^{\frac{1}{p_2}}\, \|A_1\|_{p_1, \infty}\, \|A_2\|_{p_2, \infty},
\end{align}  
see \cite[Chapter 11, Theorem 6.9]{BS}.  
 
For $p\in (0, 1)$ it is often more convenient to use the following ``triangle" inequality for  
arbitrarily many operators $A_j\in\BS_{p, \infty}$, $j = 1, 2, \dots$:
\begin{align}\label{eq:ptriangle}
\big\|\sum_{j} A_j\big\|_{p, \infty}^p\le (1-p)^{-1}
\sum_j \|A_j\|_{p, \infty}^p,
\end{align}
see 
\cite[Lemmata 7.5, 7.6]{AJPR2002}, \cite[\S 1]{BS1977} and references therein. Under 
the additional assumption
\begin{align}\label{eq:orthog}
A_k A_j^* = 0\quad \textup{or}\quad A_k^*A_j = 0,\quad j\not = k, 
\end{align}   
the inequality of the form \eqref{eq:ptriangle} extends to all $p\in (0, 2)$: 

\begin{prop}\cite[Lemma 1.1]{BS1977}\label{prop:orthog}
Let $p\in (0, 2)$ and let $A_j\in \BS_{p, \infty}$, $j= 1, 2, \dots$, 
be a family of operators satisfying 
\eqref{eq:orthog}. 
Then for the operator $A = \sum_j A_j$ we have the inequality
\begin{align}\label{eq:pnorm}
\|A\|_{\scalet{p}, \scalel{\infty}}^{\scalet{p}}\le 
\frac{2}{2-p}\,
\sum_j \|A_j\|_{\scalet{p}, \scalel{\infty}}^{\scalet{p}}. 
\end{align} 
\end{prop}

\begin{proof} 
It suffices to conduct the proof under the first condition in 
\eqref{eq:orthog} only, so that  
\begin{align*}
A A^* = \sum_j A_j A_j^*. 
\end{align*}
Since $A_j A_j^*\in\BS_{q, \infty}, q = p/2<1$, 
and $\|A_j A_j^* \|_{q, \infty} = \|A_j\|_{p, \infty}^2$, 
the inequality \eqref{eq:ptriangle} leads to the bound 
\begin{align*}
\|A\|_{p, \infty}^p = \|A A^*\|_{q, \infty}^q \le (1-q)^{-1} 
\sum_j \|A_j A_j^*\|_{q, \infty}^q
= (1-p/2)^{-1}\sum_j \|A_j\|_{p, \infty}^p. 
\end{align*}
This inequality coincides with \eqref{eq:pnorm}.
\end{proof}

Let us now introduce the functionals describing the asymptotic behaviour of singular values. 
For $A\in \BS_{p, \infty}$ the following quantities are finite:
\begin{align}\label{eq:limsupinf}
\begin{lcases}
\SfG_p(A) = & 
\big(\limsup\limits_{k\to\infty} k^{\frac{1}{p}}s_k(A)\big)^{p},\\
\sg_p(A) = &  
\big(\liminf\limits_{k\to\infty} k^{\frac{1}{p}}s_k(A)\big)^{p},
\end{lcases}
\end{align}
and they clearly satisfy the inequalities
\begin{align} \label{eq:gnorm}
\sg_p(A)\le \SfG_p(A)\le \|A\|_{p, \infty}^p.
\end{align}
Note that if 
$A\in\BS_{p, \infty}$, then 
$\SfG_q(A) = 0$ for all $q > p$. 
Observe that 
\begin{align}\label{eq:double}
\sg_{p/2}(A A^*) = \sg_{p/2}(A^* A) = \sg_{p}(A),\quad 
\SfG_{p/2}(A A^*) = \SfG_{p/2}(A^* A) = \SfG_{p}(A).
\end{align}
If $\SfG_p(A) = \sg_p(A)$, then the singular values of $A$ satisfy the asymptotic formula
\begin{align*} 
s_k(A) = \big(\SfG_p(A)\big)^{\frac{1}{p}} k^{-\frac{1}{p}} + o(k^{-\frac{1}{p}}),\ k\to\infty.
\end{align*}
The functionals $\sg_p(A)$, $\SfG_p(A)$ 
also satisfy the inequalities of the type \eqref{eq:triangle}:
\begin{align}\label{eq:trianglep}
\begin{lcases}
\SfG_p(A_1 + A_2)^{{\frac{\scalel{1}}{\scalet{p+1}}}}
\le & \SfG_p(A_1)^{{\frac{\scalel{1}}{\scalet{p+1}}}}
+ \SfG_p(A_2)^{{\frac{\scalel{1}}{\scalet{p+1}}}},\\[0.3cm]
\sg_p(A_1+A_2)^{{\frac{\scalel{1}}{\scalet{p+1}}}}
\le & \sg_p(A_1)^{{\frac{\scalel{1}}{\scalet{p+1}}}}
+ \SfG_p(A_2)^{{\frac{\scalel{1}}{\scalet{p+1}}}}.
\end{lcases}
\end{align}
It follows from \eqref{eq:trianglep} that 
the functionals $\SfG_p$ and $\sg_p$ are continuous on $\BS_{p, \infty}$:
\begin{align*}
\big|
\SfG_p(A_1)^{{\frac{\scalel{1}}{\scalet{p+1}}}} - \SfG_p(A_2)^{{\frac{\scalel{1}}{\scalet{p+1}}}}
\big|\le &\ \SfG_p(A_1 - A_2)^{{\frac{\scalel{1}}{\scalet{p+1}}}},\\
\big|\sg_p(A_1)^{{\frac{\scalel{1}}{\scalet{p+1}}}} - \sg_p(A_2)^{{\frac{\scalel{1}}{\scalet{p+1}}}}
\big|\le &\ \SfG_p(A_1 - A_2)^{{\frac{\scalel{1}}{\scalet{p+1}}}}.
\end{align*} 
We need the following two corollaries of this fact:

\begin{cor}\label{cor:zero}
Suppose that $\SfG_p(A_1-A_2) = 0$. Then 
\begin{align*}
\SfG_p(A_1) = \SfG_p(A_2),\quad \sg_p(A_1) = \sg_p(A_2).
\end{align*}
\end{cor}

The next corollary is more general:

\begin{cor}\label{cor:zero1}
Suppose that $A\in\BS_{p, \infty}$ and that for every $\nu>0$ there exists an operator 
$A_\nu\in \BS_{p, \infty}$ such that $\SfG_p(A - A_\nu)\to 0$, 
$\nu\to 0$. Then the functionals 
$\SfG_p(A_\nu), \sg_p(A_\nu)$ have limits as $\nu\to 0$ and 
\begin{align*}
\lim_{\nu\to 0} \SfG_p(A_\nu) = \SfG_p(A),\quad 
\lim_{\nu\to 0} \sg_p(A_\nu) = \sg_p(A).
\end{align*}
\end{cor}

\section{Nikol'skii-Besov spaces}\label{sect:besov}

This section is a preparation for the study of singular values of integral operators 
conducted in Section \ref{sect:intop}. In Section  \ref{sect:intop} 
we rely on the paper \cite{BS1977} where   
the membership of integral operators 
in classes $\BS_{p, \infty}$ with some $p>0$ is described 
in terms of the smoothness of their integral kernels. 
More specifically, we need the bounds in  
$\BS_{p, \infty}$ under the assumption that the integral kernel 
lies in a certain \emph{Nikol'skii-Besov space}.   
Thus we start with  
basic facts about Besov 
spaces. More details can be found 
in \cite[Chapter 4]{Nikol1975},\cite[Sect. 5.3]{Burenkov1998}, 
\cite[Sect. V.5]{Stein1970} and \cite[Ch. 4]{Triebel1978}. 
 
\subsection{Besov spaces} \label{subsect:nikol} 
For a function $u = u(x)$, $x\in \R^d$, and arbitrary $l = 0, 1, 2, \dots,$ 
define the finite difference  
\begin{align}\label{eq:findif}
\Delta_h^{(l)} u (x) = \sum_{j=0}^l (-1)^{j+l} {l \choose j} u(x+jh),
\end{align}
and the \emph{$\plainL{q}$-modulus of smoothness of order} $l$:
\begin{align}\label{eq:ms}
\om_q^{(l)}(u; t) = \sup_{|h|\le t} \|\Delta_h^{(l)}u\|_{\plainL{q}},\quad t >0.
\end{align}
The function $u$ is said to belong to the Besov space $\plainB^{s}_{q, \infty}(\R^d)$,  
$s>0$, $q\in [1, \infty]$, if 
for some $l > s$ we have 
\begin{align}\label{eq:besov}
\|u\|_{\plainB^{s}_{q, \infty}} := \|u\|_{\plainL{q}} + \sup_{t>0} t^{-s}\om_q^{(l)}(u; t) < \infty.
\end{align}
This quantity defines a norm on 
$\plainB^{s}_{q, \infty}(\R^d)$. 
Different values of $l > s$ lead to equivalent norms \eqref{eq:besov}, 
see \cite[Sect. 5.3]{Burenkov1998}. 
The space $\plainB^{s}_{q, \infty}(\R^d)$ is just one representative of the whole scale 
of Besov spaces $\plainB^{s}_{q, r}(\R^d)$, $0 <r \le \infty$. In this paper we need only 
$r = \infty$. 
Sometimes in the literature the space $\plainB^{s}_{q, \infty}$ is called 
the \emph{Nikol'skii space} or \emph{Nikol'skii-Besov space}, 
see e.g. \cite{BS1977}, \cite{Burenkov1998}. We use the notation 
$\plainN{s}_{q}$ for it and write  
$\|u\|_{\plainN{s}_q}$ for the norm \eqref{eq:besov}. 

Note that 
\begin{align}\label{eq:iter}
\Delta_h^{(l)} u(x) = \Delta_h^{(l-m)} (\Delta_h^{(m)} u)(x)
= \sum_{j=0}^{l-m}(-1)^{j+l-m}{l-m \choose j} \Delta^{(m)}_hu(x+jh),
\end{align}
for all $m < l$.
Thus, under the assumption 
$u\in \plainW{l, 1}_{\textup{loc}}(\R^d)$, 
iterating the identity
\begin{align*}
\Delta_h^{(1)} u(x) = u(x+h)- u(x) = \int_0^1 h\cdot\nabla u(x+sh) ds,
\end{align*}
we obtain the formula
\begin{align}\label{eq:fintodiff}
\Delta_h^{(l)} u(x) = \int_0^1 \int_0^1 \cdots \int_{0}^1 
(h\cdot \nabla)^l u\big( x+ \sum_{j=1}^l s_j h\big)\, ds_1 ds_2\dots ds_l.
\end{align}
For a domain $\Om\subset\R^d$ the space $\plainN{s}_q(\Om)$ is defined as the restriction 
of $\plainN{s}_q(\R^d)$ to $\Om$. The corresponding norm of the function 
$u\in \plainN{s}_q(\Om)$ is defined as 
\begin{align}\label{eq:restr}
\|u\|_{\plainN{s}_q(\Om)} = \inf \|g\|_{\plainN{s}_q(\R^d)},
\end{align}
where the infimum is taken over all functions $g\in \plainN{s}_q(\R^d)$ such 
that $u=g$ for a.e. $x\in\Om$. In other words, a function $u$ belongs to 
$\plainN{s}_q(\Om)$ if it has an extension $g\in \plainN{s}_q(\R^d)$. 
Here we have adopted the definitions from \cite[Sect. 4.2.1]{Triebel1978}. 
They will be sufficient for our purposes.  

There are also intrinsic (i.e. not based on restriction) definitions of the Nikol'skii and Besov 
spaces on domains, see for example \cite[Sect. 4.3]{Nikol1975} or \cite[Section 7.32]{Adams_Fournier_2003}. 
For domains 
with ``good" extension properties the intrinsic and restriction definitions are equivalent, see 
\cite[Sect. 7.32]{Adams_Fournier_2003}.  
The collection of such domains includes 
bounded domains satisfying the cone condition (see \cite[Sect. 4.2.3]{Triebel1978}), 
and in particular, bounded Lipschitz domains.   
 
\subsection{A test of membership in $\SN^s_q$}\label{subsect:mbr}
We are interested in functions that fall in the appropriate spaces $\plainN{s}_q$  
in a natural way. 
Let $u\in \plainC\infty(\R^d\setminus X)$, where $X = \{a_1, a_2, \dots, a_N\}$ is a finite collection of points $a_k\in \R^d$, $k = 1, 2, \dots, N$. 
Assume that  
for some $\a \in\R$ and some $A\ge 0$ we have 
\begin{align}\label{eq:sing}
|\p_x^j u(x)|\lesssim A\big(1+\sum_{k=1}^N \big(|x-a_k|\wedge 1\big)^{\a-|j|}\big),
\quad \textup{for all} \quad x\notin X\quad \textup{and}
\quad j \in \mathbb N_0^d,
\end{align} 
where, as explained in the Introduction, 
$a\wedge b = \min(a, b)$ for any two non-negative $a$ and $b$. 
A good example that ``almost" satisfies this condition 
with $X = \{0\}$ is the class of homogeneous functions, i.e.   
functions 
$\Phi\in\plainC\infty(\R^d\setminus\{0\} )$ such that $\Phi(tx) = t^{\a} \Phi(x)$ for all
$t >0$ and all $x\in \R^d, x\not = 0$, with some $ \a\in\R$. As follows from 
\cite[Lemma 10.1]{BS1977}, if $\a > -d/q$, then $\Phi\in \plainN{s}_q(B)$, $s = \a + d/q$, for 
any ball $B\in\R^d$. 
In the next lemma we show the same for functions satisfying \eqref{eq:sing}.  
This result is elementary but we have not been able to find it in the literature.
It will serve as a basis for our study of integral operators.
  
\begin{thm}  \label{thm:nik} 
Suppose that $d\ge 2$ and that the function $u$ satisfies \eqref{eq:sing}.  
If $\a > -d/q$ with some $1\le q\le d$, then 
for all $x_0\in\R^d$ and $R>0$ the function $u$ belongs to $\plainN{s}_q(B)$,  $B = B(x_0, R)$, 
with $s = \a + d/q$.
Moreover, $\|u\|_{\plainN{s}_q(B)}\lesssim A$ where the implicit constant 
does not depend on $x_0$ and the set $X$, but may depend on the radius $R$.  
\end{thm}

We precede the proof with the following useful remark.

\begin{rem}\label{rem:remove}
Define the integer $m = m(\a)$: 
\begin{align}\label{eq:mit}
\begin{lcases}
m = &\  0, \quad \textup{if} \quad \a\le 0,\\
m = &\  [\a]+1, \quad \textup{if}\quad \a>0, \a\notin \mathbb Z,\\
m = &\  \a, \quad \textup{if}\quad \a>0, \a\in\Z.
\end{lcases}
\end{align}
Under the assumption $d\ge 2$ the potential singularities of the function $u$ 
at the points $a_k, k= 1, 2, \dots, N$, 
are ``removable" in the following sense. 
If $\a >0$ (i.e. $m\ge 1$) the function $u$ clearly belongs to 
$\plainW{m, d}_{\textup{loc}}(\R^d\setminus X)$. Thus it follows from 
Lemma \ref{lem:remove} that for $d\ge 2$ the function $u$ also belongs to 
$\plainW{m, d}_{\textup{loc}}(\R^d)$. 
\end{rem}

\begin{proof}[Proof of Theorem \ref{thm:nik}] 
In view of the definition \eqref{eq:restr} we may assume that 
$u$ is supported on the ball $B(x_0, 2R)$ and shall prove that 
$\|u\|_{\plainN{s}_q(\R^d)}\lesssim A$. 
Since $\a > -d/q$, we have $u\in\plainL{q}(\R^d)$ and $\|u\|_{\plainL{q}}\lesssim A$. 
It remains to 
estimate the modulus of smoothness  \eqref{eq:ms}.

Let 
\begin{align}\label{eq:el}
s = \a + d/q\quad \textup{and} \quad l = [s]+1.
\end{align} 
We shall show that 
\begin{align}\label{eq:ms2}
\|\Delta_h^{(l)} u \|_{\plainL2}\lesssim A\,
(1\wedge|h|^s),\quad s = \a+\frac{d}{2},
\end{align} 
for all $h\in\R^d$, 
with an implicit constant independent of $x_0$ and of the set $X$. 
If $|h|\ge 1$, then by the definition \eqref{eq:findif}, 
\begin{align*}
\|\Delta_h^{(l)} u \|_{\plainL{q}}\le 2^l\|u\|_{\plainL{q}}\lesssim A,
\end{align*}
so that it remains to prove \eqref{eq:ms2} for $|h|\le 1$.

Denote by $\chi_k(x; r), r>0,$ the indicator of the ball $B(a_k, r)$, and define 
$\kappa_k(x; r)= 1-\chi_k(x; r)$. 
The first step is to prove that 
\begin{align}\label{eq:inside}
\|\chi_k(\ \cdot\ ; 2l|h|) \Delta_h^{(l)} u \|_{\plainL{q}}\lesssim A\, |h|^s,\quad 
k = 1, 2, \dots, N.
\end{align}
For the finite difference $\Delta_h^{(l)} u$ we use the representation \eqref{eq:iter} 
where $m$ is as defined in~\eqref{eq:mit}. 
In view \eqref{eq:iter}, we have 
\begin{align*}
|\chi_k(x ; 2l|h|)\, \Delta_h^{(l)} u(x) |
\lesssim \sum_{j=0}^{l-m}|\chi_k(x ; 2l|h|)\,  \Delta_h^{(m)} u(x+jh)|.
\end{align*}
Estimate $\plainL{q}$-norms of each term on the right-hand side individually. 
For $m = 0$ we have 
\begin{align}\label{eq:deltam0}
\|\chi_k(\, \cdot\, ; 2l|h|)\,  \Delta_h^{(0)} & u(\ \cdot\ + jh)\|_{\plainL{q}}^q\notag\\
\le &\ 
\int_{\R^d}
\chi_k(x ; 2l|h|)\, \big|u\big(x+ jh\big)\big|^q dx\notag\\
\le &\  \int_{\R^d} \chi_k(x ; 4l|h|)\,
\big|u(x)\big|^q dx\,,  
\end{align}
for all $j = 0, 1, \dots, l$. 

If $m\ge 1$, then in view of Remark 
\ref{rem:remove} 
we have $u\in \plainW{m, d}_{\textup{loc}}(\R^d)$. Therefore, using 
\eqref{eq:fintodiff}, for arbitrary $q \in [1, d]$ we can estimate: 
\begin{align}\label{eq:deltam}
\|\chi_k(\, \cdot\, ; & 2l|h|)\,  \Delta_h^{(m)}   u(\ \cdot\ + jh)\|_{\plainL{q}}^q\notag\\
\le &\ |h|^{qm}  \int_0^1 \int_0^1 \cdots \int_{0}^1\int_{\R^d}
\chi_k(x ; 2l|h|)\, \big|\nabla^m 
u\big(x+ jh + \sum_{n=1}^m s_n h\big)\big|^q dx\, ds_1 ds_2 \dots ds_n\notag\\
\le &\ |h|^{qm}  \int_{\R^d} \chi_k(x ; 4l|h|)\,
\big|\nabla^m u(x)\big|^q dx\,,  
\end{align}
for all $j = 0, 1, \dots, l-m$, where we have used the notation
\begin{align*}
|\nabla^m u| = \big[\sum_{|j|=m} |\p^j u|^2\big]^{\frac{1}{2}}.
\end{align*}
Now use \eqref{eq:sing}, so that the right-hand side 
of \eqref{eq:deltam0} (for $m=0$) and \eqref{eq:deltam} (for $m\ge 1$) is bounded by 
\begin{align}\label{eq:polar}
|h|^{qm}\,A^q\,\sum_{n=1}^N\int\limits_{|x-a_k|<4l|h|} \big(1+&\ |x-a_n|^{q(\a-m)}\big)\,  dx
\notag\\[0.2cm]
\lesssim &\ |h|^{qm+d}\,A^q + |h|^{qm}\,A^q\,\sum_{n=1}^N\int\limits_{|x-a_k|<4l|h|} |x-a_n|^{q(\a-m)}\,dx.
\end{align}
Since $\a-m>-1$ and $d\ge q$, each integral in the above sum is finite.  
To estimate each of the integrals consider two cases. If $a_n\in B(a_k, 5l|h|)$, then 
\begin{align*}
\int\limits_{|x-a_k|<4l|h|} |x-a_n|^{q(\a-m)}\,dx 
\lesssim \int\limits_{|x-a_n|<10l|h|} |x-a_n|^{q(\a-m)}  dx 
\lesssim |h|^{q\a-qm+d}. 
\end{align*}
If $a_n\notin B(a_k, 5l|h|)$, then $|x - a_n|\ge l|h|$ for $x\in B(a_k, 4l|h|)$, 
and as $\a-m \le 0$, we have again
\begin{align*}
\int\limits_{|x-a_k|<4l|h|} |x-a_n|^{q(\a-m)}\, dx 
\lesssim 
(l|h|)^{q(\a-m)}\,
\int\limits_{|x-a_k|<4l|h|} dx
\lesssim |h|^{q\a-qm+d}. 
\end{align*}
Substituting the last two bounds in \eqref{eq:polar} we conclude that 
the right-hand side of \eqref{eq:deltam} is bounded from above 
by $A^q\,|h|^{q\a+d}= A^q\,|h|^{qs}$, which leads to 
\eqref{eq:inside}. 

Now, denote 
\begin{align*}
\kappa(x; r) = \prod_k \kappa_k(x; r),\, \kappa_k(x; r) = 1 - \chi_k(x; r),
\end{align*}
and prove that  
\begin{align}\label{eq:outside}
\|\kappa(\ \cdot\ ; 2l|h|)\Delta_h^{(l)}u\|_{\plainL{q}}\lesssim A\,|h|^{s},
\end{align}
where the number $l$ is defined in \eqref{eq:el}. 
The function $u$ is $\plainC\infty$ on the support of the function $\kappa$. Therefore, 
similarly to \eqref{eq:deltam}, 
\begin{align*}
\|\kappa(\ \cdot\ ; 2l|h|)\, \Delta_h^{(l)} & u\|_{\plainL{q}}^q\\
\le &\ |h|^{ql}  \int_0^1 \int_0^1 \cdots \int_{0}^1\int_{\R^d}
\kappa(x; 2l|h|)\,\big|\nabla^l 
u\big(x+ \sum_{n=1}^l s_n h\big)\big|^q dx\, ds_1 ds_2 \dots ds_n\\
\le &\ |h|^{ql}  \int 
\kappa(x; l|h|)
\, \big|\nabla^l u(x)\big|^q dx\,. 
\end{align*}
Since $ql > q\a+d$, by \eqref{eq:fintodiff}, the right-hand side does not exceed 
\begin{align*}
|h|^{q l}\, A^q\,\sum_{k=1}^N  \int\limits_{\substack{|x-a_k|> l|h|\\ x\in B(x_0, 2R)}} 
\big(1+|x-a_k|^{q(\a-l)}\big) \, dx\lesssim \, 
|h|^{ql}A^q\big(1 + |h|^{q\a-ql + d}\big)
\lesssim 
A^q\,|h|^{qs}.
\end{align*}
This implies \eqref{eq:outside}. 

Using the inequality $1\le \sum_{k=1}^N \chi_k(x; r) + \kappa(x; r)$, we can estimate
\begin{align*}
\|\Delta_h^{(l)}u\|_{\plainL{q}}^q
\le \sum_{k=1}^N\|\chi_k(\ \cdot\ ; 2l|h|)\, \Delta_h^{(l)}u\|_{\plainL{q}}^q 
+ \|\kappa(\ \cdot\ ; 2l|h|)\Delta_h^{(l)}u\|_{\plainL{q}}^q,
\end{align*}
and hence putting \eqref{eq:inside} and \eqref{eq:outside} together we obtain \eqref{eq:ms2}. 
Along with the fact that $\|u\|_{\plainL{q}}\lesssim A$ this implies that 
$\|u\|_{\plainN{s}_q}\lesssim A$ as required. 
\end{proof}

\section{Singular values of integral and pseudodifferential operators}\label{sect:intop}

We begin with Sect.\ref{subsect:BS} where we derive estimates for singular values of integral operators from 
the results of \cite{BS1977} relying on Theorem \ref{thm:nik}. In Sect.~\ref{subsect:symbols} we collect the known results on spectral asymptotics of pseudo-differential operators with asymptotically homogeneous symbols. In Sect.~\ref{subsect:iop} we 
use these results to obtain the asymptotics for integral operators with homogeneous kernels.

\subsection{Bounds for singular values of integral operators}\label{subsect:BS}
Let $X\subset \R^d$, be a bounded Lipschitz domain.
Let $T(t, x)$, $t\in \R^m$, $x\in X$, be some complex-valued function. We estimate 
singular values of the operator $\iop(T) \, a$ with some weight $a = a(x), x\in X$, 
using \cite[Corollaries 4.1 and 4.5, Theorems 4.4, 4.7]{BS1977} 
which we state here in the form convenient for 
our purposes. 

\begin{prop}\label{prop:BS} 
Let $X\subset \R^d$ be a bounded Lipschitz domain. 
Assume that the kernel $T(t,x), t\in \R^m,\, x\in X,$ is such that 
$T(t, \, \cdot\, )\in \plainN{s}_2(X)$ with some $s>0$, 
for a.e. $t\in \R^m$. Assume that 
$a\in\plainL{r}(X)$, where 
\begin{align*}
\begin{lcases}
r = &\ 2, \quad \textup{if} \quad 2s >d,\\
r > &\ d s^{-1} \quad\textup{is arbitrary}, \quad \textup{if}\quad 2s \le d. 
\end{lcases}
\end{align*}
Then $\iop(T)\, a\in\BS_{q, \infty}$ where $q^{-1} = 2^{-1} + s d^{-1}$ and 
\begin{align*}
\|\iop(T)\, a\|_{q, \infty}\lesssim \biggl[\int_{\R^m}  
\|T(t,\ \cdot\ )\|_{\plainN{s}_2(X)}^2\, dt \biggr]^{\frac{1}{2}}
\|a\|_{\plainL{r}(X)},
\end{align*}
under the assumption that the right-hand side is finite.
\end{prop}    

We apply this proposition to the kernel $T(t, x)$ that satisfies the 
bound \eqref{eq:sing} 
with some $\a>-d/2$ as a function of $x$. 
In this case we define $s = \a + d/2>0$, so that the conditions on the 
number $r$ in Proposition 
\ref{prop:BS} rewrite as follows  
\begin{align}\label{eq:r}
\begin{lcases}
r = &\ 2, \quad \textup{if} \quad \a >0,\\
r > &\ 2d(2\a+d)^{-1}, \quad \textup{if}\quad -d < 2\a \le 0. 
\end{lcases}
\end{align}
For $\ell >0$ below we denote 
\begin{align}\label{eq:cube}
\CC_n^{(\ell)} = [-\ell/2, \ell/2)^d + n, n\in\mathbb Z^d. 
\end{align}
The notation $\1_n$ is used for the 
indicator function of the cube $\CC_n = \CC_n^{(1)}$.  

In the next theorem we consider an operator from $\plainL2(\R^d)$ to $\plainL2(\R^m)$. 

\begin{thm} \label{thm:BSspace} 
Let $d\ge 2$ and $m\ge 1$. 
Let $z_k: \R^m\mapsto \R^d$,\, 
$k = 1, 2, \dots, N$, be some functions of $t\in\R^m$, and let $\a > -d/2$. 
Consider the kernel $T(t, x)$, $x\in\R^d, t\in\R^m,$ satisfying the condition 
\begin{align}\label{eq:kernelexp}
|\p_x^j T(t, x)| \le A(t, x) 
\bigg(1+\sum_{k=1}^N \big(1\wedge |x-z_k(t)|\big)^{\a-|j|}\bigg),\quad  
j \in \mathbb N_0^d, 
\end{align}    
for a.e. $t\in\R^m$ and a.e. $x\in \R^d$, where 
$A\in\plainL\infty_{\textup{\tiny loc} }(\R^m\times\R^d)$ 
is a non-negative function. 
Let  $a\in\plainL{r}_{\textup{\tiny loc}}(\R^d)$ 
where the parameter $r$ is as specified in \eqref{eq:r}.
Then the operator $\iop(T)\, a$ 
belongs to $\BS_{q, \infty}$ where $q^{-1} = 1+\a d^{-1}$ and 
\begin{align}\label{eq:piv}
\|\iop(T)\,a\|_{q, \infty}\lesssim 
\bigg[\sum_{n\in\Z^d}
\|A_n\|_{\plainL2(\R^m)}^q\,\|a\|_{\plainL{r}(\CC_n^{(1)})}^q
\bigg]^{\frac{1}{q}},\quad 
A_n(t) = \|A(t,\,\cdot\,)\|_{\plainL\infty(\CC_n^{(2)})},
\end{align}
under the assumption that the right-hand side of the above inequality is finite. 
\end{thm}

\begin{proof} 
Consider the 
kernel $W_n(t, x) = T(t, x) a(x)\1_n(x)$ 
and apply Proposition \ref{prop:BS} to each $\iop(W_n)$.  
Let $\eta\in \plainC\infty_0(\CC^{(2)}_0)$  
be a function such that $\eta(x) = 1$ if $x\in \CC^{(1)}_0$. 
In view of \eqref{eq:kernelexp}, the kernel $T(t, x)\eta_n(x)$ satisfies 
\begin{align*}
|\p_x^j \,\big(T(t, x)\eta_n(x)\big)|\le A_n(t) 
\bigg(1+\sum_{k=1}^N \big(1\wedge |x-z_k(t)|\big)^{\a-|j|}\bigg),\quad  
j \in \mathbb N_0^d, 
\end{align*}
for a.e. $t\in\R^m$. 
By Theorem \ref{thm:nik}, $T(t, \,\cdot\,) \eta_n(\,\cdot\,)\in \SN_2^s(\CC_n^{(2)})$, 
$s = \a+d/2$, for a.e. $t\in\R^m$, and hence $T(t, \,\cdot\,)\in \SN_2^s(\CC_n^{(1)})$, and 
\begin{align*}
\|T(t, \, \cdot\, )\|_{\plainN{s}_2(\CC_n^{(1)})}\lesssim A_n(t),
\end{align*}
with a constant independent of $n$. 
 Thus, by Proposition \ref{prop:BS}, the operator $\iop(W_n)$ belongs to 
$\BS_{q, \infty}$ with $1/q = 1/2 + s/d = 1+ \a/d$, and 
\begin{align}\label{eq:wn}
\|\iop( W_n)\|_{q, \infty}\lesssim \bigg[\int_{\R^m} A_n(t)^2 \,dt \bigg]^{\frac{1}{2}}
 \|a\|_{\plainL{r}(\CC_n^{(1)})} = \|A_n\|_{\plainL2(\R^m)}\|a\|_{\plainL{r}(\CC_n^{(1)})},
\end{align} 
with a constant independent of $n$. 
Now observe that $\iop(W_k)\iop(W_j)^* = 0$ for $k\not = j$. Furthermore, 
since $\a > -d/2$ we have $q < 2$, and hence by Proposition \ref{prop:orthog},
\begin{align*}
\|\iop(T)\, a\|_{q, \infty}^q\le 2(2-q)^{-1} \sum_{n\in\mathbb Z^n} \|\iop(W_n)\|_{q, \infty}^q.
\end{align*}
Using \eqref{eq:wn} we obtain the required bound \eqref{eq:piv}.
\end{proof}

\begin{cor}\label{cor:pivo}
Suppose that the kernel $T$ satisfies \eqref{eq:kernelexp} with a function 
$A\in \plainL\infty(\R^m\times\R^d)$. 
Let $b\in \plainL2(\R^m)$  and $a\in\plainL{r}_{\textup{\tiny loc}}(\R^d)$ be such that 
$\4\, a\4_{\,r, q}<\infty$, where $q^{-1} = 1+\a d^{-1}$. 
Then $b\, \iop(T)\, a$ belongs to $\BS_{q, \infty}$ and 
\begin{align}\label{eq:pivo}
\|b\, \iop(T)\,a\|_{q, \infty}\lesssim \|A\|_{\plainL\infty(\R^m\times\R^d)}\, 
\|b\|_{\plainL2(\R^m)}
\, \4\, a \4_{\,r, q}.
\end{align}
The constant in the bound does not depend on the functions $a, b$ and $A$.
\end{cor}

\begin{proof} 
We use Theorem \ref{thm:BSspace} replacing the function $A$ by $b A$. Thus
\begin{align*}
\| b A_n\|_{\plainL2(\R^m)}\le \|A\|_{\plainL\infty(\R^m\times\R^d)} \|b\|_{\plainL2(\R^m)}.
\end{align*}
Now \eqref{eq:pivo} follows from \eqref{eq:piv}.
\end{proof}

The next corollary is adapted for the use with the operators $\iop(\SPsi)$ 
and $\iop(\SV)$ introduced in the introduction.
Now we use the previous results with $d = 3, m = 3N-3$, and instead of $(t, x)$ we write 
$(\hat\bx, x)$. We also use the notation $\l(\bx) = \l_N(\bx)$, where $\l_j$ is defined in \eqref{eq:dq}, \eqref{eq:lj}.

\begin{cor}\label{cor:homo}
Let the functions $a = a(x)$, $b= b(\hat\bx)$, $A = A(\hat\bx, x)$, 
$x\in\R^3, \hat\bx\in\R^{3N-3}$, 
be as in Corollary \ref{cor:pivo}, and let $\a > -3/2$. Assume that the kernel $T = T(\hat\bx, x)$ 
satisfies the bounds 
\begin{align}\label{eq:s}
|\p_x^j\, T(\hat\bx, x)|\lesssim A(\hat\bx, x)\big(1 + \l(\hat\bx, x)^{\a-|j|}\big),\quad j\in \mathbb N_0^3.
\end{align}
Then the operator $b\,\iop(T)\, a$ 
belongs to $\BS_{p, \infty}$ where $1/p = 1+ \a/3$, and 
\begin{align}\label{eq:homoprep}
\|b\,\iop(T) \, a\|_{p, \infty}\lesssim \|A\|_{\plainL\infty(\R^{3N})}\,
\|b\|_{\plainL2(\R^{3N-3})} 
\, \4\, a \4_{\,r, p},
\end{align}
with a constant independent of the functions $a, b$ and $A$.

\end{cor}

\begin{proof} This is an immediate consequence of Corollary \ref{cor:pivo}.
\end{proof}

\begin{cor}\label{cor:product} 
Let the functions $a = a(x)$, $b= b(\hat\bx)$, $A = (\hat\bx, x)$, 
$x\in\R^3, \hat\bx\in\R^{3N-3}$, 
be as in Corollary \ref{cor:pivo}, and let $\a > -3/2$. 
Assume that the kernel $T = T(\hat\bx, x)$ 
is represented as a product $T = T_1 T_2$, where $T_1$ and $T_2$ 
satisfy the bounds 
\begin{align}\label{eq:t1}
|\p_x^j\, T_1(\hat\bx, x)|\lesssim A(\hat\bx, x)\, (|x-x_k|\wedge 1)^{\a-|j|},\quad j\in \mathbb N_0^3,
\end{align}
and
\begin{align}\label{eq:t2}
|\p_x^j\, T_2(\hat\bx, x)|\lesssim 1+ (|x-x_k|\wedge 1)^{1-|j|},\quad j\in \mathbb N_0^3,
\end{align}
for some fixed $k=1, 2, \dots, N-1$,  
Then 
$b\, \iop(T) \, a\in \BS_{p, \infty}$, $1/p = 1+\a/3$, and the bound \eqref{eq:homoprep} holds.  

If $T_2(\hat\bx, x_k) = 0$, and 
 $\4\, a\4_{\,r,q}<\infty$, where $1/q = 1+ (\a+1)/d$, 
then 
$b\, \iop(T) \, a\in \BS_{q, \infty}$.   
\end{cor}

\begin{proof} Assume without loss of generality that $\|A\|_{\plainL\infty(\R^{3N})} = 1$. 
Use the Leibnitz rule:
\begin{align}\label{eq:fr}
\p_x^j\, T(\hat\bx, x) 
= &\ \sum_{s = 0}^j \, {j\choose s}\ \p_x^s\, T_1(\hat\bx, x)\, \p_x^{j-s}\, 
T_2(\hat\bx, x)\notag\\
= &\ \sum_{s = 0}^{j-1} \, {j\choose s}\ \p_x^s\, T_1(\hat\bx, x)\, 
\p_x^{j-s}\, T_2(\hat\bx, x) 
+ \p_x^j\, T_1(\hat\bx, x)\, T_2(\hat\bx, x). 
\end{align}
By \eqref{eq:t1} and \eqref{eq:t2}, each term under the sum sign is bounded by 
\begin{align*}
(|x-x_k|\wedge 1)^{\a-|s|} \, (|x-x_k|\wedge 1)^{1-|j|+|s|}
= (|x-x_k|\wedge 1)^{\a + 1-|j|},
\end{align*}
whereas the term outside the sum is bounded by $(|x-x_k|\wedge 1)^{\a-|j|}$, 
so that together these two estimates yield
\begin{align*}
|\p_x^j\, T (\hat\bx, x)|\lesssim   (|x-x_k|\wedge 1)^{\a-|j|},\quad j\in \mathbb N_0^3.
\end{align*}
Hence \eqref{eq:s} also holds, whence \eqref{eq:homoprep}. 

If $T_2(\hat\bx, x_k) = 0$, then thanks to \eqref{eq:t2} we can write that
\begin{align*}
|T_2(\hat\bx, x)| = |T_2(\hat\bx, x) - T_2(\hat\bx, x_k)|
\lesssim |x-x_k|\wedge 1, 
\end{align*}
and therefore the term outside the sum in \eqref{eq:fr} 
is bounded by $(|x-x_k|\wedge 1)^{\a + 1-|j|}$, and as a consequence, 
\begin{align*}
|\p_x^j\, T (\hat\bx, x)|\lesssim   (|x-x_k|\wedge 1)^{\a+1-|j|},\quad j\in \mathbb N_0^3.
\end{align*}
It remains to use again Corollary \ref{cor:homo}.
\end{proof}

The asymptotics of singular values for operators 
with homogeneous symbols or kernels are studied in the next Section.

\subsection{Operators with asymptotically homogeneous symbols} \label{subsect:symbols}
In this section 
we consider pseudodifferential operators with asymptotically homogeneous matrix-valued symbols. 
Spectral asymptotics for such operators were studied in \cite{BS1977_1}, \cite{BS1979}, 
see also \cite{Ponge2023}. In fact, 
these papers allow for more general operators, but we need only a relatively simple 
special case of those results.  
Precisely, let $\CA(x), \CB(x), X(x, \xi)$, where $x, \xi\in\R^d$, be 
rectangular matrix-valued functions 
of matching dimensions, i.e. such that the product 
\begin{align*}
\CB(x) X(x, \xi)\CA(y) 
\end{align*}
is again a rectangular matrix.
We study the spectrum of the pseudodifferential operator $T:\plainL2(\R^d)\mapsto\plainL2(\R^d)$ 
defined by the formula 
\begin{align}\label{eq:pdo}
(Tu)(x) = \frac{1}{(2\pi)^{d}}\iint\CB(x) e^{i\xi(x-y)}
X(x, \xi) \CA(y) u(y) dy d\xi,
\end{align}
under the assumption that 
\begin{align*}
\CB\in \plainC{}_0(\R^d),
\quad \CA\in\plainC{}_0(\R^d). 
\end{align*}
We call the function $X(x, \xi)$ symbol and $\CA(x)$, $\CB(x)$ -- weights.
We often do not reflect the matrix nature of the functional spaces in the 
notation to avoid cumbersome formulas, and 
this should not cause confusion. 
Suppose that the symbol $X(x, \xi)$ is a bounded $\plainC\infty(\R^{2d})$-function which is 
asymptotically homogeneous of negative order in the variable $\xi$, i.e. 
there exists a matrix-valued function $X_\infty(x, \xi)$ such that 
$X_\infty(x, \,\cdot\,)\in\plainC\infty(\R^d\setminus\{0\})$ for a.e. $x$, and that for some 
$\tau >0$,
\begin{align}\label{eq:homo}
 X_\infty(x, t\xi) = t^{-\tau} X_\infty(x, \xi), \ \xi\not = 0, 
\end{align}
for all $t >0$, and 
\begin{align}\label{eq:Xas}
\sup_{x\in \R^d}\, |X(x, \xi) - X_\infty(x, \xi)| = o(|\xi|^{-\tau}),\quad |\xi|\to \infty,
\end{align}
where by $| T |$ we denote the Hilbert-Schmidt norm of the matrix $T$. 
Define the matrix-valued function
\begin{align*}
\CT_\infty(x, \xi) = \CB(x) X_\infty(x, \xi)\CA(x). 
\end{align*}
Recall that the functionals $\SfG_p$ and $\sg_p$ are 
defined in \eqref{eq:limsupinf}.

\begin{prop}\label{prop:homo}
Let the above conditions on 
$\CA, \CB, X$ be satisfied and let $p = d\tau^{-1}$. 
Then the operator \eqref{eq:pdo} 
is compact, it belongs to $\BS_{p, \infty}$ and satisfies the asymptotic formula 
\begin{align*}
\SfG_p(T) = \sg_p(T) = 
\frac{1}{ (2\pi)^d} \iint\, n(1, \CT_\infty(x, \xi)) \, dx d\xi,
\end{align*}
or, equivalently, 
\begin{align}\label{eq:pdoas}
\SfG_p(T) = \sg_p(T) = 
\frac{1}{d (2\pi)^d} \int\limits_{\R^d} \int\limits_{\mathbb S^{d-1}} 
\sum\limits_k \big[s_k\big( \CT_\infty(x, \om)\big)\big]^p\, d\om dx.
\end{align} 
\end{prop}

This proposition is a consequence of Theorem 2 from \cite{BS1977_1} and Remark 3 following 
this theorem.

\subsection{Integral operators} \label{subsect:iop}
We illustrate the utility of Proposition \ref{prop:homo} with the example of 
integral operators with homogeneous kernels. Let  
$\SQ\in\plainC\infty(\R^d\setminus\{0\})$ be a matrix-valued function such that 
\begin{align}\label{eq:homophi}
\SQ(tx) = t^{\a} \SQ(x),\quad x\not = 0, \a >-d, 
\end{align} 
for all $t >0$. Consider the integral operator $\iop(W)$ with the 
vector-valued kernel 
\begin{align}\label{eq:w}
W(x, y) = \CB(x) \SQ(x-y) \bbeta(x) a(y),
\end{align} 
where $a\in \plainC{}_0(\R^d)$ is scalar, $\CB\in \plainC{} (\R^d)$ is matrix-valued,   
$\bbeta\in \plainC\infty_0(\R^d)$ is a vector-function 
(i.e. an $(n\times 1)$-matrix-valued function with some $n$),  
and the dimensions of $\CB, \SQ$ and $\bbeta$ are appropriately matched so that 
the kernel $W(x, y)$ makes sense. 
We study spectral asymptotics of the operator $\iop(W)$ 
by reducing it to the operator of the form 
\eqref{eq:pdo}. 
Let $\t$ be as defined in \eqref{eq:sco}, \eqref{eq:sco1}, and let 
$R_0>0$ be a number such that 
\begin{align*}
W(x, y) = W(x, y) \t\big(|x-y| R^{-1}\big), \quad \textup{for all}\quad R\ge R_0.
\end{align*}
Consequently, the operator $\iop(W)$ has the form 
\eqref{eq:pdo} with the function  $X(x, \xi) = X_R(x, \xi) = X_R(\xi) \bbeta(x)$, where 
\begin{align}\label{eq:XR}
X_R(\xi) =  \int e^{-i\xi x} \t\big(|x| R^{-1}\big) \SQ(x) dx.
\end{align} 
Integrating by parts, we conclude that for each $\xi\not = 0$ 
the function $X_R(\xi)$ converges as $R\to\infty$ to a 
$\plainC\infty(\R^{d}\setminus\{0\})$-function
\begin{align}\label{eq:X0}
X_\infty(\xi) = \lim_{R\to\infty}X_R(\xi).
\end{align}
The function $X_\infty$ satisfies \eqref{eq:homo} with $\tau = \a + d$. Indeed, 
using \eqref{eq:homophi} write for $t >0$:
\begin{align}\label{eq:erelimit}
X_R(t\xi) = &\ t^{-\a-d}  
\int e^{-i\xi x} \t\big(|x|(Rt)^{-1}\big) \SQ(x) dx\notag\\
= &\ t^{-\a-d} X_{Rt}(\xi). 
\end{align}
Passing to the limit as $R\to\infty$, we get \eqref{eq:homo} with $\tau = \a + d$, as claimed. 
The equality \eqref{eq:erelimit} also implies that
\begin{align*}
X_R(t\xi) - X_\infty(t\xi) = t^{-\a-d}\big(X_{Rt}(\xi) - X_\infty(\xi)\big) = o(t^{-\a-d}), 
\quad t\to\infty, 
\end{align*}
for each $\xi\not = 0$ and $R>0$, which entails \eqref{eq:Xas}. Thus, applying 
Proposition \ref{prop:homo}, we obtain the following spectral asymptotics.

\begin{thm} \label{thm:w} 
Let the kernel $W$ be as defined in \eqref{eq:w}, and 
let $\a > -d$. Then the singular values of $\iop(W)$ satisfy the relation
\begin{align*}
\SfG_p(\iop(W)) = \sg_p(\iop(W)) = \frac{1}{d(2\pi)^d}
\int_{\R^d}\,|a(x)|^p\,\int_{\mathbb S^{d-1}}|\CB(x) X_\infty(\om) \bbeta(x)|^p\, d\om  \, dx, 
\end{align*}
where $p^{-1} = 1 + \a d^{-1}$.
\end{thm}

\begin{proof}
The matrix $\CT_\infty(x, \xi)$ is rank one and 
\begin{align*}
s_1\big(\CT_\infty(x, \xi)\big) = |a(x)|\, |\CB(x) X_\infty(\xi) \bbeta(x) |.
\end{align*}
The required formula follows from \eqref{eq:pdoas}.
\end{proof}

Consider two examples in which the above formula can be simplified.  

\begin{example}\label{ex:scal} 
Suppose that $\SQ$, $\CB = b$ and $\bbeta = \b$ are scalar functions, and that 
$\SQ(x) = |x|^{\a}$, $\a > -d$.   
Then (see, e.g. \cite[Ch. 2, Sect. 3.3]{GS1964})
\begin{align}\label{eq:xinf}
X_\infty(\xi) = 
\nu_{\a, d}\, |\xi|^{-(d+\a)},\quad \nu_{\a, d} =  
2^{d+\a}\pi^{\frac{d}{2}} 
\frac{\G\big(\frac{d+\a}{2}\big)}{\G\big(-\frac{\a}{2}\big)},\quad \a \not = 0, 2, 4, \dots, 
\end{align}
and $X_\infty(\xi) = 0$ for $\a = 0, 2, 4, \dots$. 
Thus, for $1/p = 1+\a/d$ and $\a\not = 0, 2, 4, \dots,$ we have 
\begin{align}\label{eq:Xint}
\frac{1}{d(2\pi)^d}\int_{\mathbb S^{d-1}} 
|X_\infty(\om)|^p d\om
= \frac{1}{d(2\pi)^d} \big|\nu_{\a, d}\big|^p \frac{2\pi^{\frac{d}{2}}}{\G\big(\frac{d}{2}\big)}
=  
\frac{ \big|\nu_{\a, d}\big|^p}{d\,  2^{d-1} \pi^{\frac{d}{2}}\G\big(\frac{d}{2}\big)}.
\end{align}
Now Theorem \ref{thm:w} yields 
\begin{align*}
\SfG_p(\iop(W)) = \sg_p(\iop(W)) = \frac{ \big|\nu_{\a, d}\big|^p}{d\,  2^{d-1} \pi^{\frac{d}{2}}\G\big(\frac{d}{2}\big)}\,\int_{\R^d} 
|a(x)b(x)\b(x)|^p  dx,\quad \frac{1}{p} = 1 + \frac{\a}{d}.
\end{align*}  
\end{example}

Note that the case of scalar functions $\SQ$ was studied in \cite{BS1970}, 
see also \cite[Theorem 10.9]{BS1977}. 
Next we consider an important example of a vector-valued function $\SQ$.

\begin{example}\label{ex:grad} 
Suppose that $\CB = b$ and $\bbeta = \b$ are scalar functions, and that 
$\SQ(x) = \nabla |x|^{\a+1} = (\a+1)|x|^{\a-1} x$, $\a > -d$. 
The vector-valued function $\SQ$ 
is homogeneous of order $\a$, and it follows from Example 
\ref{ex:scal} that 
\begin{align}\label{eq:xinfgrad}
X_\infty(\xi) = i\nu_{\a+1, d}
|\xi|^{-(\a+1+d)} \xi,
\quad \a \not = -1, 1, 3, \dots,
\end{align}
and $X_\infty(\xi) = 0$ for $\a = -1, 1, 3, \dots$. 
Thus, for $1/p = 1+\a/d$ and $\a\not = -1, 1, 3, \dots,$ we have 
\begin{align}\label{eq:Xintvec}
\frac{1}{d(2\pi)^d}\int_{\mathbb S^{d-1}} 
|X_\infty(\om)|^p d\om 
= 
 \frac{1}{d(2\pi)^d} \big|\nu_{\a+1, d}\big|^p 
 \frac{2\pi^{\frac{d}{2}}}{\G\big(\frac{d}{2}\big)}
 = \frac{ \big|\nu_{\a+1, d}\big|^p}{d\,  2^{d-1} \pi^{\frac{d}{2}}\G\big(\frac{d}{2}\big)}.
\end{align}
Now Theorem \ref{thm:w} yields 
\begin{align*}
\SfG_p(\iop(W)) = \sg_p(\iop(W)) = 
\frac{ \big|\nu_{\a+1, d}\big|^p}{d\,  2^{d-1} \pi^{\frac{d}{2}}\G\big(\frac{d}{2}\big)}\, 
\int_{\R^d} |a(x)b(x)|^p dx,\quad \frac{1}{p} = 1 + \frac{\a}{d}.
\end{align*}
\end{example} 
  
\section{Spectral asymptotics for the model problem}\label{sect:model} 
 
The objective of this section is to find the spectral asymptotics for a model 
integral operator. Recall that for any function $\mathcal K = \mathcal K(x, y)$, $x\in \R^n, y\in\R^d$, 
we denote by $\iop(\CK)$ the integral operator acting from $\plainL2(\R^d)$ into $\plainL2(\R^n)$.  
In each case the values of $n$ and $d$ are clear from the context. 
If $\CK(x, y)$ is $\mathbb C^s$-valued then the ``target" space 
$\plainL2(\R^n)$ is replaced by $\plainL2(\R^n; \mathbb C^s)$.

\subsection{The model operator}   
Let $\SQ\in\plainC\infty(\R^3\setminus \{0\})$ be an  $m\times n$-matrix-valued  
function $\SQ(x) = \{\SQ^{(l, s)}(x)\}$, $1\le l\le m$, $1\le s\le n$, (positively) 
homogeneous of order $\a>-3/2$. 

Let $\boldsymbol\b_{j,k} = \{\b_{j, k}^{(s)}\}_{s = 1,2,\dots, n}: \R^{3N}\mapsto \mathbb C^n$, 
$j = 1, 2, \dots, M$, $k = 1, 2, \dots, N-1$, be continuous 
vector-valued functions such that  
\begin{align}\label{eq:betajk}
|\p_x^{s}\, \boldsymbol \beta_{j, k}(\hat\bx, x)| \lesssim 
1+ (|x-x_k|\wedge 1)^{1 - |s|},\quad s\in\mathbb N_0^{3}.
\end{align}
Let 
\begin{align}\label{eq:rmodel}
\begin{lcases}
r = &\ 2, \quad \textup{if} \quad \a >0,\\
r > &\ 6(2\a+3)^{-1}\quad \textup{is arbitrary}, \quad \textup{if}\quad -3 < 2\a \le 0. 
\end{lcases}
\end{align}
Introduce also the scalar weight functions 
$a = a(x), b_{j, k} = b_{j, k}(\hat\bx)$ for the same range of values of $j$ and $k$ as above,
such that  
\begin{align}\label{eq:abbeta_as}
a\in\plainL{r}(\R^3),  \quad b_{j, k}\in \plainL2(\R^{3N-3}), \quad 
\textup{and are compactly supported}.  
\end{align}
Thus, automatically, $\4\, a\4_{\,r,q}<\infty$ for every $q>0$.
Now consider the vector-valued kernel $\CM(\hat\bx, x)$ with $m M$ components: 
\begin{align}\label{eq:cm_as}
\begin{lcases}
\CM(\hat\bx, x) = &\ \{\CM_j(\hat\bx, x)\}_{j=1}^M,\ 
\CM_j(\hat\bx, x) = \sum_{k=1}^{N-1} \CM_{j, k}(\hat\bx, x),\\
\CM_{j, k}(\hat\bx, x) = &\  
b_{j, k}(\hat\bx) \big(\SQ(x_k-x) \boldsymbol\b_{j, k}(\hat\bx, x)\big)\, a(x). 
\end{lcases}
\end{align}
Our aim is to find an asymptotic formula for the singular values of the operator 
$\iop(\CM): \plainL2(\R^3, \mathbb C)\to\plainL2(\R^{3N-3}; \mathbb C^{M m})$. 
The results on homogeneous kernels from Sect. \ref{sect:intop} 
are not applicable directly, since 
the number of ``target" variables (i.e. $3N-3$) 
is greater than the number of the input variables (i.e. $3$), unless $N=2$. 
The proof of Theorem \ref{thm:model_as} below amounts to reducing the operator 
$\iop(\CM)$ to a form for which Proposition \ref{prop:homo} can be used. 

As in \eqref{eq:XR} and \eqref{eq:X0} define 
\begin{align*}
X_R(\xi) = &\ \int e^{-i x\xi}\, \t(|x| R^{-1}) \, \SQ(x)\, dx,\ R>0,\\
&\ \textup{and}\quad X_\infty(\xi) = \lim_{R\to\infty} X_R(\xi). 
\end{align*}
For $t\in\R^3$ and $0\not = \xi\in\R^3$ define 
\begin{align*}
Y(t, \xi) = 
\sum_{j=1}^M \,|b_{j,1}(t)|^2 \,
|X_\infty(\xi) \boldsymbol\b_{j,1}(t, t)|^2, \,\ N=2,
\end{align*}
and
\begin{align}\label{eq:Y}
Y(t, \xi) = 
\sum_{j=1}^M \,\sum_{k=1}^{N-1}\int_{\R^{3N-6}}\, |b_{j,k}(\tilde\bx_{k, N}, t)|^2 \,
|X_\infty(\xi) \boldsymbol\b_{j,k}(\tilde\bx_{k, N}, t, t)|^2 
\,d\tilde\bx_{k, N},\ N\ge 3.
\end{align}
Here we have used 
the representation $(\hat\bx, x) = (\tilde\bx_{k, N}, x_k, x)$ introduced in \eqref{eq:xtilde}.  

\begin{thm}\label{thm:model_as} 
Let $\CM$ be the operator defined above, where   
$\SQ\in\plainC\infty(\R^3\setminus\{0\})$ is a homogeneous matrix-valued  
function of order $\a > -3/2$.  
Then the operator $\iop(\CM)$ belongs to $\BS_{p, \infty}$, with 
\begin{align}\label{eq:pmodel}
\frac{1}{p} = 1+ \frac{\a}{3},
\end{align}
and 
\begin{align}\label{eq:model} 
\SfG_p\big(\iop(\CM)\big)
=  \sg_p\big(\iop(\CM)\bigl)
=  \frac{1}{3(2\pi)^3}\, \int_{\R^3}\, \int_{\mathbb S^2}\, 
|a(x)|^p\, |Y(x, \om)|^{\frac{p}{2}}
\, d\om\, dx.
\end{align}
\end{thm} 

Note that the condition $\a>-3/2$ guarantees that $p <2$. 

First we estimate the $\BS_{p, \infty}$-norm of each of the operators $\iop(\CM_{j,k})$.

\begin{lem}\label{lem:cmjk} 
Assume that $\SQ$ is $\a$-homogeneous with some $\a>-3/2$. 
Let $1/p = 1 + \a/3$ and let $r$ be as defined in \eqref{eq:rmodel}.
Then $\iop(\CM_{j, k})\in \BS_{p, \infty}$, $j=1, 2, \dots, M$, $k = 1, 2, \dots, N-1$, and 
\begin{align}\label{eq:cmjk}
\| \iop(\CM_{j, k})\|_{p, \infty}\lesssim \|b_{j, k}\|_{\plainL2} \, 
\4\, a\4_{\,r,p},
\end{align}
where the implicit constant depends on 
the function 
$\SQ$, the constants in \eqref{eq:betajk}, 
and the radii of supports of $a$ and $b_{j, k}$, 
but does not depend on the functions $a$ and $b_{j, k}$ otherwise.

If $\boldsymbol\b_{j, k}(\hat\bx, x_k) = 0$, then 
$\iop(\CM_{j, k})\in \BS_{q, \infty}$ with $1/q = 1 + (\a+1)/3$, so that 
$\SfG_p\big(\iop(\CM_{j, k})\big)= 0$.
\end{lem} 
 
\begin{proof} Since $\SQ$ is $\a$-homogeneous, we have 
\begin{align*}
|\p_x^j\,\SQ(x_k-x)|\lesssim |x-x_k|^{\a-|j|},\quad j\in\mathbb N_0^3.
\end{align*} 
Since the kernel 
$\CM_{j, k}$ is compactly supported, on the right-hand side we can replace 
$|x-x_k|$ with $|x-x_k|\wedge 1$. Taking into account also 
\eqref{eq:betajk} we see that the conditions of Corollary \ref{cor:product} 
are satisfied with $T_1 = \SQ$ and $T_2 = \boldsymbol\b_{j, k}$. Consequently, 
$\iop(\CM_{j, k})\in \BS_{p, \infty}$, and the bound 
\eqref{eq:cmjk} holds, as claimed. 

If $\boldsymbol\b_{j, k}(\hat\bx, x_k) = 0$, then  the 
inclusion 
$\iop(\CM_{j, k})\in \BS_{q, \infty}$ follows from the second part of 
Corollary \ref{cor:product}. 
\end{proof}

This lemma has two important consequences:
\begin{enumerate}
\item  
It suffices to prove Theorem \ref{thm:model_as} for functions 
$a, b_{j, k}$ that satisfy 
\begin{align}\label{eq:binfty}
a\in\plainC{\infty}_0(\R^3), \quad b_{j, k}\in \plainC{\infty}_0(\R^{3N-3}), 
\end{align}
instead of \eqref{eq:abbeta_as}. 
Indeed, let $b_{j, k}^{(\nu)} \in\plainC\infty_0, \nu >0$, be a family of functions such that 
$\|b_{j, k}^{(\nu)} - b_{j, k}\|_{\plainL2}\to 0$ as $\nu\to 0$, and such that 
the supports of all 
$b_{j, k}^{(\nu)}$ and $b$ are contained in the same ball of fixed radius independent of $\nu$. 
Then, denoting by $\CM_{j, k}^{(\nu)}$ 
the corresponding kernels \eqref{eq:cm_as}, we obtain from 
Lemma \ref{lem:cmjk} that 
\begin{align*}
\|\iop(\CM_{j, k}^{(\nu)}) - \iop(\CM_{j, k})\|_{p, \infty}
\lesssim \|b_{j, k}^{(\nu)} - b_{j, k}\|_{\plainL2}   
\, \bigg[\sum_{n\in \Z^3}\|a\|_{\plainL{r}(\CC_n)}^p\bigg]^{\frac{1}{p}}\to  0,\quad \nu\to 0.
\end{align*}
Thus, by \eqref{eq:triangle} and \eqref{eq:blockvec}, 
$\|\iop(\CM) - \iop(\CM^{(\nu)})\|_{p, \infty}\to 0$,   
as $\nu\to 0$.  Now Corollary~\ref{cor:zero1} implies that 
\begin{align*}
\SfG_p\big(\iop(\CM^{(\nu)})\big) \to \SfG_p\big(
\iop(\CM)\big),\quad 
\sg_p\big(\iop(\CM^{(\nu)})\big) \to 
\sg_p\big(\iop(\CM)\big),\quad \nu\to 0.
\end{align*} 
In the same way one proves that $a$ may be assumed to satisfy \eqref{eq:binfty}. 
\item 
It suffices to prove Theorem \ref{thm:model_as} for vectors 
$\boldsymbol\b_{j, k}$ that are independent of $x$, i.e. 
for 
\begin{align}\label{eq:xind}
\boldsymbol\b_{j, k}(\hat\bx, x) = \boldsymbol\b_{j, k}(\hat\bx).
\end{align} 
Indeed, in the definition of $\CM_{j, k}$ 
replace $\boldsymbol\b_{j, k}(\hat\bx, x)$ 
with $\widetilde{\boldsymbol\b_{j, k}}(\hat\bx)= \boldsymbol\b_{j, k}(\hat\bx, x_k)$ 
and denote the associated operator by $\widetilde\CM_{j, k}$. 
Since  
$\boldsymbol\b_{j, k}(\hat\bx, x) - \widetilde{\boldsymbol\b_{j, k}}(\hat\bx)$ vanishes at 
$x = x_k$, by Lemma \ref{lem:cmjk} we have  $\CM_{j, k} - \widetilde\CM_{j, k}\in \BS_{q, \infty}$. 
Since $q<p$, this means that 
$\SfG_p\big(\CM_{j, k} - \widetilde\CM_{j, k}\big) = 0$.
Thus Corollary \ref{cor:zero} implies that 
\begin{align*}
\SfG_p\big(\iop(\CM)\big) = 
\SfG_p\big(\iop(\widetilde\CM)\big),\quad 
\sg_p\big(\iop(\CM)\big) = \sg_p\big(\iop(\widetilde\CM)\big).
\end{align*}
\end{enumerate}

Combining the observations (1) and (2), in the 
proof below we always assume that 
\eqref{eq:xind} and \eqref{eq:binfty} hold. 
Moreover, we absorb $b_{j, k}$ into $\boldsymbol\b_{j, k}$. 
Therefore, the kernels $\CM_{j, k}$ take the form
\begin{align}\label{eq:smoothm}
\CM_{j, k}(\hat\bx, x) =  
\big(\SQ(x_k-x) \boldsymbol\b_{j, k}(\hat\bx)\big)\, a(x), 
\end{align}
with $a\in \plainC\infty_0(\R^3)$ and $\boldsymbol\b_{j, k}\in \plainC\infty_0(\R^{3N-3})$. 

\subsection{Proof of Theorem \ref{thm:model_as} for $N=2$}
By \eqref{eq:smoothm} the kernel \eqref{eq:cm_as} has the form 
\begin{align*}
\CB(x_1) \big(\widetilde{\SQ}(x_1-x) \bbeta(x_1)\big) a(x) 
\end{align*}
with $\CB(x_1) = I$, 
\begin{align*}
\textup{$M\times M$-block-matrix}\quad 
\widetilde\SQ(t) = 
\begin{pmatrix}
\SQ(t) & 0 & \cdots&0\\
0 & \SQ(t)& \dots&0\\
\vdots & \vdots &\ddots &\vdots\\
0&0&\cdots & \SQ(t)
\end{pmatrix}
\quad 
\textup{and}\quad 
\bbeta(x_1) =  
\begin{pmatrix}
\bbeta_{1, 1}(x_1)\\
\bbeta_{2, 1}(x_1)\\
\vdots\\
\bbeta_{M, 1}(x_1)
\end{pmatrix}.
\end{align*}
Now \eqref{eq:model} follows directly from Theorem \ref{thm:w}. 

\subsection{Proof of Theorem \ref{thm:model_as} for $N\ge 3$ } \label{subsect:model_as}
We still assume that $\CM_{j, k}$ have the form \eqref{eq:smoothm}. 

\begin{lem}\label{lem:cross_as} 
For each $j = 1, 2, \dots, M$ and each pair $k, l = 1, 2, \dots, N-1$, $k\not = l$, we have 
\begin{align*}
\SfG_{p/2}\big(\iop(\CM_{j,k})^*\iop(\CM_{j,l})\big) = 0.
\end{align*}
\end{lem}

\begin{proof} 
Fix a $j = 1, 2, \dots, M$ and write the kernel $\CP_{k, l}(x, y)$ of the operator 
$\iop(\CM_{j,k})^*\iop(\CM_{j,l})$:
\begin{align*}
\CP_{k,l}(x, y) 
= &\ \overline{a(x)}a(y)
\int  \big(\overline{\SQ(x_k-x) \boldsymbol\b_{j, k}(\hat\bx)} 
\big) 
\cdot \big(\SQ(x_l-y)  \boldsymbol\b_{j, l}(\hat\bx)\big)\, d\hat\bx.
\end{align*}
Write $\hat\bx = (\tilde\bx_{l, N}, x_l)$, $d\hat\bx = d\tilde\bx_{l, N} dx_l$  
and change $x_l$ to $x_l+y$, so that
\begin{align*}
\CP_{k,l}(x, y) 
= &\ \overline{a(x)}a(y)
\int  \big(\overline{\SQ(x_k-x) \boldsymbol\b_{j, k}(\tilde\bx_{l, N}, x_l+y)} 
\big) 
\cdot \big(\SQ(x_l)  \boldsymbol\b_{j, l}(\tilde\bx_{l, N}, x_l+y)\big)\, 
d\tilde\bx_{l, N}\, dx_l.
\end{align*}
Because of the assumptions $a\in \plainC\infty_0(\R^3)$ and 
$\boldsymbol\b_{j, k}\in \plainC\infty_0(\R^{3N-3})$, 
for all $x\in \R^3$ 
the kernel $\CP_{k,l}$ is a $\plainC\infty_0$-function of $y\in\R^3$. 
Hence by Proposition \ref{prop:BS} the singular values of the operator $\iop(\CP_{k,l})$ decay 
faster than any negative power of their number.  In particular, 
$\SfG_{p/2}\big(\iop(\CP_{k,l})\big) = 0$, as required.
\end{proof}

By Corollary \ref{cor:zero}, Lemma \ref{lem:cross_as} entails that   
\begin{align*}
\SfG_{p/2}\big(\iop(\CM)^* \iop(\CM)\big) 
= \SfG_{p/2}\bigg(\sum_{j=1}^M\, \sum_{k=1}^{N-1} \iop(\CM_{j,k})^*\iop(\CM_{j,k})\bigg),
\end{align*}
and the same equality holds for the functional $\sg_{p/2}$.
Let us write the kernel 
$\CF(x, y)$ 
of the operator on the right-hand side:
\begin{align*}
\CF(x, y) &\ = \overline{a(x)}a(y)
\sum_{j=1}^M\, \sum_{k=1}^{N-1} \int\limits_{\R^{3N-3}}  \big(\overline{\SQ(x_k-x) \boldsymbol\b_{j, k}(\hat\bx)} 
\big) 
\cdot \big(\SQ(x_k-y)  \boldsymbol\b_{j, k}(\hat\bx)\big)\, d\hat\bx\\
= &\ \overline{a(x)}a(y)
\sum_{j=1}^M\, \sum_{k=1}^{N-1} \int\limits_{\R^3}\int\limits_{\R^{3N-6}}  
 \big(\overline{\SQ(t-x) \boldsymbol\b_{j, k}((\tilde\bx_{k, N}, t)} 
\big) 
\cdot \big(\SQ(t-y)  \boldsymbol\b_{j, k}((\tilde\bx_{k, N}, t)\big)\, 
d\tilde\bx_{k, N}\, dt.
\end{align*}
To rewrite the right-hand side in a more compact form introduce 
the ($n\times n$)-matrix-function $\SB(t) = \{ B^{(l, s)}(t) \}_{l, s = 1}^n$: 
\begin{align*}
B^{(l, s)}(t)
 = \sum_{j=1}^M\, \sum_{k=1}^{N-1}\int_{\R^{3N-6}}\, 
\overline{\b_{j, k}^{(s)}(\tilde\bx_{k, N}, t)}
 \b_{j, k}^{(l)}(\tilde\bx_{k, N}, t) 
\,d\tilde\bx_{k, N},
\end{align*}
so that 
\begin{align*}
\CF(x, y) = 
 \overline{a(x)}a(y) 
\int_{\R^3} \tr\, \big( {\sf B}(t) \SQ^*(t-x) \SQ(t-y)\big)\, dt.
\end{align*}
Since both $a$ and $\sf B$ are compactly supported, the above integral 
coincides with
\begin{align*}
\CF(x, y) = 
 \overline{a(x)}a(y) &\ 
\int_{\R^3} \tr\, \big( {\sf B}(t) \SQ^*(t-x) \SQ(t-y)\big)\, \\
&\ \qquad \times \t(|t-x| R^{-1}) \,\t(|t-y| R^{-1})\, dt,
\end{align*} 
for all sufficiently large $R>0$. Split $\CF$ into two kernels: $\CF = \CF_0+\CF_1$, where 
\begin{align*}
\CF_0(x, y) = 
 \overline{a(x)}a(y) &\ 
\int_{\R^3} \tr\, \big( {\sf B}(x) \SQ^*(t-x) \SQ(t-y)\big)\, \\
&\ \qquad \times \t(|t-x| R^{-1}) \,\t(|t-y| R^{-1})\, dt,\\
\CF_1(x, y) = 
 \overline{a(x)}a(y) &\ 
\int_{\R^3} \tr\, \big[ \big(\SB(t) - \SB(x)\big)\SQ^*(t-x) \SQ(t-y)\big]\, \\
&\ \qquad \times \t(|t|R_1^{-1})^2 \t(|t-x| R^{-1}) \,\t(|t-y| R^{-1})\, dt.
\end{align*} 
The cut-off $\t(|t|R_1^{-1})^2$ (with a sufficiently large $R_1>0$) could be inserted since 
$a$ and $\SB$ is compactly supported. 
Let us show that 
the operator $\iop(\CF_1)$ gives a zero contribution to the asymptotics. Indeed, represent 
$\iop(\CF_1)$ as the sum $\sum_{l, s =1}^n \sum_{r = 1}^m
\iop(\CR_1^{(l, s, r)})\,\iop(\CR_2^{(r, l)})$ where
\begin{align*}
\CR_1^{(l, s, r)}(x, t) = &\ \overline{a(x)} \, \t(|t|R_1^{-1})\, 
  \overline{\CT^{(r, s)}(t-x)}
  \, \big(B^{(l, s)} (t) - B^{(l, s)} (x)\big),\\
\CR_2^{(r, l)}(t, y) = &\ \, \t(|t|R_1^{-1})\, a(y)\, \CT^{(r, l)}(t-y),
\quad\textup{where}\quad 
 \CT^{(r, s)}(x) = \SQ^{(r, s)}(x)  
 \t(|x| R^{-1}). 
\end{align*}
To determine to which compact class the operator 
$\iop(\CF_1)$ belongs we consider each 
of the above factors separately and then use \eqref{eq:multipl}. By Corollary \ref{cor:homo},
$\op(\CR_2^{(r, l)})\in \BS_p$ with the $p$ as in \eqref{eq:pmodel}.  
Let us consider now $\CR_1^{(l, s, r)}$. 
Since $B^{(l, s)} (t) - B^{(l, s)} (x)$ vanishes for $t = x$, 
applying Corollary 
\ref{cor:product}, 
we conclude that 
$\op(\CR_1^{(l, s, r)})\in \BS_{q}$
with $1/q = 1 + {(\a+1)}/{3}$. By \eqref{eq:multipl} and \eqref{eq:triangle}, 
$\iop(\CF_1)\in \BS_{q_1},\ 1/{q_1} = 2+ {(2\a+1)}/{3}$. As 
$q_1 < p/2$, Corollary \ref{cor:zero} ensures that  
\begin{align*}
\SfG_{p/2} (\iop(\CF)) = \SfG_{p/2} (\iop(\CF_0)), 
\quad \sg_{p/2} (\iop(\CF)) = \sg_{p/2} (\iop(\CF_0)).
\end{align*}
Rewrite the kernel $\CF_0$ in the form
\begin{align*}
\CF_0(x, y) = 
 \overline{a(x)}a(y) &\ 
\int_{\R^3} \tr\, \big( {\sf B}(x) \SQ^*(t) \SQ(t+x-y)\big)\, \\
&\ \qquad \times \t(|t| R^{-1}) \,\t(|t + x-y| R^{-1})\, dt.
\end{align*}
Taking Fourier transform, we can now rewrite 
the operator $\iop(\CF_0)$ as a pseudo-differential operator 
in $\plainL2(\R^3)$ of the form 
\eqref{eq:pdo} with 
\begin{align*}
\CA(x) = a(x), \,\CB(x) = \overline{a(x)},
\end{align*}
and symbol
\begin{align*}  
\tr\, \big( {\sf B}(x) X_R^*(\xi) X_R(\xi)\big) = &\ Y_R(x, \xi)
: = 
\sum_{j=1}^M \,\sum_{k=1}^{N-1}\int_{\R^{3N-6}}\, 
|X_R(\xi) \boldsymbol\b_{j,k}(\tilde\bx_{k, N}, t, t)|^2 
\,d\tilde\bx_{k, N},
\\ \quad 
\textup{where}\quad X_R(\xi) = &\ \int e^{-i\xi x} \, \t(|x| R^{-1})\, \SQ(x)\, dx.
\end{align*}
It is clear that the function $X_\infty(\xi) = \lim_{R\to\infty} X_R(\xi)$  
is homogeneous of order $(-\a-d)$ and that 
\begin{align*}
Y_R(x, \xi) 
= 
Y(x, \xi) 
+ o(|\xi|^{-2(\a+d)}),\quad |\xi|\to\infty,
\end{align*}
uniformly in $x\in\R^3$, where $Y(x, \xi)$  is defined in \eqref{eq:Y}. 
Thus by Proposition \ref{prop:homo}, 
\begin{align*}
\SfG_{p/2} (\iop(\CF_0)) = \sg_{p/2} (\iop(\CF_0))
= \frac{1}{3(2\pi)^3}\, \int_{\R^3}\, \int_{\mathbb S^2}\, 
|a(x)|^p\, |Y(x, \om)|^{\frac{p}{2}}
\, d\om\, dx.
\end{align*}
In view of the identities 
\begin{align*}
\SfG_p(\iop(\CM)) = 
\SfG_{p/2} (\iop(\CF)) = 
\SfG_{p/2} (\iop(\CF_0)),\\
\sg_p(\iop(\CM)) = 
\sg_{p/2} (\iop(\CF)) = 
\sg_{p/2} (\iop(\CF_0)),
\end{align*}
the above formula completes the proof of Theorem \ref{thm:model_as}.
\qed

\subsection{Three examples} 
Here we illustrate the 
utility of Theorem \ref{thm:model_as} by three examples that will 
be important in the proof of Theorems \ref{thm:maingk} and \ref{thm:maingkasym}.  
As in Section \ref{sect:model}, the (matrix-valued) function $\SQ$ 
is assumed to be $\a$-homogeneous, where $\a > -3/2$, and the parameter $p\in (0, 2)$ 
is defined by 
\begin{align*}
\frac{1}{p} = 1 + \frac{\a}{3}.
\end{align*}

\begin{example}\label{ex:vect}
 Suppose that $\boldsymbol\beta_{j, k} = \b_{j, k}$ are scalar functions and 
\begin{align*}
\SQ = 
\begin{pmatrix}
Q^{(1)}\\Q^{(2)}\\ \vdots\\ Q^{(m)}\\
\end{pmatrix}
\end{align*}
is an $\a$-homogeneous vector-function. 
Then with the notation 
\begin{align*}
B(t) = \sum_{j=1}^M \, |b_{j,k}(t)|^2 |\b_{j,k}(t, t)|^2, 
\quad\textup{if}\quad N = 2,
\end{align*}
and 
\begin{align*}
B(t) = \sum_{j=1}^M \,\sum_{k=1}^{N-1}\int_{\R^{3N-6}}\, |b_{j,k}(\tilde\bx_{k, N}, t)|^2 
|\b_{j,k}(\tilde\bx_{k, N}, t, t)|^2 
\,d\tilde\bx_{k, N}, \quad \textup{if}\quad N\ge 3,
\end{align*}
the formula \eqref{eq:Y} becomes 
\begin{align*}
Y(t, \xi) = B(t) |X_\infty(\xi)|^2. 
\end{align*}
Therefore, by Theorem \ref{thm:model_as}, 
\begin{align*} 
\SfG_p\big(\iop(\CM)\big)
= &\ \sg_p\big(\iop(\CM)\bigl)\notag\\ 
= &\ \frac{1}{3(2\pi)^3}\, \int_{\R^3}\, |a(x)|^p\, 
|B(x)|^{\frac{p}{2}}\, dx \int_{\mathbb S^2}\, 
\, |X_\infty(\om)|^p\, d\om.
\end{align*}
Consider two important special cases of this formula.

\begin{enumerate}
\item 
Let $m = 1$ and $\SQ(x) = |x|^\a$, $\a\not = 0, 2,\dots$, as in Example \ref{ex:scal}. Using the formula 
\eqref{eq:Xint} with $d=3$ for the function $X_\infty$, we obtain that 
\begin{align} \label{eq:scalar}
\SfG_p\big(\iop(\CM)\big)
=  \sg_p\big(\iop(\CM)\bigl) 
=  \frac{|\nu_{\a, 3}|^p}{6\pi^2}
\, \int_{\R^3}\, |a(x)|^p\, 
|B(x)|^{\frac{p}{2}}\, dx.
\end{align}
\item 
Let $m = 3$ and 
$\SQ^{(1)}(x) = \nabla |x|^{\a+1} = (\a+1)|x|^{\a-1}x$, $\a\not = -1, 1, 3, \dots,$ as in 
Example \ref{ex:grad}. Using the formula 
\eqref{eq:Xintvec} with $d=3$ for the function $X_\infty$, we obtain that 
\begin{align} \label{eq:vector}
\SfG_p\big(\iop(\CM)\big)
= \sg_p\big(\iop(\CM)\bigl)
= \frac{|\nu_{\a+1, 3}|^p}{6\pi^2}
\, \int_{\R^3}\, |a(x)|^p\, 
|B(x)|^{\frac{p}{2}}\, dx.
\end{align}
\end{enumerate}
\end{example}

In the next two examples our objective 
is to write a formula $Y(t, \xi)$, defined in \eqref{eq:Y}. 
Substituting $Y(t, \xi)$ into \eqref{eq:model} gives the sought asymptotics.

\begin{example}\label{ex:scaltriplet} 
Suppose that $\boldsymbol\b_{j, k}$ are vector-functions with three components and 
that 
\begin{align*}
\SQ(x) =  \nabla^T\,|x|^{\a+1}
=  (\a+1)|x|^{\a-1}\,\begin{pmatrix}
x^{(1)}&x^{(2)}& x^{(3)} 
\end{pmatrix},\ \a\not = -1, 1, 3, \dots,
\end{align*}
so that $X_\infty(\xi)$ is given by \eqref{eq:xinfgrad}:
\begin{align*}
X_\infty(\xi) =  i \nu_{\a+1, 3}
\, |\xi|^{-(\a+4)}\begin{pmatrix}
\xi^{(1)}& \xi^{(2)}& \xi^{(3)}
\end{pmatrix}.
\end{align*}
Substituting this into \eqref{eq:Y} we obtain that 
\begin{align*}
Y(t, \xi) = \frac{|\nu_{\a+1, 3}|^2}{|\xi|^{2(\a+4)}} \, 
\sum_{j=1}^M \, |b_{j,1}(t)|^2 \,
|\xi\cdot \boldsymbol\b_{j,1}(t, t)|^2,\ N = 2,
\end{align*}
and 
\begin{align}\label{eq:Yvec}
Y(t, \xi) = \frac{|\nu_{\a+1, 3}|^2}{|\xi|^{2(\a+4)}} \, 
\sum_{j=1}^M \,\sum_{k=1}^{N-1}\int_{\R^{3N-6}}\, |b_{j,k}(\tilde\bx_{k, N}, t)|^2 \,
|\xi\cdot \boldsymbol\b_{j,k}(\tilde\bx_{k, N}, t, t)|^2 
\,d\tilde\bx_{k, N},\ N\ge 3.
\end{align} 
\end{example}

\begin{example}\label{ex:matriplet}

Suppose that $\SQ(x)$ is the matrix 
\begin{align*}
\SQ(x) = \nabla \nabla^T\, |x|^{\a+2},\quad \a \not = -2, 0, 2, \dots,
\end{align*}
i.e.
\begin{align*}
\SQ^{(l, s)}(x) =   \p_{x^{(l)}} \p_{x^{(s)}} |x|^{\a+2},\quad l, s = 1, 2, 3.
\end{align*}
According to \eqref{eq:xinf}, the Fourier transform $X_\infty(\xi)$ for this matrix-function is given by
\begin{align*}
X_\infty(\xi) = - \nu_{\a+2, 3}\, |\xi|^{-(\a+5)}
\begin{pmatrix}
(\xi^{(1)})^2 & \xi^{(1)}\xi^{(2)} & \xi^{(1)} \xi^{(3)}\\[0.2cm]
\xi^{(2)}\xi^{(1)} & (\xi^{(2)})^2 & \xi^{(2)}\xi^{(3)}\\[0.2cm]
\xi^{(3)}\xi^{(1)} & \xi^{(3)}\xi^{(2)} & (\xi^{(3)})^2
\end{pmatrix}.
\end{align*}
Consequently, 
\begin{align*}
|X_\infty(\xi)\boldsymbol\b_{j, k}(\hat\bx, x)|^2 =  
\frac{|\nu_{\a+2, 3}|^2}{|\xi|^{2(\a+4)}}\, |\xi\cdot \boldsymbol\b_{j, k}(\hat\bx, x)|^2.
\end{align*}
After substitution into \eqref{eq:Y} we obtain 
\begin{align*}
Y(t, \xi) = \frac{|\nu_{\a+2, 3}|^2}{|\xi|^{2(\a+4)}} \, 
\sum_{j=1}^M \, |b_{j,1}(t)|^2 \,
|\xi\cdot \boldsymbol\b_{j,1}(t, t)|^2,\ N = 2,
\end{align*}
and
\begin{align}\label{eq:Ymat}
Y(t, \xi) = \frac{|\nu_{\a+2, 3}|^2}{|\xi|^{2(\a+4)}} \, 
\sum_{j=1}^M \,\sum_{k=1}^{N-1}\int_{\R^{3N-6}}\, |b_{j,k}(\tilde\bx_{k, N}, t)|^2 \,
|\xi\cdot \boldsymbol\b_{j,k}(\tilde\bx_{k, N}, t, t)|^2 
\,d\tilde\bx_{k, N},\ N\ge 3.
\end{align}

\end{example}

\section{Factorization of $\SG$ and $\SK$: restating Theorems \ref{thm:maingk} and 
\ref{thm:maingkasym}  }\label{sect:factor}

In this section we begin the proof of Theorems \ref{thm:maingk} and 
\ref{thm:maingkasym} by restating them first 
using a natural factorization of the 
operators $\SG$ and $\SK$ which will be specified in Sect.~\ref{subsect:fact}. 

Our first step however is to prove that under the conditions of Theorems 
\ref{thm:maingk} and \ref{thm:maingkasym} 
the asymptotic coefficients 
\eqref{eq:coeffAB} and \eqref{eq:coeffABasym} are finite. 

\subsection{Asymptotic coefficients}\label{subsect:ascoeff}
Let $M(\bx, R; \psi)$ be as defined in \eqref{eq:locl2},
and let 
\begin{align*}
M(\bx; \psi) : = M(\bx, 1; \psi) = \|\psi\|_{\plainL2(B(\bx, 1))},\quad \bx\in \R^{3N}.
\end{align*}
Recall the notation \eqref{eq:lattice} for the lattice quasi-norms. 

\begin{lem}\label{lem:holder}
Suppose that $\4\, \rho \4_{\,1, q}<\infty$ with some $q\in (0, 1]$. If $N=2$, then 
\begin{align*}
\int_{\R^3}\, M(x, x; \psi)^{2q}\, dx
\lesssim \4\, \rho \4_{\,1, q}^{\, q}.
\end{align*}
If $N\ge 3$, then for all $j, k = 1, 2, \dots, N-1$, $j <k$, we have
\begin{align*}
\int_{\R^3}\, \bigg( \int_{\R^{3N-6}}
\,\big(M(\tilde\bx_{j, k}, x, x; \psi)\big)^2   \, d\tilde\bx_{j, k}\bigg)^q\, dx
\lesssim \4\, \rho \4_{\,1, q}^{\, q}.
\end{align*}
\end{lem}

\begin{proof} We prove the lemma only for $N\ge 3$. The assumption $N=2$ only simlipfies 
the calculations. 
We omit $\psi$ from the notation and 
denote the integral on the left-hand side by $I$. We use the familiar notation  
$\CC^{(r)}_n = [-r/2, r/2)^3+n, n\in\Z^3, r >0$, see \eqref{eq:cube}. 
By H\"older's inequality, 
\begin{align*}
I = &\ \sum_{n\in\Z^3}\, \int_{\CC^{(1)}_n}\, \bigg( \int_{\R^{3N-6}}
\,\big(M(\tilde\bx_{j, k}, x, x)\big)^2   \, d\tilde\bx_{j, k}\bigg)^q\, dx\\[0.2cm]
\le &\ \sum_{n\in\Z^3}\, \bigg[\int_{\CC^{(1)}_n}\, \int_{\R^{3N-6}}
\,\big(M(\tilde\bx_{j, k}, x, x)\big)^2   \, d\tilde\bx_{j, k}\, dx\bigg]^{q}.
\end{align*}
Estimate:
\begin{align*} 
\int_{\CC_n^{(1)}}  &\  \int_{\R^{3N-6}}\,
\big(M(\tilde\bx_{j, k}, x, x)\big)^2\, d\tilde\bx_{j, k} 
dx\\
  \le  &\  
  \int\limits_{\CC_n^{(1)}}  \int\limits_{\R^{3N-6}}\bigg[ \int\limits_{|w-x|<1}
 \,  \int\limits_{|t-x| < 1} \,  \int\limits_{|\tilde\by_{j, k} - \tilde\bx_{j, k}|<1 } 
  |\psi(\tilde\by_{j, k}, t, w)|^2\, d\tilde\by_{j, k} dt dw\bigg]
  d\tilde\bx_{j, k}
  \, dx\\
  \lesssim &\ \int\limits_{\CC_n^{(3)}}  
 \,  \int\limits_{\R^3} \,  \int\limits_{\R^{3N-6}} 
  |\psi(\tilde\by_{j, k}, t, w)|^2\, d\tilde\by_{j, k} dt dw
  = \int\limits_{\CC_n^{(3)}}  
 \,  \int\limits_{\R^{3N-3}} 
  |\psi(\hat\by_k, w)|^2\, d\hat\by_k dw \le \| \rho\|_{\plainL1(\CC_n^{(3)})}.
 \end{align*}
Therefore,
\begin{align*} 
I\lesssim \sum_{n\in\Z^3} \| \rho\|_{\plainL1(\CC_n^{(3)})}^q 
\lesssim \sum_{n\in\Z^3} \| \rho\|_{\plainL1(\CC_n^{(1)})}^q 
= \4\, \rho\4_{\, 1, q}^{\,q} < \infty,   
\end{align*}
as claimed.
\end{proof}

Now we can estimate the coefficients \eqref{eq:coeffAB} and \eqref{eq:coeffABasym}:

 \begin{lem}\label{lem:coeffs} 
 The coefficients $A$ and $B$ satisfy the bounds 
 \begin{align*}
 A\lesssim \4\, \rho \4_{\,1, 3/8}^{\, 3/8},\quad B\lesssim \4\, \rho \4_{\,1, 1/2}^{\, 1/2}. 
 \end{align*}
If the function $\psi$ satisfies \eqref{eq:antisym}, then 
the coefficients $A_{\textup{asym}}$ and $B_{\textup{asym}}$ satisfy the bounds 
  \begin{align*}
A_{\textup{asym}}\lesssim \4\, \rho\4_{\, 1, 3/10}^{\,3/10},
\quad B_{\textup{asym}} \lesssim \4\, \rho\4_{\, 1, 3/8}^{\,3/8}.
\end{align*}
\end{lem}

\begin{proof}
By \eqref{eq:grpsi}, for all $j<k$, we have 
\begin{align*}
|\psi(\tilde\bx_{j, k}, x, x)|\lesssim M(\tilde\bx_{j, k}, x, x; \psi),  
\end{align*}
By definitions \eqref{eq:H} and \eqref{eq:coeffAB}, 
the required bounds follow from Lemma \ref{lem:holder} for $q = 3/8$ 
and $q = 1/2$ respectively. 
 
Now assume that $\psi$ satisfies \eqref{eq:antisym}. 
By Remark \ref{rem:grad}, the function $\sv(\bx)=\nabla_x\psi(\hat\bx, x)$ 
is continuous on $\R^{3N}\setminus\{\bx = (\hat\bx, 0)\}$.
Together with \eqref{eq:grpsi} this implies that 
\begin{align*}
|\sv(\bx)|
\lesssim M(\bx; \psi), \quad \textup{for all}\quad \bx\in\R^{3N} \quad \textup{s.t.}\quad x \not = 0.  
\end{align*}
By definitions \eqref{eq:E} and \eqref{eq:coeffABasym}, 
the required bounds follow from Lemma \ref{lem:holder} with $q = 3/10$ and 
$q = 3/8$ respectively. 
\end{proof}

 \subsection{Factorization of the operators $\SG$ and $\SK$} 
\label{subsect:fact}
Let us describe the factorization of the operators  
$\sf \G$ and  $\SK$, 
which is central to the proof of Theorems  
\ref{thm:maingk} and \ref{thm:maingkasym}. 

\subsubsection{General case: Theorem \ref{thm:maingk}}
Introduce the new set of functions
\begin{align}\label{eq:psij}
\begin{lcases}
\psi_j(\hat\bx, x) = &\  
\psi(x_1, \dots, x_{j-1}, x, x_j, \dots, x_{N-1})\\
\sv_j(\hat\bx, x) = &\  \nabla_x\psi_j(\hat\bx, x)
\end{lcases}
,\quad j = 1, 2, \dots, N,
\end{align}
depending on the fixed pair of variables $\hat\bx$ and $x$. Thus the definitions 
\eqref{eq:den} and 
 \eqref{eq:tau} rewrite in the form:
\begin{align*}
\g(x, y) = \sum_{j=1}^N
\int_{\R^{3N-3}} \overline{\psi_j(\hat\bx, x)}
 \psi_j(\hat\bx, y)\, d\hat\bx,
\end{align*}
and 
\begin{align*}
\vark(x, y) = &\ \sum_{j=1}^N \int_{\R^{3N-3}} 
\overline{\sv_j(\hat\bx, x)} \cdot \sv_j(\hat\bx, y) d\hat\bx.
\end{align*}
Therefore the operators $\SG = \iop(\g)$ and $\SK = \iop(\vark)$ can be represented  
in the form 
\begin{align}\label{eq:factorize}
\SG = \iop({\sf\Psi})^*\, \iop({\sf \Psi}),\quad 
\SK = \iop(\SV)^*\, \iop(\SV),
\end{align}
where 
$\iop({\sf\Psi}): \plainL2(\R^3)\mapsto\plainL2(\R^{3N-3}; \mathbb C^N)$ and 
$\iop(\SV):\plainL2(\R^3)\mapsto\plainL2(\R^{3N-3};\mathbb C^{3N})$ 
are operators with vector-valued kernels given by 
\begin{align}\label{eq:bpsi}
{\sf\Psi}(\hat\bx, x) = \{ \psi_j(\hat\bx, x)\}_{j=1}^N,\quad 
\SV(\hat\bx, x) = \{\sv_j(\hat\bx, x)\}_{j=1}^N. 
\end{align} 
The regularity bounds of Sect.~\ref{sect:regpsi} can be rephrased in terms 
of the vector-function $\SPsi$. For instance, the bound 
\eqref{eq:FS} takes the form 
\begin{align}\label{eq:FSj}
|\p_{x}^m\SPsi(\hat\bx, x)|\lesssim 
\big(1+\l(\hat\bx, x)^{1-|m|}\big) 
M(\bx; {\sf \Psi}), 
\, \quad \textup{for all}\quad m\in \mathbb N_0^3, 
\end{align}
and all $(\hat\bx, x)\notin\Sigma:=\Sigma_N$,  where 
$\l(\bx) = \l_N(\bx)$ (see Sect.~\ref{subsect:hod} for 
definitions of $\Sigma_j$ and $\l_j$),
and 
\begin{align}\label{eq:newm}
M(\bx; \SPsi) = \bigg[\sum_{j=1}^N \, M(\bx; \psi_j)^2\bigg]^{\frac{1}{2}} 
= \bigg[\int_{|\bx-\by|<1}\, |\SPsi(\by)|^2 \,d\by\bigg]^{\frac{1}{2}}.
\end{align} 
Furthermore, using the functions $\psi_j$ one can rewrite 
the function $H(x)$ defined in \eqref{eq:H} as follows.

\begin{lem}\label{lem:Hnew}
\begin{align}\label{eq:Hnew}
H(x) =  
\begin{lcases}
&\ \frac{1}{16}\, \big(|\psi_1(x, x)|^2 + |\psi_2(x, x)|^2\big),\quad \textup{if}\ N = 2;\\[0.3cm]
&\ \frac{1}{16}\sum\limits_{j=1}^N\,\sum\limits_{k=1}^{N-1}\int\limits_{\R^{3N-6}} \big| 
\psi_j(\tilde\bx_{k, N},x, x)|^2\, d\tilde\bx_{k, N},\quad 
\textup{if}\ N\ge 3.
\end{lcases}
\end{align}
\end{lem}

\begin{proof}
Let us show this, for example, for $N\ge 3$. Indeed, it follows from 
\eqref{eq:psij} and \eqref{eq:xtilde} that for any $x, y\in \R^3$,
\begin{align*}
\int\limits_{\R^{3N-6}} \big| 
\psi_j(\tilde\bx_{k, N},y, x)|^2\, d\tilde\bx_{k, N}
=  
\begin{lcases}
&\ \int\limits_{\R^{3N-6}} \big| 
\psi(\tilde\bx_{k, j}, y, x)|^2\, d\tilde\bx_{k, j},\quad k < j;\\
&\ \int\limits_{\R^{3N-6}} \big| 
\psi(\tilde\bx_{j, k+1}, x, y)|^2\, d\tilde\bx_{j, k+1},\quad k \ge  j,\\
\end{lcases}
\end{align*}
and hence the right-hand side of \eqref{eq:Hnew} coincides with
\begin{align*}
\frac{1}{16}\, \sum\limits_{j=1}^N\,\sum\limits_{k=1}^{j-1}&\ \int\limits_{\R^{3N-6}} \big| 
\psi(\tilde\bx_{k, j},x, x)|^2\, d\tilde\bx_{k, j} 
  + \frac{1}{16}\sum\limits_{j=1}^{N-1}\,\sum\limits_{k=j}^{N-1}\int\limits_{\R^{3N-6}} \big| 
\psi(\tilde\bx_{j, k+1},x, x)|^2\, d\tilde\bx_{j, k+1} \\
= &\ \frac{1}{8}\, \sum\limits_{j=1}^{N-1}\,\sum\limits_{k=j+1}^{N}\int\limits_{\R^{3N-6}} \big| 
\psi(\tilde\bx_{j, k},x, x)|^2\, d\tilde\bx_{j, k} = H(x),
\end{align*}
as claimed.
\end{proof}

Given the factorization \eqref{eq:factorize}, by relation \eqref{eq:double}, 
the main Theorem  \ref{thm:maingk}
translates into the following theorem.

\begin{thm}\label{thm:recast}
The following asymptotics hold:
\begin{enumerate}
\item\label{item:repsi} 
If $\4\, \rho\4_{\, 1, 3/8} < \infty$, then $\SfG_{3/4}(\iop(\SPsi)) = \sg_{3/4}(\iop(\SPsi)) = A$;\\
\item\label{item:ret} 
If $\4\, \rho\4_{\, 1, 1/2} < \infty$, then $\SfG_1(\iop(\SV)) = \sg_1(\iop(\SV)) = B$; 
\end{enumerate}
\end{thm}

\subsubsection{Antisymmetric case: Theorem \ref{thm:maingkasym}}

If the function $\psi$ is totally antisymmetric, i.e.~\eqref{eq:antisym} holds, 
then by \eqref{eq:fb}  
the  factorization formulas \eqref{eq:factorize} hold with the kernels 
\begin{align}\label{eq:bpsias}
{\sf\Psi}(\hat\bx, x) = \sqrt{N}\, \psi(\hat\bx, x),\quad 
\SV(\hat\bx, x) = \sqrt{N}\, \nabla_x\psi(\hat\bx, x), 
\end{align}
and the bound \eqref{eq:varphider} with $R=1$ rewrites as 
\begin{align}\label{eq:FSas}
|\p_x^m\, \big( e^{-F_1(x)}    \SPsi(\hat\bx, x)\big)|  \lesssim \big(1 + \l(\bx)^{2-m}\big)\, M(\bx, \SPsi),\quad 
m\in\mathbb N_0^3,
\end{align}
where $\l(\bx)$ is as in \eqref{eq:FSj}.

\begin{thm}\label{thm:recastasym}
Let $\psi$ be totally antisymmetric. Then 
\begin{enumerate}
\item\label{item:repsiasy} 
If 
$\4\, \rho\4_{\, 1, 3/10} < \infty$, then 
$\SfG_{3/4}(\iop(\SPsi)) = \sg_{3/4}(\iop(\SPsi)) = A_{\textup{\tiny asym}}$;\\
\item\label{item:retasy} 
If 
$\4\, \rho\4_{\, 1, 3/8} < \infty$, then 
$\SfG_1(\iop(\SV)) = \sg_1(\iop(\SV)) = B_{\textup{\tiny asym}}$.
\end{enumerate}
\end{thm}

For both representations \eqref{eq:bpsi} and \eqref{eq:bpsias}, 
the definition \eqref{eq:density} of the one-particle density takes the form
\begin{align*}
\rho(x) = \int_{\R^{3N-3}}\, |\SPsi(\hat\bx, x)|^2\, d\hat\bx.
\end{align*}
Along with $\rho(x)$ we also define the ``weighted" density
\begin{align*}
\rho[f](x) =  
\int_{\R^{3N-3}} \, |f(\hat\bx)|^2 \, |\SPsi(\hat\bx, x)|^2   \, d\hat\bx,
\end{align*}
with an arbitrary function $f\in\plainL\infty(\R^{3N-3})$.
Now we are ready to take the first step towards the proof of Theorems \ref{thm:recast} 
and \ref{thm:recastasym}.

\subsection{Preliminary spectral bounds} \label{subsect:psb}
In this section we obtain some spectral bounds for the operator 
$\iop(T)$ with a vector-valued 
kernel $T = T(\hat\bx, x)$ that satisfies the bound
\begin{align}\label{eq:imit}
|\p_x^s \, T(\hat\bx, x)|\lesssim \big(1+ \l(\hat\bx, x)^{\alpha - |s|}\big) 
\, M(\bx; \SPsi),\quad s\in \mathbb N_0^3,
\end{align}
with some $\a>-3/2$. Later on these bounds will be used 
for $T(\hat\bx, x) = \SPsi(\hat\bx, x)$ (for Theorem \ref{thm:recast}) 
or $T(\hat\bx, x) = e^{-F_1(x)} \SPsi(\hat\bx, x)$ (for Theorem \ref{thm:recastasym}), see 
\eqref{eq:f1} for the definition of $F_1(x)$.  
We denote 
\begin{align}\label{eq:peee}
\frac{1}{p} = 1 + \frac{\alpha}{3}.
\end{align} 
As usual, the condition $\a>-3/2$ ensures that $p <2$. 

\begin{lem}\label{lem:boundexp}
 Let $T = T (\hat\bx, x)$, $\hat\bx\in\R^{3N-3}, x\in \R^3$, be  
as above, and let 
\begin{align}\label{eq:rholat}
\4\, \rho\4_{\, 1, p/2} < \infty.
\end{align}   
Let $b = b(\hat\bx)$ and $a = a(x)$ be weights that satisfy the conditions 
\begin{align}\label{eq:ba}
\sup_{\hat\bx} M(\hat\bx; b)<\infty, \ 
\sup_{n\in \Z^3}\, \|a\|_{\plainL{r}(\CC^{(1)}_n)}<\infty,
\end{align}
where the parameter $r$ is as in \eqref{eq:rmodel}. 
Then $b \, \iop(T)\, a\in \BS_{p, \infty}$, and  
 \begin{align}\label{eq:boundexp}
\|b\, \iop(T)\, a\|_{p, \infty}^p\lesssim  \sum_{n\in\Z^3}
 \bigg[\int\limits_{\CC^{(4)}_n}  \, \rho[M(b)](x)  \, dx\bigg]^{\frac{p}{2}} \ 
 \|a\|_{\plainL{r}(\CC^{(1)}_n)}^p\lesssim \4\, \rho\4_{\, 1, p/2}^{\frac{p}{2}}.
 \end{align}
In particular,
\begin{align}\label{eq:byrho}
\| \iop(T) \|_{p, \infty}\lesssim \4\, \rho\4_{\, 1, p/2}^{\frac{1}{2}}.
\end{align} 

\end{lem}

\begin{proof}
We use Theorem \ref{thm:BSspace} with $d=3$ and $m = 3N-3$. 
By \eqref{eq:imit}, the kernel $b(\hat\bx) T(\hat\bx, x)$ 
satisfies the bound \eqref{eq:kernelexp} 
with $z_j(x) = x_j, j = 1, 2, \dots, N-1$, 
and 
\begin{align*}
A(\hat\bx, x) = |b(\hat\bx)| M(\bx; \SPsi),
\end{align*}
$\bx = (\hat\bx, x)$.  
Therefore 
\begin{align*}
A_n(\hat\bx) = &\ \|A(\hat\bx,\,\cdot\,)\|_{\plainL\infty(\CC_n^{(2)})}\le 
|b(\hat\bx)| 
\|\SPsi\|_{\plainL2(B_n(\hat\bx))},\, n\in\Z^3,\\  
\textup{where}\quad B_n(\hat\bx) 
= &\ 
\{\hat\by\in\R^{3N-3}: |\hat\by-\hat\bx|<1\}
\times \{x\in\R^3: x\in \CC^{(4)}_n\}
\end{align*}
As a consequence, 
\begin{align*}
\|A_n\|_{\plainL2(\R^{3N-3})}^2
\le &\ \int\limits_{\R^{3N-3}} \,  |b(\hat\bx)|^2   
\int\limits_{|\hat\by-\hat\bx|< 1}   \bigg[\int\limits_{\CC^{(4)}_n} |\SPsi(\hat\by, x)|^2\,   
dx\bigg]\, d\hat\by  d\hat\bx\\
= &\ \int\limits_{\CC^{(4)}_n}\int\limits_{\R^{3N-3}} \,     
\bigg[\int\limits_{|\hat\by-\hat\bx|<1} |b(\hat\bx)|^2 d\hat\bx\bigg] \, |\SPsi(\hat\by, x)|^2\,
d\hat\by\, dx\\
= &\ \int\limits_{\CC^{(4)}_n}\int\limits_{\R^{3N-3}} \,     
M(\hat\by; b) |\SPsi(\hat\by, x)|^2\,
d\hat\by\, dx\\
= &\  \int\limits_{\CC^{(4)}_n}  \, \rho[M(b)](x)  \, dx.
\end{align*}
Now the first bound in \eqref{eq:boundexp} follows from \eqref{eq:piv}. The second bound holds because of \eqref{eq:ba}. In particular, it holds when $a =1$ and $b = 1$. This leads to \eqref{eq:byrho}.
\end{proof}

\subsection{Reduction to compact supports}
Let $a_R(x) = a(x) \t(|x| R^{-1})$, $b_R(\hat\bx) = b(\bx) \t(|\bx| R^{-1})$.
 
\begin{lem}\label{lem:tocompact} 
Suppose that \eqref{eq:rholat} is satisfied and that the weights $a$ and $b$ 
are as in Lemma \ref{lem:boundexp}. Then 
\begin{align*}
\lim_{R\to\infty}\, \SfG_p (b_R\, \iop(T)\, a_R) = &\ \SfG_p (b\, \iop(T)\, a),\\
\lim_{R\to\infty}\, \sg_p (b_R\, \iop(T)\, a_R) = &\ \sg_p (b\, \iop(T)\, a).
\end{align*} 
 \end{lem} 
 
\begin{proof}
Estimate the $\BS_{p, \infty}$-quasi-norm of 
$(b - b_R)  \, \iop(T)\, a$ using Lemma \ref{lem:boundexp}. Since $M(\hat\bx; b-b_R) = 0$ for 
$|\hat\bx|<R/4$, we have 
\begin{align*}
\|(b-b_R)\, \iop(T)\, a\|_{p, \infty}^p\lesssim  &\ \sum_{n\in\Z^3}
 \bigg[\int_{\CC^{(4)}_n}  \, \bigg(\int_{|\hat\bx|>R/4}\, |\psi(\hat\bx, x)|^2 \, d\hat\bx\bigg)  \, dx\bigg]^{\frac{p}{2}} \ 
 \|a\|_{\plainL{r}(\CC^{1}_n)}^p\\
\lesssim &\ \sum_{n\in\Z^3}
 \bigg[\int_{\CC^{(4)}_1}  \, \bigg(\int_{|\hat\bx|>R/4}\, |\psi(\hat\bx, x)|^2 \, d\hat\bx\bigg)  \, dx\bigg]^{\frac{p}{2}}. 
\end{align*}
Each integral tends to zero as $R\to\infty$. Due to the condition \eqref{eq:rholat}, by the Dominated Convergence Theorem, the sum on the right-hand side tends to zero as well. 

By Lemma \ref{lem:boundexp} again, 
\begin{align*}
\|(b_R\, \iop(T)\, (a-a_R)\|_{p, \infty}^p\lesssim  &\ \sum_{|n|> R/4}
 \bigg(\int_{\CC^{(4)}_n}  \, \rho(x)\, dx\bigg)^{\frac{p}{2}} \ 
 \|a\|_{\plainL{r}(\CC^{1}_n)}^p\\
\lesssim &\ \sum_{|n|>R/4} \|\rho\|_{\plainL1(\CC^{(1)}_n)}^{\frac{p}{2}}.
\end{align*}
By the condition \eqref{eq:rholat}, the right-hand side tends to zero as $R\to\infty$. 

Due to \eqref{eq:triangle}, 
\begin{align*}
\|b\, \iop(T)\, a & 
- b_R\, \iop(T)\, a_R\|_{p, \infty}^{{\frac{\scalel{1}}{\scalet{p+1}}}}\\
\le &\ 
\|(b - b_R)\, \iop(T)\, a\|_{p, \infty}^{{\frac{\scalel{1}}{\scalet{p+1}}}}
+ \|b_R\, \iop(T)\, (a -a_R)\|_{p, \infty}^{{\frac{\scalel{1}}{\scalet{p+1}}}}
\to 0,\quad\textup{as}\quad R\to \infty.
\end{align*}
The required equalities follow from 
\eqref{eq:gnorm} and Corollary \ref{cor:zero1}.
\end{proof}

\subsection{Separation of the particles}
First we construct appropriate cut-offs. 
Recall the notation $x_0 = 0$.  
Fix a $\d > 0$.  Define the set 
\begin{align}\label{eq:wd}
\SU_k(\d) = \{ (\hat\bx, x): |x_s-x_l|> 4\d, &\ \, 0\le s < l\le N-1,\ 
\    \textup{and}\, \, 
|x-x_k|<\d \},\notag\\ 
&\ \qquad\qquad k = 0, 1, \dots, N-1.
\end{align}
In words, for $\bx\in \SU_k(\d)$ 
the coordinate pair $x, x_k$ are close to each other, whereas all the other pairs are 
kept separate. 
It is clear that for distinct $j$ and $k$ 
the sets $\SU_k(\d)$ and $\SU_j(\d)$ 
are disjoint.  

Let us now construct cut-off functions supported on $\SU_k(\d)$. 
Let $\t\in\plainC\infty_0(\R)$ and $\z = 1-\t$ be as defined in \eqref{eq:sco}, \eqref{eq:sco1}. 
Denote 
\begin{align}\label{eq:ydel}
\eta_\d(\hat\bx) = \prod_{0\le l < s \le N-1} \z\big(|x_l-x_s|(8\d)^{-1}\big).
\end{align}
By the definition of $\z$, 
\begin{align*}
\supp \eta_\d\subset \{\hat\bx\in \R^{3N-3}: |x_s-x_l|>4\d,\, 0\le s < l\le N-1\}.
\end{align*}
As $\t \ge 0$ and $\z\le 1$, we have 
\begin{align}\label{eq:comply}
1-\eta_\d(\hat\bx)\le \sum_{0\le l<s\le N-1} \t\big(|x_l-x_s|(8\d)^{-1}\big),
\end{align}
Now it is clear that 
\begin{align}\label{eq:cutoff}
\textup{the support of the function} 
\quad \eta_\d(\hat\bx) \t\big(|x-x_k|\d^{-1}\big)
\ \  
\textup{belongs to}\quad \SU_k(\d),  
\end{align} 
for all $k = 0, 1, \dots, N-1$. 
Furthermore, due to \eqref{eq:sco1}, we have 
\begin{align}\label{eq:partun}
\eta_\d(\hat\bx)\bigg[1 - \sum_{k=0}^{N-1} \t(|x-x_k|\d^{-1})\bigg] 
= \eta_\d(\hat\bx)\prod_{k=0}^{N-1} \z\big(|x-x_k|\d^{-1}\big).
\end{align} 
Along with the kernel $T$ introduced earlier, 
define the truncated kernel 
\begin{align}\label{eq:sdelta}
T_\d(\hat\bx, x) = \eta_\d(\hat\bx) T(\hat\bx, x)\,\sum_{k=1}^{N-1}  \t(|x-x_k|\d^{-1}).
\end{align} 
 Recall that $p$ is as defined in \eqref{eq:peee}. 
 
\begin{lem} \label{lem:separate}
Let $b\in \plainL2(\R^{3N-3})$, and let $a$ be such that 
$\4\, a\4_{\,r,p}<\infty$. Then 
\begin{align*}
\lim_{\d\to 0}\, \SfG_p (b\, \iop(T_\d)\, a) = &\ \SfG_p (b\, \iop(T)\, a),\\
\lim_{\d\to 0}\, \sg_p (b \, \iop(T_\d)\, a) = &\ \sg_p (b\, \iop(T)\, a).
\end{align*} 
 \end{lem}

\begin{proof} 
By \eqref{eq:trianglep},
\begin{align}\label{eq:di}
\SfG_p\big(b\, \iop(T_\d)\, a - b\, \iop(T)\, & a\big)
^{{\frac{\scalel{1}}{\scalet{p+1}}}}\notag\\
\le &\ \SfG_p\big((I-\eta_\d)b\, \iop(T)\, a\big)^{{\frac{\scalel{1}}{\scalet{p+1}}}}
+ \SfG_p\big(b\, \iop(\eta_\d T - T_\d)\, a\big)^{{\frac{\scalel{1}}{\scalet{p+1}}}}.
\end{align}
First we estimate the first term on the right-hand side using Corollary \ref{cor:homo}. 
Since $M(\bx; \SPsi)\lesssim 1$, and $1-\eta_\d$ satisfies \eqref{eq:comply}, we have
\begin{align*}
\|(I-\eta_\d)b\, \iop(T)\, a\|_{p, \infty}
\lesssim   &\  \sum_{0\le l < s\le N-1}
\bigg[\int_{\R^{3N-3}}  \t\big(|x_l-x_s|(8\d)^{-1}\big)^2 |b(\hat\bx)|^2 \, d\hat\bx\bigg]^{\frac{1}{2}}\, \4 \, a\4_{\,r, p}\\
 \lesssim &\  \sum_{0\le l < s\le N-1}
\bigg[\int_{|x_l-x_s|< 8\d}  |b(\hat\bx)|^2 \, d\hat\bx\bigg]^{\frac{1}{2}}\, 
\4 \, a\4_{\,r, p}.
\end{align*}
As $b\in \plainL2(\R^{3N-3})$,
the right-hand side tends to zero as $\d\to 0$. By 
virtue of \eqref{eq:gnorm} so does the first term on the right of \eqref{eq:di}.

Let us now estimate the second term on the right-hand side of \eqref{eq:di}. 
Due to \eqref{eq:partun}, we can write that 
\begin{align*}
 \eta_\d(\hat\bx, x) T(\hat\bx, x) - &\ T_\d(\bx, x)  \\
= &\ \eta_\d(\hat\bx)\bigg[ 1 - \sum_{k=0}^{N-1}\, \t(|x-x_k|\d^{-1})
+ \t(|x|\d^{-1})
\bigg] \,
T(\hat\bx, x)\\[0.2cm]
= &\ \eta_\d(\hat\bx) \prod_{k=0}^{N-1} \z(|x-x_k|\d^{-1})\, T(\hat\bx, x) 
+ \eta_\d(\hat\bx) \t(|x|\d^{-1})\, T(\hat\bx, x)\\
=: &\ T_\d^{(1)}(\hat\bx, x) + T_\d^{(2)}(\hat\bx, x).
\end{align*}
By the definition \eqref{eq:ydel},  
the support of $T_\d^{(1)}(\hat\bx, x)$  
is separated from all coalescence points $x = x_k$, 
$k = 0, 1, \dots, N-1$. 
Taking into account \eqref{eq:imit}, we conclude that  
\begin{align*}
|\p_x^j\, T_\d^{(1)}(\hat\bx, x)|
\lesssim M(\bx; \SPsi)\lesssim 1, \quad j\in \mathbb N_0^3, 
\end{align*}
with constants depending on $\d$. 
By Lemma \ref{lem:tocompact}, we may assume that $a$ is compactly supported, so that 
$\4\, a\4_{\,r, q}<\infty$ for all $q >0$. 
Thus Corollary \ref{cor:homo} ensures 
that $\iop(b\, T_\d^{(1)}\, a)\in \BS_{q, \infty}$ 
for all $q >0$ and hence $\SfG_{p}(b\, T_\d^{(1)}\, a) = 0$. 

By Corollary \ref{cor:homo} again,
\begin{align*}
\|b\, T_\d^{2)}\, a\|_{p, \infty}\lesssim \|b\|_{\plainL2(\R^{3N-3})} \4 \, \t(\,\cdot\,/\d) a\4_{\,r, p} 
= \|b\|_{\plainL2(\R^{3N-3})}\,  \| \t(\,\cdot\,/\d) a\|_{\plainL{r}(\CC_0^{(1)})}\to 0,\quad \textup{as} 
\quad \d\to 0. 
\end{align*}
Thanks to \eqref{eq:gnorm}, the same is true for $\SfG_p(b \,T_\d^{(2)} a)$. 
Using \eqref{eq:trianglep} we get that 
\begin{align*}
\SfG_p\big(\eta_\d b\, \iop(T - T_\d)\, a\big)\to 0, \quad \textup{as}\quad \d\to 0.
\end{align*}
A few lines above we proved that the first term in \eqref{eq:di} also tends to zero as $\d\to 0$. 
Now the required equalities follow from Corollary \ref{cor:zero1}.
\end{proof}

The rest of the paper is focused on the proof of Theorems \ref{thm:recast}
and \ref{thm:recastasym}.  

\section{Proof of Theorem \ref{thm:recast}}\label{sect:prfrecast}
\label{sect:proofs}

\subsection{Regularity of $\psi$ on the sets $\SU_k(\d)$}
Her we lay the grounds for the proof of Theorem \ref{thm:recast} by studying the structure of 
the function $\psi$ on the sets $\SU_k(\d), k = 1, 2, \dots, N-1$, defined in \eqref{eq:wd}. 
We begin with rewriting the factorization $\psi = e^F\phi$ 
introduced in Sect. \ref{sect:reg}, in terms of the functions $\psi_j(\hat\bx, x)$:
\begin{align*}
\psi_j(\hat\bx, x) = &\ e^{F(\bx)}\, \phi_j(\hat\bx, x),\quad \textup{where}\\
\phi_j(\hat\bx, x) = &\ 
\phi(x_1, \dots, x_{j-1}, x, x_j, \dots, x_{N-1}),\quad j = 1, 2, \dots, N,
\end{align*}
where $F(\hat\bx, x)$ is the function \eqref{eq:jascut} 
which we represent in a slightly different form, using the notation  
\begin{align*}
G(t) = \frac{1}{4}|t|\, \t(|t|), \ t\in\R^3:
\end{align*}
\begin{align}\label{eq:FG}
F(\bx) = &\ \sum_{k=1}^{N-1} G(x-x_k) - 2 Z \, G(x) 
 - 2Z\sum_{s=1}^{N-1} G(x_s) + \sum_{1\le l < s\le N-1} \, G(x_s-x_l).
\end{align}
By Theorem \ref{thm:regphi}, 
\begin{align}\label{eq:phij}
|\p_x^s\,  \phi_j(\bx)|\lesssim \big(1+\l(\hat\bx, x)^{2-|s|}\big) 
\, M(\bx; \SPsi), \quad s\in\mathbb N_0^3,
\end{align}
where $M(\bx; \SPsi)$ is defined in \eqref{eq:newm}.
For $\bx\in \SU_k(\d)$ it will be convenient to use 
a different factorization of $\psi_j$. Namely, for each value of $k$ we split the function 
\eqref{eq:FG} into two summands,   
\begin{align*}
F(\bx) = G(x-x_k) + G_k(\hat\bx, x), \quad k = 1, 2, \dots, N-1,
\end{align*}
where 
\begin{align*}
  G_k(\bx) = 
\sum_{\substack{1\le s\le N-1\\
s\not = k}} G(x-x_s)
- 2 Z \, G(x) - 2Z\sum_{s=1}^{N-1} G(x_s) + \sum_{1\le l < s\le N-1} \, G(x_s-x_l),
\end{align*}
and rewrite the product $\psi_j = e^F \phi_j$ as follows: 
\begin{align}\label{eq:phijk}
\psi_j(\hat\bx, x) =  e^{G(x-x_k)}\, \phi_{j,k}(\hat\bx, x), \ k = 1, \dots, N-1.
\end{align}
where
\begin{align*}   
\phi_{j,k}(\hat\bx, x) = e^{G_k(\hat\bx, x)}\phi_j(\hat\bx, x),\ k = 1, \dots, N-1.
\end{align*}
Observe that the function $G_k(\bx)$ is $\plainC\infty$ on $\SU_k(\d)$, and hence $\phi_{j, k}$ 
possesses the same smoothness as $\phi_j$. Furthermore, on $\SU_k(\d)$ the variable $x$ is separated from all $x_s$, $s\not = k$, and precisely, $|x-x_s|>\d$, $0\le s\le N-1, s\not = k$. Therefore \eqref{eq:phij} implies that  
\begin{align}\label{eq:onw}
|\p_x^s\, \phi_{j,k}(\hat\bx, x)|
\lesssim \big(1+ (|x-x_k|\wedge 1)^{2-|s|}\big) \, &M(\bx; \SPsi),\notag\\
(\hat\bx,\, &\ x)\in \SU_k(\d),\ k = 1, 2, \dots, N-1,
\end{align}
for all $s\in \mathbb N_0^3$, with some coefficients depending on $\d$. 

\begin{lem} For all $\bx\in \SU_k(\d)$, $k = 1, 2, \dots, N-1$, we have 
\begin{align}
\psi_j(\hat\bx, x) = &\ \frac{1}{4}
|x-x_k| \phi_{j, k}(\hat\bx, x) + R_{j, k}(\hat\bx, x),\label{eq:singlet}\\
\nabla_x \psi_j(\hat\bx, x) = &\ \frac{1}{4}
\frac{x-x_k}{|x-x_k|} \phi_{j, k}(\hat\bx, x) + R^{(1)}_{j, k}(\hat\bx, x),
\label{eq:singlet1}
\end{align}
where
\begin{align}
|\p_x^m\, R_{j, k}(\hat\bx, x)|\lesssim &\ \big(1+( |x-x_k|\wedge 1)^{2-|m|}\big) \, M(\bx, \SPsi),
\label{eq:rsmooth}\\
|\p_x^m\, R^{(1)}_{j, k}(\hat\bx, x)|\lesssim &\ 
\big(1+( |x-x_k|\wedge 1)^{1-|m|}\big) \, M(\bx, \SPsi),
\label{eq:r1smooth}
\end{align}
for all $m\in \mathbb N_0^3$.
\end{lem}
 
 \begin{proof}
 Rewrite \eqref{eq:phijk} in the form
 \begin{align*}
 \psi_j(\bx) = G(x-x_k) \phi_{j, k}(\bx) + R_{j, k}(\bx),
 \end{align*}
 with
 \begin{align*}
 R_{j, k}(\bx) = \big(e^{G(x-x_k)} - G(x-x_k)\big) \phi_{j, k}(\bx). 
 \end{align*}
 By \eqref{eq:expbnd1}, 
\begin{align*}
|\p_x^m \big(e^{G(x-x_k)} - G(x-x_k)\big) |
\lesssim 1+|x-x_k|^{2-|m|},\quad m\in\mathbb N_0^3.
 \end{align*}
 Together with \eqref{eq:onw} this leads to \eqref{eq:singlet}.
 
In order to prove \eqref{eq:singlet1} we differentiate \eqref{eq:singlet}:
\begin{align*}
\nabla_x\,\psi_j(\hat\bx, x) = &\ \frac{1}{4}
\frac{x-x_k}{|x-x_k|} \phi_{j, k}(\hat\bx, x) + R^{(1)}_{j, k}(\hat\bx, x),\\
R^{(1)}_{j, k}(\hat\bx, x) = &\ 
\frac{1}{4}
|x-x_k| \nabla_x\phi_{j, k}(\hat\bx, x) + \nabla_x R_{j, k}(\hat\bx, x),
\end{align*} 
The bounds \eqref{eq:r1smooth} hold thanks to \eqref{eq:onw} and \eqref{eq:rsmooth}. 
 \end{proof}

\subsection{Proof of Theorem \ref{thm:recast} \eqref{item:repsi}} 
We use the results of Sect. \ref{sect:factor} 
with $T(\hat\bx, x) = \SPsi(\hat\bx, x)$. 
Our argument is built on Lemmata \ref{lem:tocompact} and \ref{lem:separate}. 
They reduce the analysis to the study of 
the  truncated kernel $T_\d(\hat\bx, x)$ 
as defined in \eqref{eq:sdelta}:
\begin{align*}
T_\d(\hat\bx, x) = \eta_\d(\hat\bx) \SPsi(\hat\bx, x)\,\sum_{k=1}^{N-1}  \t(|x-x_k|\d^{-1}).
\end{align*} 
To prove Theorem \ref{thm:recast} \eqref{item:repsi} we find the spectral asymptotics of the 
operator $ b \iop(T_\d) a$ with compactly supported $a$, $b$ and with the cut-off 
$\eta_\d$ defined in \eqref{eq:ydel}. Define 
\begin{align}\label{eq:hdelta2}
H_\d(x; a, b) =  \frac{|a(x)|^2}{16}\, |b(x)|^2  |\eta_\d(x)|^2 
\big(|\psi_1(x, x)|^2 + |\psi_2 (x, x)|^2\big),\quad \textup{if}\quad N = 2,
\end{align}
and 
\begin{align}\label{eq:hdelta}
H_\d(x; a, b) =  \frac{|a(x)|^2}{16} 
\sum_{j=1}^N &\ \,\sum_{k=1}^{N-1}\int_{\R^{3N-6}}\, |b(\tilde\bx_{k, N}, x)|^2\notag\\ 
&\ \times |\eta_\d(\tilde\bx_{k, N}, x)|^2 
|\psi_j (\tilde\bx_{k, N}, x, x)|^2 \,d\tilde\bx_{k, N},\quad \textup{if}\quad N \ge 3. 
\end{align}

\begin{lem}\label{lem:predr}
Let $a\in\plainC{}_0(\R^3)$ and $b\in\plainC{}_0(\R^{3N-3})$. 
Then for each $\d\in (0, 1/2)$ we have  
\begin{align}\label{eq:predr}
\SfG_{3/4} (b\, \iop(T_\d)\, a) 
= \sg_{3/4} (b \,\iop(T_\d)\, a) = \frac{1}{3}\bigg(\frac{2}{\pi}\bigg)^{\frac{5}{4}} \, \int_{\R^3}\, 
H_{\d}(x; a, b)^{\frac{3}{8}}\,dx,
\end{align}
where $H_\d$ is defined in \eqref{eq:hdelta2} or \eqref{eq:hdelta}.
\end{lem}

\begin{proof} 
We conduct the proof only for $N\ge 3$. The case $N=2$ requires only minor adjustments. 

By \eqref{eq:cutoff}, the support of $\eta_\d(\hat\bx) \t(|x-x_k|\d^{-1})$ belongs to the set 
$\SU_k(\d), k = 1, 2, \dots, N-1$. Thus we can 
write out the components of the kernel $T_\d$ using \eqref{eq:singlet}:
\begin{align*}
(T_\d(\hat\bx, x))_j = &\ \eta_\d(\hat\bx)
\sum_{k=1}^{N-1} e^{G(x-x_k)} \t(|x-x_k|\d^{-1})\phi_{j, k}(\hat\bx, x)\\
= &\ \sum_{k=1}^{N-1} S_{j, k}(\hat\bx, x) + \sum_{k=1}^{N-1} \Om_{j, k}(\hat\bx, x), 
\quad j = 1, 2, \dots, N, 
\end{align*}
where
\begin{align*}
S_{j, k}(\hat\bx, x) = &\ \frac{1}{4}\eta_\d(\hat\bx)\, |x-x_k| \,
\t(|x-x_k|\d^{-1})\phi_{j, k}(\hat\bx, x),\\[0.3cm]
\Om_{j, k}(\hat\bx, x) = &\ 
\eta_\d(\hat\bx) \t(|x-x_k|\d^{-1}) R_{j, k}(\hat\bx, x).
\end{align*} 
By \eqref{eq:rsmooth}, each kernel $\Om_{j, k}$ satisfies 
\eqref{eq:s} with $A(\bx) = M(\bx, \SPsi)\le 1$ and $\a = 2$. 
Since both weights $a$ and $b$ are continuous and compactly supported, 
by Corollary \ref{cor:homo}, the operator 
$\iop(b\, \Om_{j, k}\,  a)$  
belongs to $\BS_{3/5, \infty}$, and hence  $\SfG_{3/4}(b \iop(\Om_{j, k}) a) = 0$. 
Thanks to Corollary \ref{cor:zero}, 
\begin{align*}
\SfG_{3/4}(b\, \iop(T_\d)\, a ) = &\ \SfG_{3/4}(b\, \iop(S)\, a),\quad  
\sg_{3/4}(b\, \iop(T_\d)\, a ) = \sg_{3/4}(b\, \iop(S)\, a),\\
 \textup{where} \quad S(\hat\bx, x) = &\ \{ S_j(\hat\bx, x)\}_{j=1}^N,\, 
 S_j(\hat\bx, x) = \sum_{k=1}^{N-1} S_{j, k}(\hat\bx, x).
\end{align*}
From Example  \ref{ex:vect} with $\a=1$, $\SQ(x) = |x|$ 
and formula \eqref{eq:scalar} 
we obtain that 
\begin{align}\label{eq:tech}
\SfG_{3/4}(b\, \iop(S)\, a) = \sg_{3/4}(b\, \iop(S)\, a)
= \frac{|\nu_{1,3}|^{\frac{3}{4}}}{6\pi^2}
\,\int_{\R^3}  |a(x)|^{\frac{3}{4}} |\SB(x)|^{\frac{3}{8}} \, dx,
\end{align}
where
\begin{align}\label{eq:sbd}
\SB(x) 
 = \frac{1}{16}\sum_{j=1}^N \,\sum_{k=1}^{N-1}\int_{\R^{3N-6}}\, |b(\tilde\bx_{k, N}, x)|^2 
 |\eta_\d(\tilde\bx_{k, N}, x)|^2 
|\phi_{j,k}(\tilde\bx_{k, N}, x, x)|^2 \,d\tilde\bx_{k, N},
\end{align}
By \eqref{eq:phijk}, $\phi_{j, k}(\tilde\bx_{k, N}, x,x) = \psi_j(\tilde\bx_{k, N}, x,x)$, 
and an easy calculation shows that the coefficient 
on the right-hand side of \eqref{eq:tech} equals $\frac{1}{3}(\frac{2}{\pi})^{\frac{5}{4}}$.  
Now \eqref{eq:predr} follows. 
\end{proof}

\begin{proof}[Proof of Theorem \ref{thm:recast}\eqref{item:repsi}]
Let $a(x) = a_R(x) = \t(|x| R^{-1})$, $b(\hat\bx) = b_R(\hat\bx) = \t(|\hat\bx| R^{-1})$. 
According to Lemmata \ref{lem:tocompact} and \ref{lem:separate}, 
to prove Theorem \ref{thm:recast}\eqref{item:repsi} it suffices to show that 
\begin{align*}
\lim_{\d\to0, R\to\infty}\SfG_{3/4}(b_R\, \iop (T_\d) \, a_R) = \lim_{\d\to0, R\to\infty}
\sg_{3/4}(b_R\, \iop (T_\d) \, a_R) 
= \frac{1}{3} \bigg(\frac{2}{\pi}\bigg)^{\frac{5}{4}}\, \int_{\R^3}\, H(x)^{\frac{3}{8}}   \, dx.
\end{align*}
Due to \eqref{eq:predr} it suffices to check that  
\begin{align}\label{eq:h38}
\lim_{\d\to0, R\to\infty} 
\int_{\R^3}\, H_\d(x; a_R, b_R)^{\frac{3}{8}}   \, dx
= \int_{\R^3}\, H(x)^{\frac{3}{8}}   \, dx.
\end{align}
According to Lemma \ref{lem:coeffs}, 
the function $H(x)^{3/8}$ is integrable. Therefore, the function 
$\psi(\tilde\bx, x, x)$ is 
square-integrable in $\tilde\bx$ for a.e. $x\in\R^3$.
Furthermore, 
$|\eta_\d|\le 1$, $|a_R|, |b_R|\le 1$ and $\eta_\d\to 0, \d\to 0,$ a.e., and $a_R\to 1, b_R\to 1$, $R\to\infty$,  
a.e. Thus, by the Dominated Convergence Theorem, 
$H_\d(x; a_R, b_R)\to H(x)$ a.e. $x\in\R^3$. Since 
 $H_\d(x; a_R, b_R)\le H(x)$, the Dominated Convergence Theorem ensures that 
\eqref{eq:h38} holds. 
Theorem 
\ref{thm:recast}\eqref{item:repsi} is proved. 
\end{proof}

\subsection{Proof of Theorem \ref{thm:recast} \eqref{item:ret}} 
We use the results of Sect. \ref{sect:factor} with the vector-valued kernel 
$T(\hat\bx, x) = \SV(\hat\bx, x) = \nabla_x\SPsi(\hat\bx, x)$, see \eqref{eq:bpsi}. 
Let the vector-valued 
kernel $T_\d(\hat\bx, x)$ 
be as defined in \eqref{eq:sdelta}:
\begin{align*}
T_\d(\hat\bx, x) = \eta_\d(\hat\bx) \nabla_x\SPsi(\hat\bx, x)\,\sum_{k=1}^{N-1}  \t(|x-x_k|\d^{-1}).
\end{align*} 
To prove Theorem \ref{thm:recast} \eqref{item:ret} we find the spectral asymptotics of the 
operator $ b \iop(T_\d) a$ with compactly supported $a$, $b$ and with the cut-off 
$\eta_\d$ defined in \eqref{eq:ydel}.

\begin{lem}
Let $a\in\plainC{}_0(\R^3)$ and $b\in\plainC{}_0(\R^{3N-3})$. 
Then for each $\d\in (0, 1/2)$ we have  
\begin{align}\label{eq:predre}
\SfG_{1/2} (b\, \iop(T_\d)\, a) 
= \sg_{1/2} (b \,\iop(T_\d)\, a) =  \frac{4}{3\pi} 
 \, \int_{\R^3}\, 
H_{\d}(x; a, b)^{\frac{1}{2}}\,dx,
\end{align}
where $H_\d$ is defined in \eqref{eq:hdelta}.
\end{lem}

\begin{proof} 
We conduct the proof only for $N\ge 3$. The case $N=2$ requires only minor adjustments. 

As in the proof of Theorem 
\ref{thm:recast} \eqref{item:repsi}, due to \eqref{eq:cutoff} we can write out the components of the kernel $T_\d$ using \eqref{eq:singlet1}:
\begin{align*}
(T_\d(\hat\bx, x))_j =  
  \sum_{k=1}^{N-1} S_{j, k}(\hat\bx, x) + \sum_{k=1}^{N-1} \Om_{j, k}(\hat\bx, x), 
\quad j = 1, 2, \dots, N, 
\end{align*}
where
\begin{align*}
S_{j, k}(\hat\bx, x) = &\ \frac{1}{4}\eta_\d(\hat\bx)\, 
\frac{x-x_k}{|x-x_k|} \,\t(|x-x_k|)\d^{-1}) 
 \phi_{j, k}(\hat\bx, x),\\[0.3cm]
\Om_{j, k}(\hat\bx, x) = &\ 
\eta_\d(\hat\bx) \t(|x-x_k|\d^{-1})\, R^{(1)}_{j, k}(\hat\bx, x).
\end{align*} 

By \eqref{eq:rsmooth}, each kernel $\Om_{j, k}$ satisfies 
\eqref{eq:s} with $A(\bx) = M(\bx, \SPsi)\le 1$ and $\a = 1$. 
Since both weights $a$ and $b$ are continuous and compactly supported, 
by Corollary \ref{cor:homo}, the operator 
$\iop(b\, \Om_{j, k}\,  a)$  
belongs to $\BS_{3/4, \infty}$, and hence  $\SfG_1(b \iop(\Om_{j, k}) a) = 0$. 
Thanks to Corollary \ref{cor:zero}, 
\begin{align*}
\SfG_1(b\, \iop(T_\d)\, a ) = &\ \SfG_1(b\, \iop(S)\, a),\quad  
\sg_1(b\, \iop(T_\d)\, a ) = \sg_1(b\, \iop(S)\, a),\\
 \textup{where} \quad S(\hat\bx,x) = &\ \{ S_j(\hat\bx, x)\}_{j=1}^N,\, 
 S_j(\hat\bx, x) = \sum_{k=1}^{N-1} S_{j, k}(\hat\bx, x).
\end{align*}
Using Example \ref{ex:vect} with $\a=0$ and 
the vector-function $\SQ(x) = \nabla |x| = x|x|^{-1}$, we obtain 
from formula \eqref{eq:vector} 
that 
 \begin{align*}
\SfG_1(b\, \iop(S)\, a) = \sg_1(b\, \iop(S)\, a)
= \frac{|\nu_{1, 3}|}{6\pi^2}\,\int_{\R^3}  |a(x)|\, |\SB(x)|^{\frac{1}{2}} \, dx,
\end{align*}
where $\SB(x)$ is as in \eqref{eq:sbd}.
By \eqref{eq:phijk}, $\phi_{j, k}(\tilde\bx_{k, N}, x,x) = \psi_j(\tilde\bx_{k, N}, x,x)$, 
and an easy calculation shows that $\frac{|\nu_{1, 3}|}{6\pi^2} = \frac{4}{3\pi}$.  Now 
\eqref{eq:predre} follows. 
\end{proof}

\begin{proof}[Proof of Theorem \ref{thm:recast}\eqref{item:ret}]
Let $a(x) = a_R(x) = \t(|x| R^{-1})$, $b(\hat\bx) = b_R(\hat\bx) = \t(|\hat\bx| R^{-1})$. 
According to Lemmata \ref{lem:tocompact} and \ref{lem:separate}, 
to prove Theorem \ref{thm:recast}\eqref{item:ret} it suffices to show that 
\begin{align*}
\lim_{\d\to0, R\to\infty}\SfG_{1}(b_R\, \iop (T_\d) \, a_R) = \lim_{\d\to0, R\to\infty}
\sg_{1}(b_R\, \iop (T_\d) \, a_R) 
= \frac{4}{3\pi} \int_{\R^3}\, H(x)^{\frac{1}{2}}   \, dx.
\end{align*}
Due to \eqref{eq:predre} it suffices to check that  
\begin{align*}
\lim_{\d\to0, R\to\infty} 
\int_{\R^3}\, H_\d(x; a_R, b_R)^{\frac{1}{2}}   \, dx
= \int_{\R^3}\, H(x)^{\frac{1}{2}}   \, dx.
\end{align*}
This convergence is verified in the same way as in the proof of Theorem 
\ref{thm:recast}\eqref{item:repsi}, which completes the proof of 
 Theorem 
\ref{thm:recast}\eqref{item:ret}.
\end{proof}

This completes the proof of Theorem \ref{thm:recast}, and hence that of 
Theorem \ref{thm:maingk}.   
   
\section{Proof of Theorem \ref{thm:recastasym}} \label{sect:prfrecastasym}

From now on we assume that the function $\psi$ is totally antisymmetric, i.e. 
\eqref{eq:antisym} holds.

We use the regularity results established in Section 
\ref{sect:regasym}. 
For all $\bx\in\SU_k(\d)$ (see \eqref{eq:wd}) and each $j = 1, 2, \dots, N-1$ we have
\begin{align*}
\bigg|x_j - \frac{x_k+x}{2}\bigg|> |x_j - x_k| - \frac{1}{2}|x_k-x|> 4\d - \frac{\d}{2}> 3\d.
\end{align*}
Therefore $\SU_k(\d)\subset \SW_k(3\d, \d)$, and we can use the representations 
\eqref{eq:triplet} and \eqref{eq:triplet1}.

\subsection{Proof of Theorem \ref{thm:recastasym}\eqref{item:repsiasy}}
In line with the representation 
\eqref{eq:bpsias}, we use the results of Sect. \ref{sect:factor} with the scalar kernel 
\begin{align*}
T(\hat\bx, x) = \sqrt{N}\,\tilde\psi(\hat\bx, x) = \sqrt{N}\,\psi(\hat\bx, x)\, e^{-F_1(x)},
\end{align*}
(see \eqref{eq:f1} and \eqref{eq:tildepsi}) which satisfies \eqref{eq:imit} with $\a = 2$ due to \eqref{eq:FSas}. 
Let the truncated kernel $T_\d(\hat\bx, x)$ 
be as defined in \eqref{eq:sdelta}:
\begin{align*}
T_\d(\hat\bx, x) = \sqrt{N}\,\eta_\d(\hat\bx) \psi(\hat\bx, x)e^{-F_1(x)}\,\sum_{k=1}^{N-1}  \t(|x-x_k|\d^{-1}).
\end{align*} 
To prove Theorem \ref{thm:recastasym} \eqref{item:repsiasy} we find the spectral asymptotics of the 
operator $ b \iop(T_\d) a$ with compactly supported $a$, $b$ and with the cut-off 
$\eta_\d$ defined in \eqref{eq:ydel}. Define 
\begin{align}\label{eq:edelta2}
E_\d(x, \om; a, b) =  2  
\bigg(\frac{|a(x)|}{24}\bigg)^2 \,  |b(x)|^2\, |\eta_\d(x)|^2\,  
\big|\sv(x, x) \cdot\om\big|^2,\quad \textup{if}\quad N=2,
\end{align}
and 
\begin{align}\label{eq:edelta}
E_\d(x, \om; a, b) =  
N (N-1)  
\bigg(\frac{|a(x)|}{24}\bigg)^2 \, 
\int_{\R^{3N-6}}\, &\ |b(\tilde\bx, x)|^2\, |\eta_\d(\tilde\bx, x)|^2\, \notag\\
&\ \qquad \times 
\big|\sv(\tilde\bx, x, x) \cdot\om\big|^2\, d\tilde\bx,
\quad \textup{if}\quad N\ge 3,
\end{align}
where $\sv(\hat\bx, x) = \nabla_x\psi(\hat\bx, x)$ and $\tilde\bx = \tilde\bx_{N-1, N}$.

\begin{lem}\label{lem:predrasy}
Let $a\in\plainC{}_0(\R^3)$ and let 
$b\in\plainC{}_0(\R^{3N-3})$ be a totally symmetric function. 
Then for each $\d\in (0, 1/2)$ we have  
\begin{align}\label{eq:predrasy}
\SfG_{3/5} (b\, \iop(T_\d)\, a) 
= \sg_{3/5} (b \,\iop(T_\d)\, a) 
= (3\pi^6)^{-\frac{2}{5}} \, \int_{\R^3}\, \int_{\mathbb S^2}
E_{\d}(x, \om; a e^{-F_1}, b)^{\frac{3}{10}}\,d\om dx,
\end{align}
where $E_\d$ is defined in \eqref{eq:edelta2} or \eqref{eq:edelta}.
\end{lem}

\begin{proof}
Throughout the proof we assume that $N\ge 3$. The case $N=2$ requires only minor adjustments. 

Due to \eqref{eq:cutoff} and to the inclusion $\SU_k(\d)\subset \SW_k(3\d, \d)$ we can use the 
representation \eqref{eq:triplet} to obtain
\begin{align*}
 T_\d(\hat\bx, x) =  \sum_{k=1}^{N-1} S_{k}(\hat\bx, x) + \sum_{k=1}^{N-1} \Om_{k}(\hat\bx, x), 
\end{align*}
where
\begin{align*}
S_k(\hat\bx, x) = &\ \eta_\d(\hat\bx) \,
\SQ(x-x_k)
\cdot\bbeta_k(\hat\bx, x), \quad \SQ(x) =  \nabla|x|^3 = 3|x| x, \\ 
\bbeta_k(\hat\bx, x) =  &\ 
\frac{\sqrt{N}}{24}   \, \sv\bigg(\tilde\bx_{k, N}, \frac{x+x_k}{2}, \frac{x+x_k}{2}\bigg)
e^{-F_1(x)}
\, \t(|x-x_k|\d^{-1}),
\end{align*}
and
\begin{align*}
\Om_k(\hat\bx, x) = 
\sqrt{N}\,\eta_\d(\hat\bx) L_k(\hat{\bx}, x) e^{-F_1(x)}\, 
\t(|x-x_k|\d^{-1}).
\end{align*}
Since $e^{-F_1(x)}$ is smooth on $\SU_k(\d)$, 
by \eqref{eq:lsmooth}, each kernel $\Om_{k}$ satisfies 
\eqref{eq:s} with $A(\bx) = 1$ and $\a = 2+\mu$ with an arbitrary $\mu\in (0, 1)$. 
Since both weights $a$ and $b$ are continuous and compactly supported, 
by Corollary \ref{cor:homo}, 
the operator $\iop(b\, \Om_{k}\,a)$ belongs to 
$\BS_{q, \infty}$, $q = 3/(5+\mu)<3/5$, and hence  
$\SfG_{3/5}(b \iop(\Om_{k}) a) = 0$. Thanks to Corollary \ref{cor:zero}, 
\begin{align*}
\SfG_{3/5}(b\, \iop(T_\d)\, a ) = &\ \SfG_{3/5}(b\, \iop(S)\, a),\quad  
\sg_{3/5}(b\, \iop(T_\d)\, a ) = \sg_{3/5}(b\, \iop(S)\, a),\\
 \textup{where} \quad S(\hat\bx, x) = &\ 
  \sum_{k=1}^{N-1} S_{k}(\hat\bx, x).
\end{align*}
To find the asymptotics for $b\iop S\, a$ we use Example \ref{ex:scaltriplet} with $M=1$, $\a=2$, 
$\SQ(x) = \nabla |x|^3$, which 
gives the asymptotics of the form \eqref{eq:model} with $p = 3/5$, with the function 
$Y(t, \xi)$ given by \eqref{eq:Yvec}:
\begin{align*}
Y(t, \xi) = &\ \frac{|\nu_{3, 3}|^2}{|\xi|^{12}} \, 
 \,\sum_{k=1}^{N-1}\int_{\R^{3N-6}}\, |b(\tilde\bx_{k, N}, t)\, 
 \eta_\d(\tilde\bx_{k, N}, t)|^2 \,
|\xi\cdot \boldsymbol\b_{k}(\tilde\bx_{k, N}, t, t)|^2 
\,d\tilde\bx_{k, N}\\
= &\ \frac{|\nu_{3, 3}|^2}{|\xi|^{12}} \, \frac{N}{24^2}\, e^{-2F_1(t)}
 \,\sum_{k=1}^{N-1}\int_{\R^{3N-6}}\, |b(\tilde\bx_{k, N}, t)\, 
 \eta_\d(\tilde\bx_{k, N}, t)|^2 \,
|\xi\cdot \sv(\tilde\bx_{k, N}, t, t)|^2 
\,d\tilde\bx_{k, N}.
\end{align*}
Recall that by definition \eqref{eq:ydel} the function $\eta_\d$ is  totally 
symmetric, and by our assumption, so is $b$. Furthermore, 
due to \eqref{eq:antisym}, we have 
\begin{align*}
\sv(x_1, \dots,  x_{k-1}, t, x_{k+1},  \dots, x_{N-1}, x) 
=  (-1)^{k+N-1} \sv (x_1, \dots, x_{k-1}, x_{N-1}, x_{k+1},\dots, x_{N-2}, t, x),
\end{align*} 
and hence the expression for $Y(t, \xi)$ simplifies:
\begin{align*}
Y(t, \xi) = &\ \frac{|\nu_{3, 3}|^2}{24^2\,|\xi|^{12}} \, N(N-1) e^{-2F_1(t)}\\
&\  \times \, \int_{\R^{3N-6}}\, |b(\tilde\bx, t)\, 
 \eta_\d(\tilde\bx, t)|^2 \,
 |\xi\cdot \sv(\tilde\bx, t, t)|^2 
\,d\tilde\bx,\, \tilde \bx= \tilde\bx_{N-1, N}.
\end{align*}
Thus, it follows from \eqref{eq:model} that 
 \begin{align*}
\SfG_{3/5}(b\, \iop(S)\, a) = \sg_{3/5}(b\, \iop(S)\, a)
= \frac{|\nu_{3,3}|^{\frac{3}{5}}}{24\pi^3}\,\int_{\R^3}\,  \int_{\mathbb S^2}\, 
E_\d(x, \om; a e^{-F_1}, b)^{\frac{3}{10}}\, d\om dx,
\end{align*}
where $E_\d(x, \om; a, b)$ is as in \eqref{eq:edelta}. A simple calculation shows that 
$\frac{|\nu_{3,3}|^{\frac{3}{5}}}{24\pi^3} = (3\pi^6)^{-\frac{2}{5}}$. 
Now 
\eqref{eq:predrasy} follows.
\end{proof}

\begin{proof}[Proof of Theorem \ref{thm:recastasym}\eqref{item:repsiasy}]
Let $a(x) = a_R(x) = \t(|x| R^{-1})$, $b(\hat\bx) = b_R(\hat\bx) = \t(|\hat\bx| R^{-1})$. 
Recall that 
\begin{align*}
\SPsi(\bx) = T(\hat\bx, x) e^{F_1(x)}.
\end{align*}
Thus, according to Lemmata \ref{lem:tocompact} and \ref{lem:separate}, 
to prove Theorem \ref{thm:recastasym}\eqref{item:repsiasy} it suffices to show that 
\begin{align*}
\lim_{\d\to0, R\to\infty}\SfG_{3/5}(b_R\, \iop (T_\d) \, e^{F_1} a_R) 
= &\ \lim_{\d\to0, R\to\infty}
\sg_{3/5}(b_R\, \iop (T_\d) \,  e^{F_1}  a_R) \\
= &\ (3\pi^6)^{-\frac{2}{5}}\, \int_{\R^3}\, \int_{\mathbb S^2}\, E(x, \om)^{\frac{3}{10}} \, d\om dx.
\end{align*}
Due to \eqref{eq:predrasy} it suffices to check that  
\begin{align}\label{eq:h310}
\lim_{\d\to0, R\to\infty} 
\int_{\R^3}\,\int_{\mathbb S^2}\, E_\d(x, \om; a_R, b_R)^{\frac{3}{10}}\, dx
= \int_{\R^3}\, \int_{\mathbb S^2}\, E(x, \om)^{\frac{3}{10}}\, d\om dx.
\end{align}
According to Lemma \ref{lem:coeffs}, 
the function $E(x, \om)^{3/10}$ is integrable. Therefore, the function 
$\sv(\tilde\bx, x, x)\cdot\om$ is 
square-integrable in $\tilde\bx$ for a.e. $x\in\R^3$, and a.e. $\om\in\mathbb S^2$.
Furthermore, 
$|\eta_\d|\le 1$, $|a_R|, |b_R|\le 1$ and $\eta_\d\to 0, \d\to 0,$ a.e., and $a_R\to 1, b_R\to 1$, $R\to\infty$,  
a.e. Thus, by the Dominated Convergence Theorem, 
$E_\d(x, \om; a_R, b_R)\to E(x, \om)$ a.e. $x\in\R^3$, and a.e. $\om\in\mathbb S^2$. Since 
 $E_\d(x, \om; a_R, b_R)\le E(x, \om)$, the Dominated Convergence Theorem ensures that 
\eqref{eq:h310} holds. Theorem  
\ref{thm:recastasym}\eqref{item:repsiasy} is proved. 
\end{proof}

\subsection{Proof of Theorem \ref{thm:recastasym}\eqref{item:retasy}}
In line with \eqref{eq:bpsias}, we use the results of Sect. \ref{sect:factor} with the vector-valued kernel 
\begin{align*}
T(\hat\bx, x) = \sqrt{N}\, \nabla_x\big( e^{- F_1(x)} \psi(\hat\bx, x)\big)
\end{align*}
(see \eqref{eq:f1} for the definition of $F_1$), which satisfies \eqref{eq:imit} with $\a = 1$ due to \eqref{eq:FSas}. Let the truncated kernel $T_\d(\hat\bx, x)$ be as defined in \eqref{eq:sdelta}:
\begin{align*}
T_\d(\hat\bx, x) = \sqrt{N}\,\eta_\d(\hat\bx)\, \nabla_x\big( e^{- F_1(x)} \psi(\hat\bx, x)\big)\,
\sum_{k=1}^{N-1}  \t(|x-x_k|\d^{-1}).
\end{align*} 
To prove Theorem \ref{thm:recastasym} \eqref{item:retasy} we find the spectral asymptotics of the 
operator $ b \iop(T_\d) a$ with compactly supported $a$, $b$ and with the cut-off 
$\eta_\d$ defined in \eqref{eq:ydel}. Let $E_\d(x, \om; a, b)$ be as defined in 
\eqref{eq:edelta}.

\begin{lem}
Let $a\in\plainC{}_0(\R^3)$ and let 
$b\in\plainC{}_0(\R^{3N-3})$ be a totally symmetric function. 
Then for each $\d\in (0, 1/2)$ we have  
\begin{align}\label{eq:predretasy}
\SfG_{3/4} (b\, \iop(T_\d)\, a) 
= \sg_{3/4} (b \,\iop(T_\d)\, a)  = \bigg(\frac{8}{3\pi^9}\bigg)^{\frac{1}{4}}\, 
\int_{\R^3}\, \int_{\mathbb S^2}
E_{\d}(x, \om; a e^{-F_1}, b)^{\frac{3}{8}}\,d\om dx,
\end{align}
where $E_\d$ is defined in \eqref{eq:edelta2} or \eqref{eq:edelta}.
\end{lem}

\begin{proof} 
We conduct the proof only for $N\ge 3$. The case $N=2$ requires only minor adjustments. 

Rewrite the kernel $T_\d$:
\begin{align*}
T_\d(\hat\bx, x) = &\ \sqrt{N}\,\eta_\d(\hat\bx)\, e^{- F_1(x)} \sv(\hat\bx, x)\,
\sum_{k=1}^{N-1}  \t(|x-x_k|\d^{-1})\\
  &\ \qquad  -\sqrt{N}\,\eta_\d(\hat\bx)\,  \nabla_x F_1(x) e^{-F_1(x)}\,\psi(\hat\bx, x)\,
\sum_{k=1}^{N-1}  \t(|x-x_k|\d^{-1}).
\end{align*}
where, as usual, $\sv(\hat\bx, x) = \nabla_x\psi(\hat\bx, x)$.

Due to \eqref{eq:cutoff} and to the inclusion $\SU_k(\d)\subset \SW_k(3\d, \d)$ we can use the 
representation \eqref{eq:triplet1} to obtain
\begin{align*}
 T_\d(\hat\bx, x) =  \sum_{k=1}^{N-1} S_{k}(\hat\bx, x) + \sum_{k=1}^{N-1} \Om_{k}(\hat\bx, x), 
\end{align*}
where
\begin{align*}
S_k(\hat\bx, x) = &\ \eta_\d(\hat\bx) \, \SQ(x-x_k)\bbeta_k(\hat\bx, x), \quad 
\SQ(x) = \nabla_x\nabla_x^T\, |x|^3,
\\
\bbeta_k(\hat\bx, x) =  &\ 
\frac{\sqrt{N}}{24}   
 \, \sv\bigg(\tilde\bx_{k, N}, \frac{x+y}{2}, \frac{x+y}{2}\bigg) \, e^{-F_1(x)}
\, \t(|x-x_k|\d^{-1}),
\end{align*}
and
\begin{align*}
\Om_k(\hat\bx, x) = &\ 
\Om_k^{(1)}(\hat\bx, x) + \Om_k^{(2)}(\hat\bx, x),\\
\Om_k^{(1)}(\hat\bx, x) = &\ 
\sqrt{N}\,\eta_\d(\hat\bx) L_k^{(1)}(\hat{\bx}, x) \, e^{-F_1(x)}\, 
\t(|x-x_k|\d^{-1}),\\
\Om_k^{(2)}(\hat\bx, x) = &\   
-\sqrt{N}\,\eta_\d(\hat\bx)\,  \nabla_x F_1(x) e^{-F_1(x)}\,\psi(\hat\bx, x)\,
 \t(|x-x_k|\d^{-1}). 
\end{align*}
Since $e^{-F_1(x)}$ is smooth on $\SU_k(\d)$, by \eqref{eq:l1smooth}, 
each kernel $\Om_{k}^{(1)}$ satisfies 
\eqref{eq:s} with $A(\bx) = 1$ and $\a = 1+\mu$ with an arbitrary $\mu\in (0, 1)$. 
Since both weights $a$ and $b$ are continuous and compactly supported, 
by Corollary \ref{cor:homo}, 
the operator $\iop\big(b\, \Om_{k}^{(1)}\,a\big)$ belongs to 
$\BS_{q, \infty}$, $q = 3/(4+\mu)<3/4$, and hence  
$\SfG_{3/4}(b \iop\big(\Om_{k}^{(1)}) a\big) = 0$. 
The kernel $\Om_k^{(2)}$ has the same structure as the kernel 
considered in Lemma \ref{lem:predrasy} (with the extra weight $\nabla_x F_1(x)$). 
Due to this Lemma, $\iop\big( b \Om_k^{(2)} a\big)\in \BS_{3/5, \infty}$, and hence 
$\SfG_{3/4}(b \iop\big(\Om_{k}^{(2)}) a\big) = 0$. By \eqref{eq:trianglep}, 
$\SfG_{3/4}(b \iop\big(\Om_{k} a\big) = 0$. Thanks to Corollary \ref{cor:zero}, 
\begin{align*}
\SfG_{3/4}(b\, \iop(T_\d)\, a ) = &\ \SfG_{3/4}(b\, \iop(S)\, a),\quad  
\sg_{3/4}(b\, \iop(T_\d)\, a ) = \sg_{3/4}(b\, \iop(S)\, a),\\
 \textup{where} \quad S(\hat\bx, x) = &\ 
  \sum_{k=1}^{N-1} S_{k}(\hat\bx, x).
\end{align*}
To find the asymptotics for $b\iop S\, a$ we use Example \ref{ex:matriplet} with $M=1$, $\a=1$, 
$\SQ(x) = \nabla \nabla^T |x|^3$, which 
gives the asymptotics of the form \eqref{eq:model} with $p = 3/4$, 
with the function $Y(t, \xi)$ given by \eqref{eq:Ymat}:
\begin{align*}
Y(t, \xi) = &\ \frac{|\nu_{3, 3}|^2}{|\xi|^{10}} \, 
 \,\sum_{k=1}^{N-1}\int_{\R^{3N-6}}\, |b(\tilde\bx_{k, N}, t)\, 
 \eta_\d(\tilde\bx_{k, N}, t)|^2 \,
|\xi\cdot \boldsymbol\b_{k}(\tilde\bx_{k, N}, t, t)|^2 
\,d\tilde\bx_{k, N}\\
= &\ \frac{|\nu_{3, 3}|^2}{|\xi|^{10}} \, \frac{N}{24^2}\, e^{-2F_1(t)}\, 
 \,\sum_{k=1}^{N-1}\int_{\R^{3N-6}}\, |b(\tilde\bx_{k, N}, t)\, 
 \eta_\d(\tilde\bx_{k, N}, t)|^2 \,
|\xi\cdot \sv(\tilde\bx_{k, N}, t, t)|^2 
\,d\tilde\bx_{k, N}.
\end{align*}
As in the proof of Lemma \ref{lem:predrasy}, using the total symmetry of $\eta_\d b$ and antisymmetry of $\sv$, we can simplify the formula for $Y(t, \xi)$:
\begin{align*}
Y(t, \xi) = &\ \frac{|\nu_{3, 3}|^2}{24^2\,|\xi|^{10}} \, N(N-1) \\
&\  \times \, e^{-2F_1(t)}\,\int_{\R^{3N-6}}\, |b(\tilde\bx, t)\, 
 \eta_\d(\tilde\bx, t)|^2 \,
 |\xi\cdot \sv(\tilde\bx, t, t)|^2 
\,d\tilde\bx,\, \tilde \bx= \tilde\bx_{N-1, N}.
\end{align*}
Thus, it follows from \eqref{eq:model} that 
 \begin{align*}
\SfG_{3/4}(b\, \iop(S)\, a) = \sg_{3/4}(b\, \iop(S)\, a)
= \frac{|\nu_{3,3}|^{\frac{3}{4}}}{24\pi^3}\,\int_{\R^3}\,  \int_{\mathbb S^2}\, 
E_\d(x, \om; a e^{-F_1}, b)^{\frac{3}{8}}\, d\om dx,
\end{align*}
where $E_\d(x, \om; a, b)$ is as in \eqref{eq:edelta}. 
A simple calculation shows that 
$\frac{|\nu_{3,3}|^{\frac{3}{4}}}{24\pi^3} = (\frac{8}{3\pi^9})^{\frac{1}{4}}$. 
Now 
\eqref{eq:predretasy} follows.
\end{proof}

\begin{proof}[Proof of Theorem \ref{thm:recastasym}\eqref{item:retasy}] 
Rewrite the kernel $\SV$:
\begin{align*}
\SV(\hat\bx, x) = &\ \sqrt{N}\, \nabla_x\psi(\hat\bx, x)
= \sqrt{N}\, e^{F_1(x)} \nabla_x \big( e^{-F_1(x)} \psi(\hat\bx, x)\big) 
+ \sqrt{N}\, \big(\nabla_x F_1(x)\big)\, \psi(\hat\bx, x)\\
= &\ e^{F_1(x)} T(\hat\bx, x) + \sqrt{N}\, \big(\nabla_x F_1(x)\big)\, \psi(\hat\bx, x).
\end{align*}
By Theorem \ref{thm:recastasym}(\ref{item:repsiasy}), which was proved earlier, 
the operator with the second kernel in this sum belogs to $\BS_{3/5, \infty}$, and hence, by 
Corollary \ref{cor:zero}, 
 \begin{align*}
 \SfG_{3/4}(\iop (\SV)) =  \SfG_{3/4}(\iop (T) e^{F_1}),\quad  
  \sg_{3/4}(\iop (\SV)) =  \sg_{3/4}(\iop (T) e^{F_1}).
 \end{align*}
Let $a(x) = a_R(x) = \t(|x| R^{-1})$, $b(\hat\bx) = b_R(\hat\bx) = \t(|\hat\bx| R^{-1})$. 
By Lemmata \ref{lem:tocompact} and \ref{lem:separate}, 
to prove Theorem \ref{thm:recastasym}\eqref{item:retasy} it suffices to show that 
\begin{align*}
\lim_{\d\to0, R\to\infty}\,\SfG_{3/4}(b_R\, \iop (T_\d) e^{F_1}\, a_R) 
= &\ \lim_{\d\to0, R\to\infty}
\sg_{3/4}(b_R\, \iop (T_\d) e^{F_1} \, a_R) \\
= &\ \bigg(\frac{8}{3\pi^9}\bigg)^{\frac{1}{4}}\, 
\int_{\R^3}\, \int_{\mathbb S^2}\, E(x, \om)^{\frac{3}{8}} \, d\om dx.
\end{align*}
Due to \eqref{eq:predrasy} it suffices to check that  
\begin{align*}
\lim_{\d\to0, R\to\infty} 
\int_{\R^3}\,\int_{\mathbb S^2}\, E_\d(x, \om; a_R, b_R)^{\frac{3}{8}}\, dx
= \int_{\R^3}\, \int_{\mathbb S^2}\, E(x, \om)^{\frac{3}{8}}\, d\om dx.
\end{align*}
This convergence is verified in the same way as in the proof of Theorem 
\ref{thm:recastasym}\eqref{item:repsiasy}, which completes the proof of 
 Theorem 
\ref{thm:recastasym}\eqref{item:retasy}.
\end{proof}

This completes the proof of Theorem \ref{thm:recastasym}, 
and hence that of Theorem \ref{thm:maingkasym}.

\section{Appendix}\label{sect:app}

In this appendix we prove an elementary extension result for Sobolev spaces. 
We consider spaces of functions that depend either 
on one variable $x\in\R^d$ or on two variables $(t, x)\in \R^l\times\R^d$. 
Denote $K = \{t\in\R^l: |t| < 1\}$,\ $B = \{x\in\R^d: |x| < 1\}$ and 
$B_0 = B\setminus\{0\}$. In Section \ref{sect:intop} (see Remark \ref{rem:remove}) 
we need only the case of one variable $x\in \R^d$. The following lemma however holds for 
a more general case of two variables $(t, x)\in \R^l\times \R^d$ as well, 
which may come in handy in different 
circumstances.

Our objective is to establish the following elementary extension result. 

\begin{lem} \label{lem:remove}
Let the dimension $l$ be arbitrary, let 
$d\ge 2$, $m\ge 1$, and 
$p\in [d(d-1)^{-1}, \infty]$. Then 
$\plainW{m, p}(B_0) 
= \plainW{m, p}(B)$ and $\plainW{m, p}\big(K\times B_0\big) 
= \plainW{m, p}\big(K\times B\big)$.
\end{lem}  
  
\begin{proof} 
We prove only the second equality since the first one is obtained in the same way. 
It suffices to show that $\plainW{1, p}\big(K\times B_0\big) 
\subset \plainW{1, p}\big(K\times B\big)$. 
Let $u\in \plainW{1, p}\big(K\times B_0\big)$. 
We show that for all $j\in \mathbb N_0^{d+l}$, $|j|=1$,  and all test 
functions $\phi\in \plainC\infty_0(K\times B)$ the following identity holds:
\begin{align}\label{eq:parts}
\int u\, \p^j\phi\, dt dx = -\int \p^j u \, \phi \, dt dx,  
\end{align}
where by $\p^j u \in \plainL{p}\big(K\times B_0\big)=\plainL{p}\big(K\times B\big)$ we denote the distributional  
derivatives of $u$ on the set $K\times B_0$. 
Let $\t$ be the cut-off 
function as defined in \eqref{eq:sco} and \eqref{eq:sco1}, and introduce 
$\t_\varepsilon(x) = \t(|x|\varepsilon^{-1})$. Then 
\begin{align*}
\int u\, \p^j\phi\, dt dx = \int u\, \p^j\big(\phi\,(1-\t_\varepsilon)\big)\, dt dx 
+ \int u\, \p^j\big(\phi\,\t_\varepsilon\big)\, dt dx. 
\end{align*}
Since $1-\t_\varepsilon(x) = 0$ for $|x|<\varepsilon/2$, in the first integral 
on the right-hand side we can integrate by parts:
\begin{align}\label{eq:backpar}
\int u\, \p^j\phi\, dtdx = 
&\ -\int (\p^j u)\, \phi\,(1-\t_\varepsilon)\, dt dx  
+ \int u\, \p^j\big(\phi\,\t_\varepsilon\big)\, dt dx\notag\\ 
= &\ -\int (\p^j u)\, \phi\, dt dx
+ \int (\p^j u)\, \phi\,\t_\varepsilon\, dt dx 
+ \int u\, \p^j\big(\phi\,\t_\varepsilon\big)\, dt dx.
\end{align}
The rest of the proof is conducted for $p <\infty$ only. 
The modifications for the case $p = \infty$ are obvious. 
By H\"older's inequality, the second term on the right-hand side of 
\eqref{eq:backpar} can be estimated as follows, with $q^{-1} = 1- p^{-1}$:
\begin{align}\label{eq:1} 
\bigg|\int (\p^j u)\, \phi\,\t_\varepsilon\, dt dx\bigg|
\le \bigg[\int\limits_{|x|< \varepsilon}|\p^j u|^p\, dt dx\bigg]^{\frac{1}{p}} \, 
\bigg[\int\limits_{|x|< \varepsilon} |\phi|^q \, dtdx\bigg]^{\frac{1}{q}} \rightarrow 0
\quad \textup{as}\quad \varepsilon\to 0.
\end{align}
Consider the third integral on the right-hand side of \eqref{eq:backpar}. 
By H\"older's inequality 
again, 
\begin{align*}
\bigg|\int u\, \p^j\big(\phi\,\t_\varepsilon\big)\, dtdx\bigg|
\le \bigg[\int\limits_{|x|< \varepsilon}|u|^p\, dt dx\bigg]^{\frac{1}{p}} \,
\max|\p^j (\phi\,\t_\varepsilon)|\, 
\bigg[\int\limits_{|x|< \varepsilon}\, dt dx\bigg]^{\frac{1}{q}}. 
\end{align*} 
As $\max|\p^j (\phi\t_\varepsilon)|\lesssim \varepsilon^{-1}$, and $q = (1-p^{-1})^{-1} \le d$, 
the right-hand side does not exceed 
\begin{align*}
\varepsilon^{\frac{d}{q} - 1}\bigg[\int\limits_{|x|< \varepsilon}|u|^p\, dtdx\bigg]^{\frac{1}{p}}
\to 0,\quad \textup{as}\quad \varepsilon\to 0.
\end{align*}
Together with \eqref{eq:1} this implies that the last two terms on the right-hand side 
of \eqref{eq:backpar} tend to zero as $\varepsilon\to 0$. This entails 
\eqref{eq:parts}, as required. This completes the proof of the lemma.
\end{proof}      
      
\vskip 0.5cm

  \textbf{Acknowledgments.} 
The second author (A.V.S.) is grateful to J.~Cioslowski    
for comments and stimulating discussions, and to M.~Lewin for pointing out the asymptotic 
problem for the kinetic energy density operator. 
Thanks also go to A.~Nazarov and A.~Tyulenev for bringing to 
the authors attention the book 
\cite{Burenkov1998}, 
and to D.~Edmunds, V.~Kozlov, V.~Maz'ya, G.~Rozenblum, D.~Vassiliev 
for their advice on Sobolev spaces.  The authors are grateful to P.~Hearnshaw for useful suggestions.

\end{document}